\documentclass[11pt]{article}

\usepackage[letterpaper,margin=1in]{geometry}
\usepackage[T1]{fontenc}
\usepackage{lmodern}
\usepackage{microtype}
\usepackage{amsmath,amssymb,amsthm,mathtools}
\usepackage{booktabs}
\usepackage{placeins}
\usepackage{array}
\usepackage{tabularx}
\usepackage{enumitem}
\usepackage{cite}
\usepackage{xcolor}
\usepackage{pgfplots}
\pgfplotsset{compat=1.18}

\definecolor{darkblue}{RGB}{40,100,160}
\usepackage[
    colorlinks=true,
    linkcolor=darkblue,
    citecolor=darkblue,
    urlcolor=darkblue
]{hyperref}

\hypersetup{
 pdftitle={The Honeycomb Framework for Code Bounds},
 pdfauthor={William Gay and Fernando Granha Jeronimo and Lenny Liu},
 pdfsubject={Representation-theoretic upper bounds and finite hierarchies for binary and q-ary codes},
 pdfkeywords={binary codes, q-ary codes, honeycomb hierarchy, Horn inequalities, classical--quantum channels, MRRW, hyperoctahedral group, Littlewood--Richardson hives, semidefinite programming, moment hierarchy}
}
\allowdisplaybreaks
\newtheorem{theorem}{Theorem}[section]
\newtheorem{proposition}[theorem]{Proposition}
\newtheorem{lemma}[theorem]{Lemma}
\newtheorem{corollary}[theorem]{Corollary}
\newtheorem{definition}[theorem]{Definition}
\newtheorem{remark}[theorem]{Remark}

\newcommand{\C}{\mathbb C}
\newcommand{\R}{\mathbb R}
\newcommand{\Z}{\mathbb Z}

\newcommand{\cC}{\mathcal C}
\newcommand{\cD}{\mathcal D}
\newcommand{\cE}{\mathcal E}

\newcommand{\cI}{\mathcal I}
\newcommand{\cK}{\mathcal K}

\newcommand{\cV}{\mathcal V}

\newcommand{\Id}{\mathrm I}
\newcommand{\one}{\mathbf 1}
\newcommand{\Atwo}{A_2}
\newcommand{\Htwo}{H_2}
\newcommand{\ip}[2]{\left\langle #1,#2\right\rangle}
\newcommand{\norm}[1]{\left\lVert #1\right\rVert}
\newcommand{\tr}{\operatorname{tr}}
\newcommand{\rank}{\operatorname{rank}}
\newcommand{\Res}{\operatorname{Res}}
\newcommand{\Ind}{\operatorname{Ind}}
\newcommand{\Hom}{\operatorname{Hom}}
\newcommand{\Sym}{\operatorname{Sym}}
\newcommand{\diag}{\operatorname{diag}}
\newcommand{\defeq}{\mathrel{:=}}
\newcommand{\eps}{\varepsilon}

\newcommand{\Jmap}{\mathsf J}
\newcommand{\RMQC}{R_{\mathrm{MQC}}}
\newcommand{\RtwoMQC}{R_{\mathrm{2MQC}}}
\newcommand{\HC}{\mathsf{HC}}
\newcommand{\Mom}{\mathsf{Mom}}
\newcommand{\Par}{\mathsf{Par}}

\title{The Honeycomb Framework for Code Bounds}

\author{William Gay\thanks{Department of Computer Science, University of Illinois at Urbana-Champaign. Email: \texttt{whgay2@illinois.edu}}  \and Fernando Granha Jeronimo\thanks{Department of Computer Science, University of Illinois at Urbana-Champaign. Email: \texttt{granha@illinois.edu}}  \and Lenny Liu\thanks{Department of Computer Science, University of Illinois at Urbana-Champaign. Email: \texttt{hengyu2@illinois.edu}}}
\date{}

\begin{document}
\maketitle

\begin{abstract}
For \(0<\delta<1/2\), let \(R_2(\delta)\) be the optimal asymptotic rate of
binary codes of relative distance \(\delta\).  We introduce the
\emph{honeycomb hierarchy}, a representation-theoretic framework that gives
new asymptotic upper bounds on \(R_2(\delta)\) and, at higher levels, a
finitely exact family of finite-length relaxations.  Its first analytic level
is the complete two-row hyperoctahedral representation graph associated with
a moving stabilizer type \(S^{(n-k,k)}\).  Retaining every compatible two-row
irreducible and every coordinate box-transfer channel, together with a
profile-optimized moving-projection theorem, yields an explicit
four-parameter exponent \(\kappa_{\mathrm{HC}}\).  The earlier whole-cube
exponent \(\kappa_H\) is a boundary restriction of this optimization, whereas
the fully optimized second MRRW exponent \(M_2\) is an exact symmetric slice.
A one-sided face and an entropy-balanced two-sided branch both strictly
improve \(\kappa_H\) for every nontrivial relative distance.

The relevant prior moving-projection benchmark is the combined exponent
\(\kappa_{\mathrm{bin}}=\min\{\kappa_{\mathrm{CW}},\kappa_H\}\), which uses a
separate constant-weight branch \(\kappa_{\mathrm{CW}}\).  Replacing only the
whole-cube branch by the complete honeycomb bound gives
\(\kappa_{\mathrm{best}}=\min\{\kappa_{\mathrm{CW}},
\kappa_{\mathrm{HC}}\}\).  We prove, pointwise on \(0<\delta<1/2\),
\begin{gather*}
 R_2(\delta)\le \kappa_{\mathrm{best}}(\delta)
 \le \kappa_{\mathrm{bin}}(\delta)
 \le R_{\mathrm{2MQC}}(\delta)<M_2(\delta),\\[-1mm]
 \kappa_{\mathrm{best}}(\delta)
 \le \min\{\kappa_{\mathrm{CW}}(\delta),
             \kappa_{\mathrm{bal}}(\delta)\}
 <R_{\mathrm{2MQC}}(\delta),
 \qquad
 \kappa_H(\delta)=R_{\mathrm{MQC}}(\delta).
\end{gather*}
Thus the first honeycomb level is never weaker than the previous combined
OpenAI bound, is strictly stronger whenever its new branch is active, and
strictly improves 2MQC throughout the open distance interval. 

The hierarchy has two further finite-certificate directions.  Increasing the
representation depth replaces scalar transfers by matrix-valued transfers on
Littlewood--Richardson hive fibers; a frame-profile theorem makes every such
finite row level rigorous.  Increasing the anchor depth localizes these
quadratic certificates in a stable-set moment hierarchy.  The resulting
bounds are monotone in both directions and recover \(A_2(n,d)\) exactly at
sufficiently high anchor order.  A complementary Horn--channel hierarchy
gives nested explicit matrix optimizations whose \(2\times2\) level is exactly
\(\kappa_{\mathrm{HC}}\) and whose \(3\times3\) level is a further
unconditional higher-row bound. Already at low levels, these constructions improve the strongest previous general asymptotic bounds, while the honeycomb perspective provides a systematic route toward stronger code bounds at higher levels.
\end{abstract}

\clearpage
\tableofcontents
\clearpage

\section{Introduction}
\label{sec:statement}

For \(0<\delta<1/2\), let \(R_2(\delta)\) denote the supremal asymptotic rate
of a binary code with relative minimum distance at least \(\delta\).
Determining this rate--distance tradeoff is one of the central extremal
problems in coding theory.  The Gilbert--Varshamov argument gives
\(R_2(\delta)\ge 1-\Htwo(\delta)\)
\cite{Gilbert1952,Varshamov1957}, while the strongest general upper bounds
have historically come from Delsarte's association-scheme linear program and
the McEliece--Rodemich--Rumsey--Welch (MRRW) analysis
\cite{Delsarte1973,MRRW1977}.  The best upper and lower bounds still differ by
an exponential factor in the block length.  The main obstacle is not the lack
of stronger finite relaxations, but the lack of stronger relaxations whose
optima can be analyzed sharply at the exponential scale.

For almost five decades, the fully optimized second MRRW exponent
\(M_2(\delta)\) was the best-known general asymptotic upper bound.  The word
``optimized'' matters: the second bound minimizes a one-parameter family, and
a pointwise improvement of the first MRRW curve need not improve that minimum.
In 2026, OpenAI reported that an internal model had produced new
binary- and spherical-code bounds through the moving-projection
method \cite[Chapter~2]{OpenAI2026}.
That work produced a whole-cube exponent \(\kappa_H\) and a constant-weight
exponent \(\kappa_{\mathrm{CW}}\).  Their pointwise minimum
\begin{equation*}\label{eq:intro-kappa-bin-first}
 \kappa_{\mathrm{bin}}(\delta)
 \defeq\min\{\kappa_H(\delta),\kappa_{\mathrm{CW}}(\delta)\}
\end{equation*}
strictly improves the complete second MRRW bound.  Since the active branch
changes with \(\delta\), the natural predecessor of the present work is the
combined exponent \(\kappa_{\mathrm{bin}}\), rather than \(\kappa_H\) alone.

In this context, it is natural to ask the following:
\begin{center} 
\emph{Is the graph underlying \(\kappa_H\) only a boundary of a larger higher-order tractable representation graph, and can enlarging it improve code bounds?}
\end{center}

\subsection{The honeycomb hierarchy and the main results}

The organizing object of the paper is the \emph{honeycomb hierarchy}.  Its
first analytic level enlarges the whole-cube representation graph to the full
family of compatible two-row representations.  Higher representation levels
retain successively richer Littlewood--Richardson multiplicity spaces, and an
independent anchor parameter localizes the resulting quadratic certificates
inside a complete moment hierarchy.  The framework is designed to reconcile
two goals that are usually in tension: an explicit low-level asymptotic
improvement and a rigorous path to exact finite bounds.

\paragraph{The complete two-row graph.}
The Hamming cube is the homogeneous space \(B_n/S_n\), where
\(B_n=C_2^n\rtimes S_n\) is the hyperoctahedral group.  Fix the moving
stabilizer type \(S^{(n-k,k)}\).  The whole-cube construction of
\cite{OpenAI2026} retains the boundary family of bipartitions
\(((n-i),(i))\).  We retain every \(B_n\)-irreducible whose restriction
contains the same stabilizer type and whose two component partitions have at
most two rows.  The resulting vertices fill an explicit three-dimensional
compatibility region.  Coordinate multiplication transfers one box between
the two components, producing four same-row and cross-row channels in each
bulk cell.  This is the first \emph{honeycomb graph}.  At the two-row level the
name describes the coupled local geometry; at higher row levels the
multiplicity fibers are literally the Knutson--Tao hives that govern
Littlewood--Richardson coefficients \cite{KnutsonTao1999}.

\paragraph{The first honeycomb bound.}
We prove a profile-optimized finite moving-projection theorem, compute every
box-transfer coefficient exactly, and pass to the continuum through
F{\o}lner boxes.  This gives a four-variable exponent
\(\kappa_{\mathrm{HC}}\) with
\(R_2(\delta)\le\kappa_{\mathrm{HC}}(\delta)\).  The whole-cube exponent
\(\kappa_H\) is recovered by closing both ambient second rows, so it is a
boundary restriction of the honeycomb problem.  A symmetric slice is exactly
the fully optimized second MRRW problem.  Between these restrictions lie
asymmetric witnesses: a one-sided face and an entropy-balanced two-sided
branch define explicit exponents \(\kappa_{\mathrm{face}}\) and
\(\kappa_{\mathrm{bal}}\), and each is strictly smaller than \(\kappa_H\) for
every \(0<\delta<1/2\).

\paragraph{The main comparison.}
The final binary bound keeps the stronger constant-weight branch and replaces
only the whole-cube branch.  Define
\begin{equation*}\label{eq:intro-combined-definitions}
 \kappa_{\mathrm{best}}
 =\min\{\kappa_{\mathrm{CW}},\kappa_{\mathrm{HC}}\},
 \qquad
 \kappa_{\mathrm{pair}}
 =\min\{\kappa_{\mathrm{CW}},\kappa_{\mathrm{bal}}\}.
\end{equation*}
For every \(0<\delta<1/2\), we prove
\begin{equation*}\label{eq:intro-summary-comparison}
\begin{aligned}
 R_2(\delta)
 &\le \kappa_{\mathrm{best}}(\delta)
 \le \kappa_{\mathrm{bin}}(\delta)
 \le R_{\mathrm{2MQC}}(\delta)<M_2(\delta),\\
 R_2(\delta)
 &\le \kappa_{\mathrm{best}}(\delta)
 \le \kappa_{\mathrm{pair}}(\delta)
 <R_{\mathrm{2MQC}}(\delta),
 \qquad
 \kappa_H(\delta)=R_{\mathrm{MQC}}(\delta).
\end{aligned}
\end{equation*}
The first line compares the complete honeycomb combination with the strongest
preceding combined exponent; the inequality is strict whenever the honeycomb
branch is active.  The second line gives a pointwise strict improvement over
2MQC throughout the open distance interval.

The terminology ``honeycomb'' and ``hive'' comes from the
Knutson--Tao model for Littlewood--Richardson coefficients and the
Horn eigenvalue problem \cite{KnutsonTao1999}.

\paragraph{Representation depth and anchor depth.}
The complete two-row graph is the first analytic member of a larger finite
framework.  Increasing the representation depth allows more rows in the
stabilizer and ambient partitions.  Littlewood--Richardson multiplicity then
becomes nontrivial, and scalar box transfers become matrices acting on hive
fibers.  We prove a frame-profile moving-projection theorem that handles
arbitrary multiplicity and yields a rigorous finite bound at every row level.
Increasing the anchor depth localizes the associated positive quadratic
kernels in the stable-set moment hierarchy of the forbidden-distance graph.
The resulting two-parameter bounds are monotone and recover \(A_2(n,d)\)
exactly at sufficiently high anchor order.  This is a finite convergence
theorem; it does not assert that a fixed low level determines the asymptotic
rate.  The same finite construction extends to every alphabet size \(q\)
through \(C_q\wr S_n\).

\paragraph{The Horn--channel companion.} While this paper was being written,
Alrabiah and Guruswami, introduced an exciting
classical--quantum channel framework for binary-code converses \cite{AlrabiahGuruswami2026}.  Their mixed-qubit channel (MQC) strictly
improves the first MRRW curve, and their masked mixed-qubit channel (2MQC)
strictly improves the second MRRW bound.  We establish the precise relation
between these channel bounds and the moving-projection bounds: the whole-cube
exponent is exactly MQC,
\(\kappa_H=R_{\mathrm{MQC}}\), while the combined OpenAI exponent satisfies
\(\kappa_{\mathrm{bin}}\le R_{\mathrm{2MQC}}\).  Thus the strongest direct
prior benchmark is \(\min\{\kappa_H,\kappa_{\mathrm{CW}}\}\).  

We combine their quantum channel and our hierarchical higher-order perspectives. The Horn geometry predicted by the representation-depth axis also admits a
direct asymptotic realization.  For positive semidefinite matrices \(K,L\)
with \(\tr(K+L)=1\), we construct an output-symmetric classical--quantum
channel whose bit error under the pretty good measurement (PGM)
\cite{HausladenWootters1994} is
\((1-2\tr(K^{1/2}L^{1/2}))/2\) and whose uniform-prior Holevo information
\cite{Holevo1973} is
\(\mathsf S(K)+\mathsf S(L)-\mathsf S(K+L)\).  Applying the
Alrabiah--Guruswami pretty-good criterion \cite{AlrabiahGuruswami2026} gives a nested sequence of unconditional
Horn--channel exponents.  The scalar level is the first MRRW bound, the
\(2\times2\) level is exactly \(\kappa_{\mathrm{HC}}\), and the
\(3\times3\) level is a further higher-row bound.  The spectra of
\((K,L,K+L)\) range over the Horn--hive domain, and the objective is precisely the entropy potential predicted by the representation hierarchy.  This channel construction is complementary to the representation axis: the two agree at the first nontrivial level, but no equality with the higher-row
recoupling symbol is asserted. The second level of this hierarchy (above the scalar level) already gives stronger numeric asymptotic new code bounds\footnote{This is just an initial proof-of-concept. We leave a more comprehensive analysis to future full versions of this paper.} (see the plot in Appendix~\ref{app:numerics}).

\subsection{Benchmark exponents and the formal statement}

A binary code $\cC\subseteq\{\pm1\}^n$ has minimum Hamming distance at least
$d$ when every two distinct words differ in at least $d$ coordinates.  Let
$\Atwo(n,d)$ be the largest possible size and
\[
 R_2(\delta)=\limsup_{n\to\infty}\frac1n
 \log_2\Atwo(n,\lceil\delta n\rceil).
\]
Write
\[
 \Htwo(t)=-t\log_2t-(1-t)\log_2(1-t),
 \qquad 0\log_2 0=0,
\]
and define
\begin{equation*}\label{eq:g-def}
 g(q)=\Htwo\!\left(\frac{1-\sqrt{1-q}}2\right),
 \qquad 0\le q\le1.
\end{equation*}
The first and fully optimized second MRRW exponents are
\begin{align}
 M_1(\delta)
 &=\Htwo\!\left(\frac12-\sqrt{\delta(1-\delta)}\right),
                                                               \label{eq:M1}\\
 F_\delta(\tau)
 &=1+g(\tau^2)-g(\tau^2+2\delta\tau+2\delta),
 \qquad 0\le\tau\le1-2\delta,                               \label{eq:Fdelta}\\
 M_2(\delta)&=\min_{0\le\tau\le1-2\delta}F_\delta(\tau).
                                                               \label{eq:M2}
\end{align}
Since $F_\delta(1-2\delta)=M_1(\delta)$, one has
$M_2(\delta)\le M_1(\delta)$.

The whole-cube spectral function and exponent are
\begin{align}
 \Gamma_H(a,b)&=\frac{2(a-b)(1-a-b)}{\sqrt{a(1-a)}}, \notag \\
 \kappa_H(\delta)&=
 \inf_{\substack{0\le b<a\le1/2\\
                  \Gamma_H(a,b)>1-2\delta}}
 \bigl(\Htwo(a)-\Htwo(b)\bigr).   \notag
\end{align}
The constant-weight exponent $\kappa_{\mathrm{CW}}$ is recalled precisely in
Section~\ref{sec:combination}.  The previous combined binary bound is
\begin{equation*}\label{eq:previous-bound}
 R_2(\delta)\le \kappa_{\mathrm{bin}}(\delta)
 \defeq\min\{\kappa_H(\delta),\kappa_{\mathrm{CW}}(\delta)\}
 <M_2(\delta).
\end{equation*}
We write $\RMQC$ and $\RtwoMQC$ for the MQC and 2MQC exponents of
Alrabiah and Guruswami.  Appendix~\ref{app:channel-comparison} records their
normalized variational formulas and proves
$\kappa_H=\RMQC$, $\kappa_{\mathrm{bin}}\le\RtwoMQC$, and the strict
combined honeycomb comparison with $\RtwoMQC$.  No channel parameter is
needed in the representation-theoretic construction below.

\subsection{The first honeycomb bound and two explicit branches}

For parameters
\begin{equation}\label{eq:parameters-domain}
 0<x\le\frac12,\qquad
 0\le\eta,\rho<\frac12,\qquad
 0\le z\le\frac12,
\end{equation}
put
\begin{equation}\label{eq:ABE}
 A=(1-x)(1-2\eta),\qquad
 B=x(1-2\rho),\qquad
 E=1-2z.
\end{equation}
Let $\cD_{\mathrm{HC}}$ be the set of quadruples satisfying
\begin{equation}\label{eq:triangle-continuum}
 |A-B|\le E\le A+B.
\end{equation}
Because $A,B>0$, define
\begin{equation}\label{eq:Pdef}
 P=P(x,\eta,\rho,z)
 \defeq\frac{E^2-(A-B)^2}{4AB}.
\end{equation}
The variables have a direct asymptotic meaning: $x$ is the size fraction of
the second color class, $\eta$ and $\rho$ are the normalized second-row
fractions in the two ambient partitions, and $z$ is the normalized second row
of the moving stabilizer type.  The quantities $A,B,E$ are the corresponding
normalized doubled spins.  The triangle inequalities are equivalent to
$0\le P\le1$.  Set
\begin{align}
 \Gamma_{\mathrm{HC}}(x,\eta,\rho,z)
 &\defeq 2\sqrt{x(1-x)}\Bigl[
 P\bigl(\sqrt{(1-\eta)(1-\rho)}+\sqrt{\eta\rho}\bigr)
 \nonumber\\[-1mm]
 &\hspace{39mm}
 +(1-P)\bigl(\sqrt{\rho(1-\eta)}
             +\sqrt{\eta(1-\rho)}\bigr)
 \Bigr],                                                   \label{eq:Gamma-hyp}\\
 \Phi_{\mathrm{HC}}(x,\eta,\rho,z)
 &\defeq \Htwo(x)+(1-x)\Htwo(\eta)+x\Htwo(\rho)-\Htwo(z).
                                                               \label{eq:Phi-hyp}
\end{align}
The involution
\begin{equation*}\label{eq:hyp-symmetry}
 (x,\eta,\rho,z)\longmapsto(1-x,\rho,\eta,z)
\end{equation*}
exchanges the two bipartition components and preserves both functions, so the
restriction $x\le1/2$ loses no generality.

\begin{definition}[First honeycomb exponent and first honeycomb bound]
\label{def:honeycomb}
For $0<\delta<1/2$, define
\begin{equation}\label{eq:kappa-hyp}
 \kappa_{\mathrm{HC}}(\delta)
 \defeq
 \inf_{\substack{(x,\eta,\rho,z)\in\cD_{\mathrm{HC}}\\
                  \Gamma_{\mathrm{HC}}(x,\eta,\rho,z)>1-2\delta}}
 \Phi_{\mathrm{HC}}(x,\eta,\rho,z).
\end{equation}
The inequality \(R_2(\delta)\le\kappa_{\mathrm{HC}}(\delta)\) is called
the \emph{first honeycomb bound}.  In the hierarchy developed in
Section~\ref{sec:hierarchy-overview}, we write
\(\kappa_{\mathrm{HC}}^{[1]}=\kappa_{\mathrm{HC}}\).
\end{definition}

The formula becomes transparent after writing
\begin{equation*}\label{eq:angle-param}
 \eta=\sin^2u,\qquad \rho=\sin^2v,
 \qquad 0\le u,v<\frac\pi4.
\end{equation*}
Then
\begin{equation}\label{eq:Gamma-angle}
 \Gamma_{\mathrm{HC}}
 =2\sqrt{x(1-x)}\left[P\cos(u-v)+(1-P)\sin(u+v)\right].
\end{equation}
The two terms are respectively the same-row and cross-row transfer channels.

\paragraph{Two explicit branches.}
The first branch opens only one second row.  Set $\eta=0$ and write
$s_\delta=1-2\delta$.  For $0<x\le1/2$ and $0\le\rho<1/2$, define
\begin{align}
 P_{\mathrm{face}}(x,\rho;\delta)
 &\defeq
 \frac{\dfrac{s_\delta}{2\sqrt{x(1-x)}}-\sqrt\rho}
      {\sqrt{1-\rho}-\sqrt\rho},                         \label{eq:P-face}\\
 E_{\mathrm{face}}(x,\rho;\delta)
 &\defeq
 \sqrt{\bigl(1-2x+2x\rho\bigr)^2
       +4x(1-x)(1-2\rho)P_{\mathrm{face}}(x,\rho;\delta)},
                                                               \label{eq:E-face}\\
 z_{\mathrm{face}}(x,\rho;\delta)
 &\defeq \frac{1-E_{\mathrm{face}}(x,\rho;\delta)}2.     \label{eq:z-face}
\end{align}
Let $\cD_{\mathrm{face}}(\delta)$ consist of the pairs for which
\begin{equation*}\label{eq:face-domain}
 0<P_{\mathrm{face}}(x,\rho;\delta)<1,
 \qquad 0<E_{\mathrm{face}}(x,\rho;\delta)<1,
\end{equation*}
and set
\begin{equation}\label{eq:kappa-face}
 \kappa_{\mathrm{face}}(\delta)
 \defeq
 \inf_{(x,\rho)\in\cD_{\mathrm{face}}(\delta)}
 \left\{\Htwo(x)+x\Htwo(\rho)
             -\Htwo\bigl(z_{\mathrm{face}}(x,\rho;\delta)\bigr)\right\}.
\end{equation}

The stronger explicit branch opens both second rows in an entropy-balanced
proportion.  For
\begin{equation}\label{eq:balanced-basic-domain}
 0<x\le\frac12,
 \qquad 0<t<\frac{\pi}{4(1-x)},
\end{equation}
put
\begin{align}
 u&=xt,&v&=(1-x)t,                                 \notag \\
 S_+(x,t)&=\sqrt{x(1-x)}\cos((1-2x)t),&
 S_-(x,t)&=\sqrt{x(1-x)}\sin t,          \notag \\
 A_{\mathrm{bal}}(x,t)&=(1-x)\cos(2xt),&
 B_{\mathrm{bal}}(x,t)&=x\cos(2(1-x)t).                   \label{eq:AB-balanced}
\end{align}
The range in \eqref{eq:balanced-basic-domain} implies $S_+>S_-$ and
$A_{\mathrm{bal}},B_{\mathrm{bal}}>0$: indeed
$2(1-x)t<\pi/2$ gives
$\cos((1-2x)t)>\sin t$, and both cosine factors in
\eqref{eq:AB-balanced} are positive.  Define
\begin{align}
 P_{\mathrm{bal}}(x,t;\delta)
 &=\frac{s_\delta/2-S_-(x,t)}{S_+(x,t)-S_-(x,t)},    \notag \\
 E_{\mathrm{bal}}(x,t;\delta)
 &=\sqrt{(A_{\mathrm{bal}}-B_{\mathrm{bal}})^2
       +4A_{\mathrm{bal}}B_{\mathrm{bal}}P_{\mathrm{bal}}},
                                                               \label{eq:E-balanced}\\
 z_{\mathrm{bal}}(x,t;\delta)
 &=\frac{1-E_{\mathrm{bal}}(x,t;\delta)}2.   \notag 
\end{align}
Let $\cD_{\mathrm{bal}}(\delta)$ be the pairs in
\eqref{eq:balanced-basic-domain} for which
\begin{equation*}\label{eq:balanced-domain}
 0<P_{\mathrm{bal}}(x,t;\delta)<1,
 \qquad 0<E_{\mathrm{bal}}(x,t;\delta)<1.
\end{equation*}
On this domain $\Gamma_{\mathrm{HC}}=s_\delta$ with
$\eta=\sin^2(xt)$ and $\rho=\sin^2((1-x)t)$.  Set
\begin{align}
 \Phi_{\mathrm{bal}}(x,t;\delta)
 &=\Htwo(x)+(1-x)\Htwo(\sin^2(xt))
   +x\Htwo(\sin^2((1-x)t))
   -\Htwo(z_{\mathrm{bal}}(x,t;\delta)),                 \label{eq:Phi-balanced}\\
 \kappa_{\mathrm{bal}}(\delta)
 &=\inf_{(x,t)\in\cD_{\mathrm{bal}}(\delta)}
   \Phi_{\mathrm{bal}}(x,t;\delta).   \notag
\end{align}
Every spectral-boundary point in either explicit family is approached by
strictly feasible points by decreasing $z$ slightly.  Thus both infima are
valid rate exponents.

\begin{theorem}[The First Honeycomb Bound]
\label{thm:main-asymptotic}
For every $0<\delta<1/2$,
\begin{align}
 R_2(\delta)&\le\kappa_{\mathrm{HC}}(\delta),      \notag \\
 \kappa_{\mathrm{HC}}(\delta)
 &\le\kappa_{\mathrm{face}}(\delta)<\kappa_H(\delta), \notag \\
 \kappa_{\mathrm{HC}}(\delta)
 &\le\kappa_{\mathrm{bal}}(\delta)<\kappa_H(\delta),   \notag \\
 \kappa_{\mathrm{bal}}(\delta)&\le M_2(\delta).    \notag
\end{align}
Define
\begin{align}
 \kappa_{\mathrm{best}}(\delta)
 &\defeq\min\{\kappa_{\mathrm{CW}}(\delta),
                \kappa_{\mathrm{HC}}(\delta)\},  \notag \\
 \kappa_{\mathrm{explicit}}(\delta)
 &\defeq\min\{\kappa_{\mathrm{CW}}(\delta),
                \kappa_{\mathrm{face}}(\delta),
                \kappa_{\mathrm{bal}}(\delta)\},  \notag \\
 \kappa_{\mathrm{pair}}(\delta)
 &\defeq\min\{\kappa_{\mathrm{CW}}(\delta),
                \kappa_{\mathrm{bal}}(\delta)\}.  \notag
\end{align}
Then
\begin{align}
 R_2(\delta)
 &\le\kappa_{\mathrm{best}}(\delta)
 \le\kappa_{\mathrm{explicit}}(\delta)
 \le\kappa_{\mathrm{bin}}(\delta)
 \le\RtwoMQC(\delta)<M_2(\delta)\le M_1(\delta), \notag \\
 R_2(\delta)
 &\le\kappa_{\mathrm{best}}(\delta)
 \le\kappa_{\mathrm{explicit}}(\delta)
 \le\kappa_{\mathrm{pair}}(\delta)
 <\RtwoMQC(\delta)<M_2(\delta),    \notag \\
 \kappa_H(\delta)&=\RMQC(\delta).  \notag
\end{align}
\end{theorem}

The first comparison is the main state-of-the-art statement.  It places the
new complete honeycomb combination below the previous combined OpenAI
exponent and places that previous exponent below 2MQC.  The second comparison
is stronger where strictness matters: the explicit balanced honeycomb branch,
combined with \(\kappa_{\mathrm{CW}}\), is strictly below 2MQC for every
\(0<\delta<1/2\).  The equality \(\kappa_H=\RMQC\) identifies the two
one-branch descriptions exactly.

These conclusions come from three logically distinct arguments.  The finite
representation graph and its F{\o}lner limit prove
\(R_2\le\kappa_{\mathrm{HC}}\).  Perturbations of the whole-cube boundary and
the symmetric-slice calculation prove the honeycomb comparisons with
\(\kappa_H\) and \(M_2\).  Appendix~\ref{app:channel-comparison} then proves
the exact MQC identity and the two analytic comparisons with 2MQC.  No
numerical experiment is used in any of these statements.

\subsection{Historical context and related work}

\paragraph{Delsarte, MRRW, and spectral reformulations.}
Delsarte's linear program compresses a code to its pair-distance distribution
and imposes the positivity relations of the Hamming association scheme.
MRRW constructed explicit dual solutions from Krawtchouk and related
orthogonal polynomials, obtaining the first and second LP bounds
\cite{Delsarte1973,MRRW1977}.  Schrijver later identified Delsarte's bound
with the symmetry reduction of the Lov\'asz \(\vartheta'\) bound on the
corresponding forbidden-distance graph \cite{Schrijver1979}.

Several subsequent works clarified the spectral mechanism behind the first LP
bound.  Friedman and Tillich developed generalized Alon--Boppana arguments for
coding problems \cite{FriedmanTillich2005}.  Barg and Nogin developed a
spectral and functional approach to Delsarte bounds using finite Jacobi
matrices and Christoffel--Darboux kernels
\cite{BargNogin2006,BargNogin2008}.  Navon and Samorodnitsky gave Fourier and
covering arguments \cite{NavonSamorodnitsky2005,NavonSamorodnitsky2009};
Samorodnitsky analyzed the optimum of Delsarte's program and later recast the
first bound through low-degree Walsh projections
\cite{Samorodnitsky2001,Samorodnitsky2023}; and Linial and Loyfer gave an
elementary proof by counting walks in the Hamming cube
\cite{LinialLoyfer2023Elementary}.  These viewpoints illuminate the first LP
bound, but do not by themselves recover the full second MRRW optimization.
Samorodnitsky isolated concrete obstacles to extending this line of argument
\cite{Samorodnitsky2025}.  Chailloux and Debris-Alazard constructed universal
Delsarte-dual solutions in association schemes, recovered MRRW-type bounds
from a Laplacian viewpoint, and interpreted the second LP bound in that
framework \cite{ChaillouxDebrisAlazard2025}.

\paragraph{Multipoint bounds and complete hierarchies.}
Pairwise distance data do not exhaust the invariant information carried by a
code.  Schrijver's Terwilliger-algebra semidefinite program incorporates
three-point data, and Gijswijt, Mittelmann, and Schrijver developed stronger
four-point bounds \cite{Schrijver2005,GijswijtMittelmannSchrijver2012}.
General Lasserre/Sum-of-Squares and packing hierarchies converge at
sufficiently high level \cite{Laurent2003,Laurent2007,deLaatVallentin2015},
but their finite-blocklength strength has not come with an explicit
asymptotic analysis comparable to MRRW.  This is the longstanding tension
between complete relaxations and asymptotically tractable ones.

In the spherical setting, Bachoc and Vallentin developed a
fixed-basepoint three-point semidefinite bound whose symmetry reduction
uses matrix-valued higher-harmonic kernels
\cite{BachocVallentin2008}.  This is an important precursor to the use
of richer representation-theoretic blocks. The moving-projection
construction differs in that the selected subspaces move equivariantly
with both codewords and enter a two-point positive-kernel argument
\cite[Chapter~2]{OpenAI2026}.

For linear codes, higher-order MacWilliams and Krawtchouk constraints yield a
structured LP hierarchy.  Coregliano, Jeronimo, and Jones and Loyfer and
Linial developed related formulations; subsequent work established exact
completeness, objective-preserving lifts, and explicit spectral dual
constructions
\cite{MacWilliams1963,CoreglianoJeronimoJones2022,LoyferLinial2023,
CoreglianoJeronimoJones2023,CoreglianoJeronimoJonesLinialLoyfer2026}.
These programs rely essentially on linearity, and the corresponding
unconstrained hierarchy for general codes collapses to Delsarte's pairwise LP.
The honeycomb hierarchy takes a different route: it strengthens a
representation-theoretic certificate valid for arbitrary codes and then
embeds the resulting quadratic kernels in a complete stable-set moment
hierarchy.

\subsection{Proof overview}

The proof is organized around the honeycomb construction.
\begin{enumerate}[label=(\roman*),leftmargin=2.2em]
\item We first prove an abstract profile-optimized moving-projection theorem.
A directed Collatz--c certificate controls the spectral term, while a
blockwise trace argument retains an effective dimension instead of replacing
it immediately by the total ambient dimension.
\item We then classify every two-row hyperoctahedral irreducible compatible
with \(S^{(n-k,k)}\) and compute all coordinate box-transfer coefficients.
Multiplicity one reduces the local recoupling to half-spin Wigner \(6j\)
coefficients, and exact dimension balance makes the finite graph amenable to
the abstract theorem.
\item F{\o}lner boxes in the three-dimensional compatibility region yield the
four-channel continuum symbol and the entropy potential defining
\(\kappa_{\mathrm{HC}}\).  Opening one or both ambient second rows gives the
strict whole-cube improvements, while the symmetric slice recovers the full
second MRRW optimization exactly.
\item We combine the new branch with \(\kappa_{\mathrm{CW}}\).  A separate
compactified analysis of the MQC and 2MQC variational problems proves
\(\kappa_H=R_{\mathrm{MQC}}\),
\(\kappa_{\mathrm{bin}}\le R_{\mathrm{2MQC}}\), and
\(\kappa_{\mathrm{pair}}<R_{\mathrm{2MQC}}\) pointwise.
\item Finally, frame-profile certificates extend moving projections to
arbitrary Littlewood--Richardson multiplicity, and anchored localizers place
the resulting quadratic kernels inside a finitely exact stable-set moment
hierarchy.  In the complementary Horn--channel construction, explicit matrix
channels and the Alrabiah--Guruswami pretty-good criterion give the nested
higher-row asymptotic bounds.
\end{enumerate}

\paragraph{Scope and terminology.}
All asymptotic comparison theorems in the paper are analytic; the numerical
experiments in Appendix~\ref{app:numerics} are not used in their proofs.
``Complete two-row'' refers to the full finite family of compatible two-row
ambient representations.  The asymptotic exponent is obtained from
nondegenerate bulk scalings and their ordinary boundary limits; additional
bounded-row-difference boundary layers could only strengthen it.  ``Exact''
for the anchored hierarchy refers to finite convergence to \(A_2(n,d)\) at
sufficiently high anchor order.

\paragraph{Organization.}
Sections~\ref{sec:general}--\ref{sec:asymptotics} develop the finite
moving-projection theorem, construct the two-row representation graph, and
derive the first honeycomb exponent.  Section~\ref{sec:strict} proves the
strict comparisons and combines the new branch with the constant-weight
bound.  Sections~\ref{sec:hierarchy-overview}--\ref{sec:moment-hierarchy}
develop the representation, Horn--channel, and anchored moment hierarchies.
The appendices contain the recoupling and entropy asymptotics, a local
stability calculation for the MRRW slice, the numerical experiments, the
comparison with the MQC and 2MQC bounds, and the \(q\)-ary details.

\section{A profile-optimized finite moving-projection theorem}
\label{sec:general}

All coding bounds in the paper will be obtained from the following abstract
finite-group statement.  Compared with the Perron-vector formulation of the
moving-projection method, we retain an arbitrary positive profile and keep the
trace contribution of each irreducible block separate.  This yields a
directed Collatz--Wielandt spectral certificate \cite{Collatz1942,Wielandt1950}
together with an effective,
rather than total, ambient dimension.

Let a finite group \(G\) act transitively on a finite set \(X\), fix
\(o\in X\), and put \(H=\operatorname{Stab}_G(o)\).  Let \(W\) be a unitary
\(G\)-representation and suppose \(x\mapsto\ell_x\in W\) is equivariant with
\(\norm{\ell_x}=1\).  Assume that the matrix coefficients used for the
ordering are real and write
\[
 t(x,y)=\ip{\ell_x}{\ell_y}\in\R.
\]
Fix an irreducible unitary \(H\)-representation \(E\) of dimension \(d_E\).
Let \(\Omega\) be a finite set of pairwise inequivalent irreducible unitary
\(G\)-representations \(V_\omega\) such that
\begin{equation*}\label{eq:mult-one}
 \dim\Hom_H(E,\Res_H^G V_\omega)=1
 \qquad(\omega\in\Omega).
\end{equation*}
Write \(D_\omega=\dim V_\omega\), and choose an isometric embedding
\(\phi_{\omega,o}:E\to V_\omega\).  For \(x=go\), let \(E_x\) be the fiber
of the associated homogeneous bundle and define
\[
 \phi_{\omega,x}[g,u]=g\phi_{\omega,o}u.
\]

If \(V_{\omega'}\) occurs in \(W\otimes V_\omega\), choose an isometric
\(G\)-intertwiner
\[
 C_{\omega,\omega'}:V_{\omega'}\longrightarrow W\otimes V_\omega.
\]
Multiplicity one implies that contraction by \(\ell_x\) acts by a scalar on
the moving copy of \(E\):
\begin{equation}\label{eq:contraction-general}
 C_{\omega,\omega'}^*(\ell_x\otimes\phi_{\omega,x}u)
 =c_{\omega,\omega'}\phi_{\omega',x}u.
\end{equation}
Choose phases so that \(c_{\omega,\omega'}\ge0\), and put
\(p_{\omega,\omega'}=c_{\omega,\omega'}^2\); for a missing edge, set
\(p_{\omega,\omega'}=0\).  Assume the balance identity
\begin{equation*}\label{eq:balance-general}
 D_\omega p_{\omega,\omega'}
 =D_{\omega'}p_{\omega',\omega}
\end{equation*}
for every retained edge.  The associated symmetric matrix is
\begin{equation*}\label{eq:J-general}
 J_{\omega,\omega'}
 =\sqrt{p_{\omega,\omega'}p_{\omega',\omega}}.
\end{equation*}

The directed coefficients are automatically substochastic on every retained
vertex set.  Indeed, for a unit vector $u\in E$, the images of the
inequivalent $V_{\omega'}$-summands in $W\otimes V_\omega$ are orthogonal, so
\begin{equation*}\label{eq:directed-substochastic}
 \sum_{\omega'\in\Omega}p_{\omega,\omega'}
 =\sum_{\substack{\omega'\in\Omega\\p_{\omega,\omega'}>0}}
   \left\lVert C_{\omega,\omega'}^*
   (\ell_o\otimes\phi_{\omega,o}u)\right\rVert^2
 \le 1.
\end{equation*}
Consequently the directed matrix $P=(p_{\omega,\omega'})$ has spectral
radius at most one.  By the Collatz--Wielandt formula
\cite{Collatz1942,Wielandt1950} applied to $P^T$,
\begin{equation*}\label{eq:lambda-at-most-one}
 \lambda_\Omega(w)\le1
 \qquad\text{for every positive profile }w.
\end{equation*}
In particular, the hypothesis of Theorem~\ref{thm:general-graph} implies
$s<1$, so the prefactor $1-s$ in the conclusion is nonnegative.

For a positive profile \(w=(w_\omega)_{\omega\in\Omega}\), define
\begin{align}
 S_\nu(w)&\defeq\sum_{\omega\in\Omega}p_{\omega,\nu}w_\omega, \notag \\
 \lambda_\Omega(w)&\defeq
 \min_{\nu\in\Omega}\frac{S_\nu(w)}{w_\nu}, \notag \\
 Z_w&\defeq\sum_{\omega\in\Omega}w_\omega,\qquad
 \mathcal D_{\mathrm{eff}}(w)
 \defeq\frac{Z_w^2}{\sum_{\omega\in\Omega}w_\omega^2/D_\omega}.
                                                               \label{eq:Deff}
\end{align}
The quantity \(\lambda_\Omega(w)\) is a directed Collatz--Wielandt lower
certificate, while \(\mathcal D_{\mathrm{eff}}(w)\) is the blockwise
effective dimension.

\begin{theorem}[Profile-Optimized Representation-Graph Bound]
\label{thm:general-graph}
If a positive profile \(w\) satisfies
\(\lambda_\Omega(w)>\max\{s,0\}\), then every
\(\cC\subseteq X\) with \(t(x,y)\le s\) for distinct \(x,y\in\cC\) obeys
\begin{equation}\label{eq:general-profile-bound}
 |\cC|\le
 \frac{1-s}{d_E\bigl(\lambda_\Omega(w)-s\bigr)}
 \mathcal D_{\mathrm{eff}}(w).
\end{equation}
Consequently one may take the infimum of the right-hand side over every
positive profile satisfying the spectral inequality.
\end{theorem}

\begin{proof}
Take \(\lambda=\lambda_\Omega(w)\), \(S_\nu=S_\nu(w)\), and
\(Z=Z_w\), and set
\[
 \cV_\Omega=\bigoplus_{\omega\in\Omega}V_\omega.
\]

\smallskip
\noindent
Define the moving isometry
\[
 \Psi_xu=\bigoplus_{\omega\in\Omega}
 \sqrt{\frac{w_\omega}{Z}}\,\phi_{\omega,x}u,
 \qquad P_x=\Psi_x\Psi_x^*.
\]
For \(h=\bigoplus_\nu h_\nu\in\cV_\Omega\), define
\(B:\cV_\Omega\to W\otimes\cV_\Omega\) blockwise by
\begin{equation*}\label{eq:B-profile}
 (Bh)_\omega
 =\sum_{\substack{\nu\in\Omega\\p_{\omega,\nu}>0}}
 \alpha_{\omega,\nu}C_{\omega,\nu}h_\nu,
 \qquad
 \alpha_{\omega,\nu}
 =c_{\omega,\nu}\sqrt{w_\omega}
   \frac{\sqrt{\lambda w_\nu}}{S_\nu}.
\end{equation*}
The inequivalent \(V_\nu\)-images inside each \(W\otimes V_\omega\) are
orthogonal.  Therefore \(B^*B\) is block diagonal and its scalar on
\(V_\nu\) is
\[
 \sum_\omega\alpha_{\omega,\nu}^2
 =\frac{\lambda w_\nu}{S_\nu}\le1.
\]
Thus \(B\) is a contraction.  Moreover,
\eqref{eq:contraction-general} gives the exact identity
\begin{align}
 B^*(\ell_x\otimes\Psi_xu)
 &=\bigoplus_\nu
 \frac{\sqrt{\lambda w_\nu}}{S_\nu\sqrt Z}
 \left(\sum_\omega p_{\omega,\nu}w_\omega\right)
 \phi_{\nu,x}u\nonumber\\
 &=\sqrt\lambda\,\Psi_xu.                               \label{eq:key-contraction}
\end{align}

\smallskip
\noindent
Put \(K(x,y)=\tr(P_xP_y)\ge0\), and define Hilbert--Schmidt operators
\[
 L_x=(\ell_x\otimes\Id)P_x,
 \qquad G_x=BP_x,
 \qquad \Theta_x=L_x-\sqrt\lambda\,G_x.
\]
Let \(K_B(x,y)=\tr(P_xB^*BP_y)\).  Expanding with
\eqref{eq:key-contraction} gives
\begin{equation*}\label{eq:theta-profile}
 \ip{\Theta_x}{\Theta_y}_{\mathrm{HS}}
 =\bigl(t(x,y)-2\lambda\bigr)K(x,y)+\lambda K_B(x,y).
\end{equation*}
Hence
\begin{align}
 (t(x,y)-s)K(x,y)
 &=(\lambda-s)K(x,y)
   +\ip{\Theta_x}{\Theta_y}_{\mathrm{HS}}
   +\lambda\bigl(K(x,y)-K_B(x,y)\bigr).                 \label{eq:PD-decomp}
\end{align}
All three kernels on the right are positive definite.  For the last one,
write it as the Hilbert--Schmidt Gram kernel of
\((\Id-B^*B)^{1/2}P_x\).

\smallskip
\noindent
Let \(N=|\cC|\).  The left side of \eqref{eq:PD-decomp} is nonpositive off
the diagonal on \(\cC\), while its diagonal value is \((1-s)d_E\).
Summing over \(\cC^2\) and discarding the two additional nonnegative Gram
sums yields
\begin{equation}\label{eq:profile-upper-sum}
 (\lambda-s)\norm{\sum_{x\in\cC}P_x}_{\mathrm{HS}}^2
 \le N(1-s)d_E.
\end{equation}
Let \(\Pi_\omega\) project onto \(V_\omega\).  The diagonal block of every
\(P_x\) has trace \(d_Ew_\omega/Z\).  Blockwise trace
Cauchy--Schwarz therefore gives
\begin{align}
 \norm{\sum_{x\in\cC}P_x}_{\mathrm{HS}}^2
 &\ge\sum_\omega
 \norm{\Pi_\omega\left(\sum_xP_x\right)\Pi_\omega}_{\mathrm{HS}}^2
 \nonumber\\
 &\ge\frac{N^2d_E^2}{Z^2}
       \sum_\omega\frac{w_\omega^2}{D_\omega}
 =\frac{N^2d_E^2}{\mathcal D_{\mathrm{eff}}(w)}. \notag
\end{align}
Substitution into \eqref{eq:profile-upper-sum} proves
\eqref{eq:general-profile-bound}.
\end{proof}

\begin{corollary}
\label{cor:perron-profile}
Assume that the retained positive-weight graph is connected.  Let \(v\) be
the positive unit Perron vector of \(J_\Omega\), and let
\(\Lambda_\Omega=\lambda_{\max}(J_\Omega)\).  If
\(\Lambda_\Omega>\max\{s,0\}\), then
\begin{equation}\label{eq:perron-effective-bound}
 |\cC|\le
 \frac{1-s}{d_E(\Lambda_\Omega-s)}
 \left(\sum_{\omega\in\Omega}\sqrt{D_\omega}\,v_\omega\right)^2
 \le
 \frac{1-s}{d_E(\Lambda_\Omega-s)}
 \sum_{\omega\in\Omega}D_\omega.
\end{equation}
The second inequality is strict unless \(v\) is proportional to
\((\sqrt{D_\omega})_{\omega\in\Omega}\).
\end{corollary}

\begin{proof}
Set \(w_\omega=\sqrt{D_\omega}\,v_\omega\).  Balance and
\(J_\Omega v=\Lambda_\Omega v\) imply
\[
 \sum_\omega p_{\omega,\nu}w_\omega
 =\Lambda_\Omega w_\nu,
\]
so \(\lambda_\Omega(w)=\Lambda_\Omega\).  Since \(\sum v_\omega^2=1\),
\(\mathcal D_{\mathrm{eff}}(w)=(\sum\sqrt{D_\omega}v_\omega)^2\).
The final inequality and its equality condition are Cauchy--Schwarz.
\end{proof}

\begin{remark}
The Perron vector maximizes \(\lambda_\Omega(w)\) by the
Collatz--Wielandt theorem, but it need not minimize the product of the
spectral prefactor and effective dimension in
\eqref{eq:general-profile-bound}.  Thus optimizing \(w\) is never worse than
the Perron-effective bound and can be strictly better.

There is one important equality case.  Suppose \(\Omega\) is a complete
normalized compatibility component, so that
\(\sum_{\nu}p_{\omega,\nu}=1\) for every vertex.  Balance then gives
\[
 J_\Omega(\sqrt{D_\omega})_{\omega\in\Omega}
 =(\sqrt{D_\omega})_{\omega\in\Omega}.
\]
If the component is connected, this is its Perron vector and
\(\Lambda_\Omega=1\); consequently the Cauchy--Schwarz inequality in
\eqref{eq:perron-effective-bound} is an equality.  Strict Perron-effective
gains therefore arise from proper truncations or other retained subgraphs,
not from the entire normalized component.  The unrestricted profile
optimization may still trade a smaller spectral certificate for a smaller
effective dimension.
\end{remark}

\section{The two-row representation graph and the finite honeycomb bound}
\label{sec:representation}

We specialize the abstract theorem to the Hamming cube.  The vertices of the
finite representation graph are precisely the compatible two-row
bipartitions, and the natural signed representation connects them by
single-box transfers.  Section~\ref{sec:weights} then computes the weight of
every such transfer exactly, and Section~\ref{sec:finite} assembles these
ingredients into a finite-length coding theorem.

\subsection{The cube and its compatible irreducibles}

Let
\[
 B_n=C_2^n\rtimes S_n
\]
be the hyperoctahedral group of signed permutations.  It acts transitively on
\(X_n=\{\pm1\}^n\).  At the all-ones word \(o\), the stabilizer is the
permutation subgroup \(H=S_n\), so
\[
 X_n\simeq B_n/S_n.
\]
Let \(W=\C^n\) be the natural signed permutation representation and set
\[
 \ell_x=\frac{x}{\sqrt n}.
\]
Then
\begin{equation}\label{eq:hamming-inner}
 t(x,y)=\ip{\ell_x}{\ell_y}
 =1-\frac{2d_H(x,y)}n.
\end{equation}
Thus a binary code of minimum distance at least \(d\) has
\(t(x,y)\le s=1-2d/n\) off the diagonal.

The irreducible complex representations of \(B_n\) are indexed by
bipartitions \((\lambda,\mu)\) with \(|\lambda|+|\mu|=n\); see, for
example, \cite{JamesKerber1981,DoeraeneIommi1989,Green2022}.  We denote the corresponding
representation by \(V_{\lambda,\mu}\).  If \(|\mu|=i\), then
\begin{equation*}\label{eq:B-dimension}
 \dim V_{\lambda,\mu}
 =\binom ni f^\lambda f^\mu,
\end{equation*}
where \(f^\theta=\dim S^\theta\) is the dimension of the Specht module of
shape \(\theta\).  We use the conventions that the empty partition is
allowed, \(S^\varnothing\cong\C\), and \(f^\varnothing=1\).

Restricting to the permutation subgroup removes the distinction between the
two \(C_2\)-colors and gives
\begin{equation*}\label{eq:restriction-Sn}
 \Res_{S_n}^{B_n}V_{\lambda,\mu}
 \cong
 \Ind_{S_{n-i}\times S_i}^{S_n}
       (S^\lambda\boxtimes S^\mu)
 \cong\bigoplus_{\nu\vdash n}c_{\lambda\mu}^{\nu}S^\nu,
\end{equation*}
where \(c_{\lambda\mu}^{\nu}\) is a Littlewood--Richardson coefficient.

Fix the two-row stabilizer type
\begin{equation*}\label{eq:nu}
 E_\nu=S^\nu,
 \qquad \nu=(n-k,k),
 \qquad 0\le k\le n/2.
\end{equation*}
The next proposition gives the complete vertex set for this stabilizer
type.  In particular, the two-row family is forced by compatibility rather
than chosen as an auxiliary subfamily.

\begin{proposition}
\label{prop:complete-region}
If \(c_{\lambda\mu}^{(n-k,k)}>0\), then both \(\lambda\) and \(\mu\)
have at most two rows.  Write
\begin{equation}\label{eq:triple-param}
 \lambda=(n-i-j,j),\qquad
 \mu=(i-\ell,\ell).
\end{equation}
Set
\begin{equation}\label{eq:doubled-spins-finite}
 A_0=n-i-2j,\qquad B_0=i-2\ell,\qquad E_0=n-2k.
\end{equation}
Then
\begin{equation}\label{eq:LR-triangle}
 c_{\lambda\mu}^{(n-k,k)}=1
 \quad\Longleftrightarrow\quad
 |A_0-B_0|\le E_0\le A_0+B_0.
\end{equation}
Otherwise the coefficient is zero.
\end{proposition}

\begin{proof}
Nonvanishing of \(c_{\lambda\mu}^{\nu}\) implies \(\lambda\subseteq\nu\),
so \(\lambda\) has at most two rows.  By the symmetry
\(c_{\lambda\mu}^{\nu}=c_{\mu\lambda}^{\nu}\), the same is true for
\(\mu\).

For a two-row partition \((r+q,q)\), the associated polynomial
\(GL_2\)-representation is
\[
 \det^q\otimes\Sym^r(\C^2).
\]
The Clebsch--Gordan formula gives
\[
 \Sym^{A_0}(\C^2)\otimes\Sym^{B_0}(\C^2)
 \cong\bigoplus_{t=0}^{\min\{A_0,B_0\}}
       \det^t\otimes\Sym^{A_0+B_0-2t}(\C^2).
\]
After restoring the determinant twists from the second rows, the target
second row is \(k=j+\ell+t\).  Hence the multiplicity is one exactly when
\(0\le k-j-\ell\le\min\{A_0,B_0\}\), which is equivalent to
\eqref{eq:LR-triangle}.
\end{proof}

We call a triple \(\omega=(i,j,\ell)\) \emph{valid} when the partitions in
\eqref{eq:triple-param} exist and the triangle condition
\eqref{eq:LR-triangle} holds.  Proposition~\ref{prop:complete-region} says
that these triples are not an ad hoc subfamily: they index every ambient
irreducible containing the fixed two-row stabilizer type.

\subsection{Coordinate multiplication transfers one box}

The natural signed representation is
\[
 W\cong V_{(n-1),(1)}.
\]
Tensoring by \(W\) is governed by a multiplicity-free box-transfer rule.

\begin{lemma}
\label{lem:Pieri}
For every bipartition \((\lambda,\mu)\) of \(n\),
\(W\otimes V_{\lambda,\mu}\) is the multiplicity-free direct sum of the
bipartitions obtained either by removing one removable box from \(\lambda\)
and adding one addable box to \(\mu\), or by doing the reverse.
\end{lemma}

\begin{proof}
This is a standard consequence of induction--restriction for wreath products;
a recent explicit statement is \cite[Theorem~5.1]{Bae2026}.  For completeness,
let \(B_{n-1}\times C_2\) be the stabilizer of a signed coordinate.  The
natural representation is induced from the nontrivial character of the last
\(C_2\)-factor.  The projection formula
\(\Ind_H^G(\chi)\otimes V\cong\Ind_H^G(\chi\otimes\Res_H V)\) reduces
\(W\otimes V_{\lambda,\mu}\) to restriction from \(B_n\) to
\(B_{n-1}\times C_2\), tensoring by that character, and induction back.
The wreath-product branching rule removes one box from either component;
tensoring by the last sign character switches its color, and induction adds
that box to the other component.  Ordinary branching is multiplicity free,
and the resulting bipartition determines the removed and added boxes.
\end{proof}

For a valid source \(\omega=(i,j,\ell)\), the four forward box transfers are
\begin{equation*}\label{eq:four-directions}
 (i,j,\ell)\longrightarrow
 (i+1,j-\sigma,\ell+\tau),
 \qquad \sigma,\tau\in\{0,1\},
\end{equation*}
whenever the target partitions and stabilizer copy are valid.  Here
\(\sigma=0\) removes a box from the first row of \(\lambda\), while
\(\sigma=1\) removes one from its second row; similarly, \(\tau=0\) adds to
the first row of \(\mu\), while \(\tau=1\) adds to its second row.  Reverse
edges transfer a box from \(\mu\) back to \(\lambda\).

\subsection{Exact transition coefficients}
\label{sec:weights}

The remaining finite input is the weight of each box-transfer edge.  After
contracting a coordinate intertwiner against the stabilizer-fixed vector
\(\ell_o=n^{-1/2}(1,\ldots,1)\), Schur's lemma reduces the map between the
two multiplicity-one copies of \(S^{(n-k,k)}\) to a scalar.  Its squared
magnitude factors into three terms: a color-sector probability, a normalized
Young-branching ratio, and a unitary Racah overlap.  The following lemma fixes the normalization needed to derive the closed transition formula, the resulting
formula is stated in Theorem~\ref{thm:exact-p}.

Fix a forward edge of type \((\sigma,\tau)\) from
\(\omega=(i,j,\ell)\) to
\(\omega'=(i+1,j-\sigma,\ell+\tau)\).  Introduce half-integer spins
\begin{equation}\label{eq:spins}
 \begin{aligned}
 a&=\frac{n-i-1}{2}-(j-\sigma),
 &c&=\frac{n-i}{2}-j,\\
 d&=\frac{i}{2}-\ell,
 &f&=\frac{i+1}{2}-(\ell+\tau),\\
 e&=\frac n2-k.
 \end{aligned}
\end{equation}
Thus \(c=a\pm\tfrac12\), \(f=d\pm\tfrac12\), and the two coupling schemes
are
\[
 ((a\otimes\tfrac12)\to c)\otimes d\to e,
 \qquad
 a\otimes((\tfrac12\otimes d)\to f)\to e.
\]
Let
\begin{equation*}\label{eq:Racah}
 R_{\sigma\tau}
 =\sqrt{(2c+1)(2f+1)}
 \begin{Bmatrix}
 a&\frac12&c\\
 d&e&f
 \end{Bmatrix},
\end{equation*}
where the braces denote the Wigner \(6j\) symbol.  Its sign depends on phase
conventions but \(R_{\sigma\tau}^2\) does not.

\begin{lemma}
\label{lem:unitary-recoupling}
Put \(p=n-i\) and \(q=i\).  Let
\(\lambda\vdash p\), \(\mu\vdash q\), and
\(\nu=(n-k,k)\), and suppose \(c_{\lambda\mu}^{\nu}=1\).
Choose a removable box of \(\lambda\), producing
\(\lambda'\vdash p-1\), and an addable box of \(\mu\), producing
\(\mu'\vdash q+1\).  Normalize all Young branching maps, all
Littlewood--Richardson embeddings, and the two-row Schur--Weyl decomposition
to be isometries.  If \(T_{\lambda,\mu}^{\lambda',\mu'}\) denotes the
coordinate-transfer intertwiner normalized by the vector
\(n^{-1/2}(1,\ldots,1)\), then on the unique \(S^\nu\)-summand one has
\begin{equation}\label{eq:unitary-recoupling-scalar}
 \left|T_{\lambda,\mu}^{\lambda',\mu'}\right|
 =\sqrt{\frac pn}\,
  \sqrt{\frac{f^{\lambda'}}{f^\lambda}}\,
  \sqrt{(2c+1)(2f+1)}
  \left|
  \begin{Bmatrix}
   a&\frac12&c\\ d&e&f
  \end{Bmatrix}
  \right|,
\end{equation}
where the spins are those in \eqref{eq:spins}; the left-hand side denotes
the absolute value of the scalar induced on the unique \(S^\nu\)-summand.
The reverse transfer has the same Racah factor, with color and branching
prefactor
\[
 \sqrt{\frac{q+1}{n}}\sqrt{\frac{f^\mu}{f^{\mu'}}}.
\]
\end{lemma}

\begin{proof}
We spell out the normalization because this is the only point at which
color, Young branching, and recoupling meet.  All induced representations
below carry the standard unitary coset inner product.

\smallskip
\noindent\emph{Step 1: compatible unitary branching models.}
Use the unitary Schur--Weyl decomposition
\begin{equation}\label{eq:unitary-Schur-Weyl}
 (\C^2)^{\otimes m}
 \cong
 \bigoplus_{r=0}^{\lfloor m/2\rfloor}
 U_{m/2-r}\otimes S^{(m-r,r)},
\end{equation}
where $U_j$ is the spin-$j$ irreducible of $SU(2)$.  Splitting off one tensor
factor simultaneously gives orthogonal Young branching on the Specht factor
and unitary Clebsch--Gordan branching on the spin factor.  Thus removing a
first- or second-row box corresponds to $c=a\pm\tfrac12$, and adding a box
to $\mu$ corresponds to $f=d\pm\tfrac12$.

\smallskip
\noindent\emph{Step 2: the color and Young-branching factor.}
For each of the $p$ coordinates in the $\lambda$-colored block, let $Q_r^{\lambda'}$ be the orthogonal projector in
$S^\lambda$ onto the $S^{\lambda'}$-summand obtained by restricting to the
subgroup fixing coordinate $r$.  It has rank $f^{\lambda'}$.  The average
\[
 A_{\lambda'}=\frac1p\sum_{r=1}^p Q_r^{\lambda'}
\]
commutes with $S_p$, because conjugation permutes the summands.  Schur's
lemma and the trace identity give
\begin{equation}\label{eq:plancherel-down-factor}
 A_{\lambda'}
 =\frac{f^{\lambda'}}{f^\lambda}\,\Id_{S^\lambda}.
\end{equation}
The coordinate sectors in the natural representation are orthogonal.  On a
fixed induced color coset, let
$\ell_\lambda=n^{-1/2}\sum_{r=1}^p e_r$ be the part of $\ell_o$ supported on
the $\lambda$-colored coordinates, and let $\Pi_{\lambda'}$ project onto the
sum of the forward Pieri sectors whose down-branch is $\lambda'$, before a
particular addable $\mu'$-branch is selected.  For every $\xi\in S^\lambda$,
\begin{align}
 \norm{\Pi_{\lambda'}(\ell_\lambda\otimes\xi)}^2
 &=\frac1n\sum_{r=1}^p\ip{\xi}{Q_r^{\lambda'}\xi}
 =\frac pn\frac{f^{\lambda'}}{f^\lambda}\norm{\xi}^2. \notag
\end{align}
Hence the color-down-branch map becomes an isometry after division by
\begin{equation}\label{eq:color-down-normalization}
 \beta_{\lambda\to\lambda'}
 =\sqrt{\frac pn}\sqrt{\frac{f^{\lambda'}}{f^\lambda}}.
\end{equation}

\smallskip
\noindent\emph{Step 3: a common induced module for the two association
orders.}
It remains to identify the overlap produced by selecting $\mu'$ and the
target copy of $S^\nu$.  Put
\[
 K=S_{p-1}\times S_1\times S_q,
 \qquad
 \tau=S^{\lambda'}\boxtimes\one\boxtimes S^\mu,
 \qquad
 \mathcal I=\Ind_K^{S_n}\tau.
\]
Concretely, we use
\(
 \Ind_L^G U=\bigoplus_{gL}g\otimes U
\)
with the orthogonal coset-sum norm.  Under this convention, transitivity of
unitary induction merely regroups the same orthogonal summands and gives two
canonical unitary realizations of the same $S_n$-module:
\begin{align}
 \mathcal I
 &\cong
 \Ind_{S_p\times S_q}^{S_n}
 \left(
  \Ind_{S_{p-1}\times S_1}^{S_p}
       (S^{\lambda'}\boxtimes\one)
  \boxtimes S^\mu
 \right), \notag\\
 \mathcal I
 &\cong
 \Ind_{S_{p-1}\times S_{q+1}}^{S_n}
 \left(
  S^{\lambda'}\boxtimes
  \Ind_{S_1\times S_q}^{S_{q+1}}
       (\one\boxtimes S^\mu)
 \right). \notag
\end{align}
Selecting the multiplicity-one $S^\lambda$-summand in the first inner
induction gives an isometric inclusion
\[
 J_c:\Ind_{S_p\times S_q}^{S_n}(S^\lambda\boxtimes S^\mu)
     \longrightarrow\mathcal I,
\]
while selecting $S^{\mu'}$ in the second gives an isometric inclusion
\[
 J_f:\Ind_{S_{p-1}\times S_{q+1}}^{S_n}
          (S^{\lambda'}\boxtimes S^{\mu'})
     \longrightarrow\mathcal I.
\]
Let
\[
 \phi_{\lambda,\mu}^{\nu}:S^\nu\longrightarrow
 \Ind_{S_p\times S_q}^{S_n}(S^\lambda\boxtimes S^\mu),
 \qquad
 \phi_{\lambda',\mu'}^{\nu}:S^\nu\longrightarrow
 \Ind_{S_{p-1}\times S_{q+1}}^{S_n}
      (S^{\lambda'}\boxtimes S^{\mu'})
\]
be chosen isometric Littlewood--Richardson embeddings (unique up to
phase), and set
\begin{equation*}\label{eq:two-coupling-isometries}
 I_c=J_c\phi_{\lambda,\mu}^{\nu},
 \qquad
 I_f=J_f\phi_{\lambda',\mu'}^{\nu}.
\end{equation*}
Thus $I_c$ and $I_f$ embed the two selected copies of $S^\nu$ into the same
intermediate module $\mathcal I$.

\smallskip
\noindent\emph{Step 4: the normalization bridge.}
We now check that no scalar is hidden in the change of association order.
In the orthogonal induced model, a color coset is a $q$-subset
$S\subseteq[n]$, recording the $\mu$-colored coordinates.  A forward
coordinate sector is a pair $(S,r)$ with $r\notin S$; tensoring by $e_r$
changes the color set to $S\cup\{r\}$.  Its stabilizer is exactly $K$, and
the identification of the direct sum of these sectors with $\mathcal I$ is
unitary: it only relabels orthonormal coset summands.

Let
\[
 j_\lambda:S^\lambda\longrightarrow
 \Ind_{S_{p-1}\times S_1}^{S_p}(S^{\lambda'}\boxtimes\one)
\]
be the isometric inner branching inclusion used in $J_c$.  For a coordinate
$r$ in the $\lambda$-colored block, let
$D_r:S^\lambda\to S^{\lambda'}$ be the normalized Young coisometry for
the subgroup fixing $r$, chosen equivariantly as $r$ varies, with
$D_r^*D_r=Q_r^{\lambda'}$.  Equivariance and multiplicity-free branching
show that the $r$-component of $j_\lambda$ is $\alpha D_r$, with one scalar
$\alpha$ independent of $r$.  Since $j_\lambda$ is an isometry,
\begin{align}
 \Id_{S^\lambda}
 &=\sum_{r=1}^p(\alpha D_r)^*(\alpha D_r)
   =|\alpha|^2\sum_{r=1}^pQ_r^{\lambda'}
   =|\alpha|^2\frac{p f^{\lambda'}}{f^\lambda}
      \Id_{S^\lambda}, \notag
\end{align}
where the last equality is \eqref{eq:plancherel-down-factor}.  Absorbing
the single phase of $\alpha$ into the equivariant family of coisometries
$D_r$, we may therefore take
\begin{equation}\label{eq:source-branch-component}
 \operatorname{pr}_r j_\lambda
 =\sqrt{\frac{f^\lambda}{p f^{\lambda'}}}\,D_r.
\end{equation}

For a fixed sector $(S,r)$, write $B_{S,r}$ for $D_r\otimes\Id_{S^\mu}$
together with the unitary relabeling of the color coset.  Equation
\eqref{eq:source-branch-component} says that the corresponding component of
$J_c$ is
\(
 \sqrt{f^\lambda/(p f^{\lambda'})}\,B_{S,r}
\).
On the other hand, contraction against the $e_r$-coordinate of $\ell_o$
contributes $n^{-1/2}$, and selection of the target $\mu'$-branch contributes
$J_f^*$.  Hence the direct contraction on this sector is
$n^{-1/2}J_f^*B_{S,r}$.  The same sector of
$\beta_{\lambda\to\lambda'}J_f^*J_c$ equals this map because
\begin{equation*}\label{eq:normalization-cancellation}
 \beta_{\lambda\to\lambda'}
 \sqrt{\frac{f^\lambda}{p f^{\lambda'}}}
 =\sqrt{\frac pn}\sqrt{\frac{f^{\lambda'}}{f^\lambda}}
  \sqrt{\frac{f^\lambda}{p f^{\lambda'}}}
 =\frac1{\sqrt n}.
\end{equation*}
Thus the two maps agree on every orthogonal coordinate sector.  Let
$C_{\lambda,\mu}^{\lambda',\mu'}$ denote the chosen isometric Pieri
intertwiner from the target ambient representation into the tensor product
of the natural signed representation with the source ambient representation.
Summing the orthogonal sectors and absorbing one common phase into this
intertwiner gives
\begin{equation*}\label{eq:induced-transfer-identity}
 \left(C_{\lambda,\mu}^{\lambda',\mu'}\right)^*
 (\ell_o\otimes\,\cdot\,)
 =\beta_{\lambda\to\lambda'}J_f^*J_c
 \quad\text{on the selected $(\lambda',\mu')$ branch}.
\end{equation*}
Equivalently, if
$\widehat T=\beta_{\lambda\to\lambda'}^{-1}
T_{\lambda,\mu}^{\lambda',\mu'}$, then
\begin{equation*}\label{eq:normalized-transfer-overlap}
 \widehat T=I_f^*I_c.
\end{equation*}
This is the required normalization bridge.

\smallskip
\noindent\emph{Step 5: evaluation by the unitary Racah coefficient.}
To evaluate the overlap, consider the multiplicity Hilbert space
\[
 \mathcal R=\Hom_{S_n}(\mathcal I,S^\nu),
 \qquad
 \ip{A}{B}_{\mathcal R}=\frac1{f^\nu}\tr(AB^*).
\]
The coisometries $I_c^*$ and $I_f^*$ are unit vectors in $\mathcal R$, and
Schur's lemma gives
\begin{equation*}\label{eq:transfer-HS-overlap}
 I_f^*I_c
 =\ip{I_f^*}{I_c^*}_{\mathcal R}\,\Id_{S^\nu}.
\end{equation*}
Unitary Frobenius reciprocity and the three-block form of
\eqref{eq:unitary-Schur-Weyl} identify $\mathcal R$, preserving the normalized
Hilbert--Schmidt inner product, with
\begin{equation*}\label{eq:triple-SW-multiplicity}
 \Hom_{SU(2)}\bigl(U_e,U_a\otimes U_{1/2}\otimes U_d\bigr).
\end{equation*}
Under this identification, $I_c$ and $I_f$ are the two normalized coupling
vectors
\[
 ((a\otimes\tfrac12)\to c)\otimes d\to e,
 \qquad
 a\otimes((\tfrac12\otimes d)\to f)\to e.
\]
Their overlap is the unitary Racah coefficient
\[
 (-1)^{a+d+e+1/2}\sqrt{(2c+1)(2f+1)}
 \begin{Bmatrix}a&\tfrac12&c\\ d&e&f\end{Bmatrix};
\]
see \cite[Chapter~9]{Varshalovich1988}.  Multiplying by
\eqref{eq:color-down-normalization} proves
\eqref{eq:unitary-recoupling-scalar}.

For the reverse transfer, interchange source and target.  The common induced
module and Racah overlap are unchanged up to adjoint, while the averaging
argument on the $q+1$ target-colored coordinates gives
$\sqrt{(q+1)/n}\sqrt{f^\mu/f^{\mu'}}$.  This proves the final assertion.
\end{proof}

\begin{theorem}
\label{thm:exact-p}
Let \(\lambda=(n-i-j,j)\), \(\mu=(i-\ell,\ell)\), and let
\(\lambda'\), \(\mu'\) be the target partitions for the edge
\((\sigma,\tau)\).  The squared forward contraction coefficient is
\begin{equation}\label{eq:p-forward}
 p^+_{\sigma\tau}(i,j,\ell;k)
 =\frac{n-i}{n}\,
  \frac{f^{\lambda'}}{f^\lambda}\,
  R_{\sigma\tau}^2.
\end{equation}
The squared reverse coefficient is
\begin{equation}\label{eq:p-reverse}
 p^-_{\sigma\tau}(i,j,\ell;k)
 =\frac{i+1}{n}\,
  \frac{f^\mu}{f^{\mu'}}\,
  R_{\sigma\tau}^2.
\end{equation}
Consequently, with
\begin{equation*}\label{eq:Domega}
 D_{i,j,\ell}=\binom ni
 f^{(n-i-j,j)}f^{(i-\ell,\ell)},
\end{equation*}
one has exact dimension balance
\begin{equation}\label{eq:balance-exact}
 D_{i,j,\ell}\,p^+_{\sigma\tau}
 =D_{i+1,j-\sigma,\ell+\tau}\,p^-_{\sigma\tau}.
\end{equation}
\end{theorem}

\begin{proof}
Apply Lemma~\ref{lem:unitary-recoupling} with \(p=n-i\) and square
\eqref{eq:unitary-recoupling-scalar}.  This gives
\eqref{eq:p-forward}.  The reverse statement follows from the last sentence
of that lemma and gives \eqref{eq:p-reverse}.

For completeness, the balance identity is an exact cancellation:
\begin{align*}
 D_{i,j,\ell}p^+_{\sigma\tau}
 &=\binom ni f^\lambda f^\mu
   \frac{n-i}{n}\frac{f^{\lambda'}}{f^\lambda}R^2
 =\frac{(n-1)!}{i!(n-i-1)!}f^{\lambda'}f^\mu R^2,\\
 D_{i+1,j-\sigma,\ell+\tau}p^-_{\sigma\tau}
 &=\binom n{i+1}f^{\lambda'}f^{\mu'}
   \frac{i+1}{n}\frac{f^\mu}{f^{\mu'}}R^2
 =\frac{(n-1)!}{i!(n-i-1)!}f^{\lambda'}f^\mu R^2.
\end{align*}
This proves \eqref{eq:balance-exact}.
\end{proof}

Theorem~\ref{thm:exact-p} supplies both ingredients required by the
abstract profile theorem: explicit directed probabilities and exact balance.
The full compatibility graph is moreover normalized at every vertex.

\begin{corollary}
\label{cor:normalization}
For a fixed valid vertex \(\omega\), sum the directed coefficients
\(p_{\omega,\omega'}\) over every valid box-transfer neighbor in the full
compatibility graph.  Then
\[
 \sum_{\omega'}p_{\omega,\omega'}=1.
\]
\end{corollary}

\begin{proof}
The vector \(\ell_o\otimes\phi_{\omega,o}u\) has norm \(\norm{u}\)
and stabilizer type \(S^{(n-k,k)}\).  The multiplicity-free Pieri
decomposition splits it orthogonally among precisely the valid neighboring
ambient irreducibles.  The squared norm of its component in the
\(\omega'\)-summand is
\(p_{\omega,\omega'}\norm{u}^2\), so Pythagoras gives the identity.
Equivalently, the forward contributions total \((n-i)/n\) by ordinary
Specht branching and Racah orthogonality, while the reverse contributions
total \(i/n\).
\end{proof}

\begin{remark}
Let \(P_{\omega,E}\) be the projector onto the unique copy of \(E\) in
\(V_\omega|_{S_n}\), and let \(Q_{\omega'}\) be the central projector
onto \(V_{\omega'}\) inside \(W\otimes V_\omega\).  The same directed
coefficient is characterized without choosing branching bases by
\[
 p_{\omega,\omega'}
 =\frac1{d_E}\operatorname{tr}
 \left[(|\ell_o\rangle\langle\ell_o|\otimes P_{\omega,E})
 Q_{\omega'}\right].
\]
This gives a basis-free characterization of the same coefficient and an
independent route for checking the closed formula in
Theorem~\ref{thm:exact-p}.
\end{remark}

\begin{remark}
Fix the external spins \(a,d,e\).  The two admissible values of \(c\) and the
two admissible values of \(f\) form a \(2\times2\) orthogonal Racah matrix;
after a consistent labeling, its squared entries are
\[
 \begin{pmatrix}Q&1-Q\\1-Q&Q\end{pmatrix}.
\]
Here
\begin{equation}\label{eq:Qexact}
 Q=\frac{(e+\tfrac12)^2-(a-d)^2}{(2a+1)(2d+1)}.
\end{equation}
For the four outgoing channels of a \emph{fixed source triple}, however,
\(a=a_\sigma\) in \eqref{eq:spins} changes with \(\sigma\).  Thus there need
not be one common finite value of \(Q\) for both rows.  More precisely,
\(R_{\sigma\sigma}^2=Q_\sigma\) and
\(R_{\sigma,1-\sigma}^2=1-Q_\sigma\), where \(Q_\sigma\) is the displayed
formula with \(a=a_\sigma\).  The two values differ by only \(O(1/n)\) in
the asymptotic regime and have the same limit.  Missing boundary channels are
handled by the usual zero-dimensional convention.
\end{remark}

\paragraph{Recovery of the whole-cube boundary path.}
When \(j=\ell=0\), the retained boundary edge has
\(\sigma=\tau=0\).  In this case \(f^\lambda=f^{\lambda'}=f^\mu=f^{\mu'}=1\)
and \eqref{eq:Qexact} gives
\[
 R_{00}^2
 =\frac{(i-k+1)(n-i-k)}{(i+1)(n-i)}.
\]
Therefore
\begin{equation*}\label{eq:old-directed-recovered}
 p^+_{00}=\frac{(i-k+1)(n-i-k)}{n(i+1)},
 \qquad
 p^-_{00}=\frac{(i-k+1)(n-i-k)}{n(n-i)}.
\end{equation*}
Their geometric mean is exactly the neighboring-degree entry of the
whole-cube matrix in \cite[Equation~(22)]{OpenAI2026}.  Thus the earlier
one-dimensional graph is the \(j=\ell=0\) boundary of the present
complete graph.  For \(k>0\), additional edges from that boundary enter
vertices with nonzero second rows.

\subsection{Finite honeycomb bounds}
\label{sec:finite}

The compatibility classification and exact edge weights now give a
finite-length coding theorem.  The vertex set may be any connected
truncation of the complete two-row graph, and the profile may be optimized
independently of that truncation.

Fix \(n,k\), and let \(\Omega\) be any finite set of valid triples that is
connected in the positive-weight box-transfer graph.  Retain a box-transfer
edge whenever both endpoints belong to \(\Omega\) and its contraction
coefficient is nonzero.  Define
\begin{equation*}\label{eq:Jfinite}
 (J_\Omega)_{\omega,\omega'}
 =\sqrt{p_{\omega,\omega'}p_{\omega',\omega}},
\end{equation*}
using Theorem~\ref{thm:exact-p}.  The moving stabilizer rank is
\begin{equation}\label{eq:stabilizer-dim}
 d_k=f^{(n-k,k)}=\binom nk-\binom n{k-1},
 \qquad \binom n{-1}=0.
\end{equation}
For \(\omega=(i,j,\ell)\), abbreviate
\begin{equation}\label{eq:D-hyp-abbrev}
 D_\omega=\binom ni f^{(n-i-j,j)}f^{(i-\ell,\ell)}.
\end{equation}

\begin{theorem}
\label{thm:finite-hyp}
Let \(1\le d\le n/2\) and set \(s=1-2d/n\).  For every positive profile
\(w=(w_\omega)_{\omega\in\Omega}\), put
\begin{equation*}\label{eq:finite-lambda-Deff}
 \lambda_\Omega(w)=
 \min_{\nu\in\Omega}
 \frac{\sum_{\omega\in\Omega}p_{\omega,\nu}w_\omega}{w_\nu},
 \qquad
 \mathcal D_\Omega(w)=
 \frac{(\sum_\omega w_\omega)^2}
      {\sum_\omega w_\omega^2/D_\omega}.
\end{equation*}
If \(\lambda_\Omega(w)>s\), then
\begin{equation}\label{eq:finite-profile-hyp}
 \Atwo(n,d)
 \le
 \frac{1-s}{d_k\bigl(\lambda_\Omega(w)-s\bigr)}
 \mathcal D_\Omega(w).
\end{equation}
In particular, if \(v\) is the positive unit Perron vector of \(J_\Omega\)
and \(\Lambda_\Omega>s\), then
\begin{align}
 \Atwo(n,d)
 &\le
 \frac{1-s}{d_k(\Lambda_\Omega-s)}
 \left(\sum_{\omega\in\Omega}\sqrt{D_\omega}\,v_\omega\right)^2
                                                               \label{eq:finite-perron-hyp}\\
 &\le
 \frac{1-s}{d_k(\Lambda_\Omega-s)}
 \sum_{\omega\in\Omega}D_\omega.                            \label{eq:finite-sum-hyp}
\end{align}
\end{theorem}

\begin{proof}
Apply Theorem~\ref{thm:general-graph} with
\(G=B_n\), \(H=S_n\), \(W=\C^n\), and
\(E=S^{(n-k,k)}\).  Proposition~\ref{prop:complete-region} gives
multiplicity one.  Lemma~\ref{lem:Pieri} gives the coordinate graph,
Theorem~\ref{thm:exact-p} gives its contraction probabilities and balance,
and \eqref{eq:hamming-inner} converts minimum Hamming distance to the
threshold \(s\).  The dimensions are \eqref{eq:D-hyp-abbrev} and
\eqref{eq:stabilizer-dim}.  The Perron statements are
Corollary~\ref{cor:perron-profile}.
\end{proof}

At fixed \(n\), Theorem~\ref{thm:finite-hyp} is algorithmic: enumerate the
valid triples, assemble any desired connected truncation from the explicit
factorial and \(6j\) formulas, and either use the Perron vector in
\eqref{eq:finite-perron-hyp} or optimize the full profile objective in
\eqref{eq:finite-profile-hyp}.  Restricting to the boundary triples
\((i,0,0)\) recovers the earlier whole-cube graph.  Enlarging the retained
graph and optimizing the profile are separate sources of improvement.

For the clean continuum formula below, we use the dimension-sum
specialization~\eqref{eq:finite-sum-hyp}.  The stronger profile optimization
remains available at finite length and may also lead to sharper asymptotic
variants, but it is not needed for the honeycomb exponent defined in
\eqref{eq:kappa-hyp}.

\section{The asymptotic honeycomb formula}
\label{sec:asymptotics}

We pass from the finite lattice of valid triples to a local continuum model.
The four normalized parameters record, respectively, the color split, the
second-row proportions in the two colors, and the moving stabilizer degree.
Let valid triples scale as
\begin{equation}\label{eq:scaling}
 \frac in\to x,
 \qquad \frac{j}{n-i}\to\eta,
 \qquad \frac{\ell}{i}\to\rho,
 \qquad \frac kn\to z.
\end{equation}
The doubled spins in \eqref{eq:doubled-spins-finite}, divided by \(n\),
converge to the quantities \(A,B,E\) in \eqref{eq:ABE}.  Hence the finite
Littlewood--Richardson triangle converges to
\eqref{eq:triangle-continuum}.

For a two-row Specht module,
\begin{equation}\label{eq:two-row-dim}
 f^{(m-r,r)}=\binom mr-\binom m{r-1}.
\end{equation}
The two removal ratios are
\begin{align}
 \frac{f^{(m-1-r,r)}}{f^{(m-r,r)}}
 &=\frac{(m-r+1)(m-2r)}{m(m-2r+1)},                       \label{eq:ratio-first}\\
 \frac{f^{(m-r,r-1)}}{f^{(m-r,r)}}
 &=\frac{r(m-2r+2)}{m(m-2r+1)}.                          \label{eq:ratio-second}
\end{align}
Thus they converge to \(1-r/m\) and \(r/m\), respectively.

The half-spin Racah formula \eqref{eq:Qexact} gives, uniformly on compact
subsets of the strict triangle region (Appendix~\ref{app:racah}),
\begin{equation}\label{eq:R-limit}
 R_{\sigma\tau}^2\longrightarrow
 \begin{cases}
 P,&\sigma=\tau,\\
 1-P,&\sigma\ne\tau,
 \end{cases}
\end{equation}
where \(P\) is \eqref{eq:Pdef}.  Combining
\eqref{eq:p-forward}--\eqref{eq:p-reverse} with
\eqref{eq:ratio-first}--\eqref{eq:R-limit}, the symmetric edge of type
\((\sigma,\tau)\) converges to
\begin{equation}\label{eq:four-c}
 \begin{array}{ll}
 c_{00}=P\sqrt{x(1-x)(1-\eta)(1-\rho)},
 &c_{01}=(1-P)\sqrt{x(1-x)(1-\eta)\rho},\\[1mm]
 c_{10}=(1-P)\sqrt{x(1-x)\eta(1-\rho)},
 &c_{11}=P\sqrt{x(1-x)\eta\rho}.
 \end{array}
\end{equation}
Their doubled sum is exactly \(\Gamma_{\mathrm{HC}}\):
\begin{equation}\label{eq:Gamma-sum}
 2(c_{00}+c_{01}+c_{10}+c_{11})
 =\Gamma_{\mathrm{HC}}(x,\eta,\rho,z).
\end{equation}

\subsection{F{\o}lner boxes and the local spectral limit}

The boxes constructed below are F{\o}lner in the sense of
\cite{Folner1955}: their boundary-to-volume ratio tends to zero, so a
bounded lattice shift displaces only a vanishing fraction of each box.
Let the four forward lattice directions be
\begin{equation}\label{eq:lattice-directions}
 q_{00}=(1,0,0),\qquad q_{01}=(1,0,1),\qquad
 q_{10}=(1,-1,0),\qquad q_{11}=(1,-1,1).
\end{equation}
The first three form a unimodular basis of \(\Z^3\), and
\(q_{11}=q_{01}+q_{10}-q_{00}\).  Thus all four edge types act as bounded
lattice shifts in the box coordinates used below.

\begin{lemma}
\label{lem:folner-box}
Fix a parameter point
\(\mathbf q=(x,\eta,\rho,z)\) satisfying
\begin{equation}\label{eq:strict-hyp-interior}
 0<x<\frac12,\quad 0<\eta,\rho<\frac12,\quad
 |A-B|<E<A+B,
\end{equation}
and let \(k_n/n\to z\).  There are connected sets of valid triples
\(\Omega_n\) such that
\begin{align}
 \liminf_{n\to\infty}\lambda_{\max}(J_{\Omega_n})
 &\ge \Gamma_{\mathrm{HC}}(\mathbf q),                    \label{eq:lambda-limit-lower}\\
 \frac1n\log_2\frac{\sum_{\omega\in\Omega_n}D_\omega}
                         {f^{(n-k_n,k_n)}}
 &\longrightarrow \Phi_{\mathrm{HC}}(\mathbf q).         \label{eq:box-dimension-limit}
\end{align}
\end{lemma}

\begin{proof}
\smallskip
\noindent\emph{Step 1: choose an interior lattice box.}
For all sufficiently large \(n\), choose the explicit integer centers
\[
 i_n=\lfloor xn\rfloor,\qquad
 j_n=\lfloor\eta(n-i_n)\rfloor,\qquad
 \ell_n=\lfloor\rho i_n\rfloor,
 \qquad \omega_n^0=(i_n,j_n,\ell_n).
\]
Since \(k_n/n\to z\), every normalized row-length and triangle inequality
converges to its strict counterpart in \eqref{eq:strict-hyp-interior}.
Consequently these centers are valid for all sufficiently large \(n\), and
all such inequalities hold there with margin at least \(c n\) for some
constant \(c>0\).  Let
\(m_n\to\infty\) with \(m_n=o(n)\), and define
\begin{equation*}\label{eq:explicit-folner-box}
 \Omega_n=\left\{\omega_n^0+r_0q_{00}+r_1q_{01}+r_2q_{10}:
                  -m_n\le r_0,r_1,r_2\le m_n\right\}.
\end{equation*}
For all sufficiently large \(n\), every point of \(\Omega_n\) and every
retained edge between two such points is valid.  Since
\(q_{00},q_{01},q_{10}\) form a \(\Z^3\)-basis, the box is connected and
\(|\Omega_n|=(2m_n+1)^3\).

\smallskip
\noindent\emph{Step 2: control the edge weights uniformly.}
The exact formulas of Theorem~\ref{thm:exact-p}, together with
\eqref{eq:ratio-first}--\eqref{eq:R-limit}, converge uniformly throughout
\(\Omega_n\), because all normalized parameters stay in a fixed compact
subset of the strict region.  Hence an edge in direction \(q_{\sigma\tau}\)
has weight
\[
 (J_{\Omega_n})_{\omega,\omega+q_{\sigma\tau}}
 =c_{\sigma\tau}+o(1)
\]
uniformly in its initial vertex.  A bounded lattice shift removes only
\(O(m_n^2)\) points from a three-dimensional box.  Therefore the constant
vector satisfies
\begin{align*}
 \frac{\langle\one,J_{\Omega_n}\one\rangle}
      {\langle\one,\one\rangle}
 &=\frac2{|\Omega_n|}
   \sum_{\sigma,\tau\in\{0,1\}}
   \sum_{\substack{\omega\in\Omega_n\\
                    \omega+q_{\sigma\tau}\in\Omega_n}}
   (J_{\Omega_n})_{\omega,\omega+q_{\sigma\tau}}\\
 &=2(c_{00}+c_{01}+c_{10}+c_{11})+o(1)
  =\Gamma_{\mathrm{HC}}(\mathbf q)+o(1).
\end{align*}
Rayleigh--Ritz proves \eqref{eq:lambda-limit-lower}.

\smallskip
\noindent\emph{Step 3: compute the dimension exponent.}
Finally, every normalized parameter in the box differs from its central
value by \(o(1)\).  Stirling's formula and the two-row hook formula therefore
give, uniformly in \(\omega\in\Omega_n\),
\[
 \frac1n\log_2 D_\omega
 =\Htwo(x)+(1-x)\Htwo(\eta)+x\Htwo(\rho)+o(1).
\]
The box contains only \(\exp(o(n))\) vertices, while
\(n^{-1}\log_2 f^{(n-k_n,k_n)}\to\Htwo(z)\).  This proves
\eqref{eq:box-dimension-limit}.
\end{proof}

\begin{lemma}
\label{lem:boundary-approximation}
Every point of \(\cD_{\mathrm{HC}}\) in
\eqref{eq:parameters-domain}--\eqref{eq:triangle-continuum} is a limit of
points satisfying \eqref{eq:strict-hyp-interior}.  Moreover, if
\(\mathbf q\in\cD_{\mathrm{HC}}\) and
\(\Gamma_{\mathrm{HC}}(\mathbf q)>s\), then there are strict-interior
points \(\mathbf q_r\to\mathbf q\) with
\(\Gamma_{\mathrm{HC}}(\mathbf q_r)>s\) and
\(\Phi_{\mathrm{HC}}(\mathbf q_r)\to
  \Phi_{\mathrm{HC}}(\mathbf q)\).
\end{lemma}

\begin{proof}
Fix $\mathbf q=(x,\eta,\rho,z)\in\cD_{\mathrm{HC}}$.  Since $A,B>0$, there
is a unique $\theta\in[0,1]$ such that
\begin{equation*}\label{eq:boundary-theta}
 E=(1-\theta)|A-B|+\theta(A+B).
\end{equation*}
For $r\downarrow0$, set
\[
 x_r=(1-r)x+\frac r4,\qquad
 \eta_r=(1-r)\eta+\frac r4,\qquad
 \rho_r=(1-r)\rho+\frac r4.
\]
Then $0<x_r<1/2$ and $0<\eta_r,\rho_r<1/2$.  Let $A_r,B_r$ be the
corresponding values, put $L_r=|A_r-B_r|$ and $U_r=A_r+B_r$, and define
\[
 \theta_r=(1-r)\theta+\frac r2,
 \qquad
 E_r=(1-\theta_r)L_r+\theta_r U_r,
 \qquad
 z_r=\frac{1-E_r}{2}.
\]
For every $r>0$, one has $0<\theta_r<1$, hence $L_r<E_r<U_r$.  Moreover,
$\eta_r,\rho_r>0$ implies
\[
 U_r=1-2\bigl((1-x_r)\eta_r+x_r\rho_r\bigr)<1.
\]
Thus $0<E_r<1$ and $0<z_r<1/2$.  Therefore
$\mathbf q_r=(x_r,\eta_r,\rho_r,z_r)$ satisfies
\eqref{eq:strict-hyp-interior}, and $\mathbf q_r\to\mathbf q$.

The formulas for $\Gamma_{\mathrm{HC}}$ and $\Phi_{\mathrm{HC}}$ are
continuous throughout the stated domain, including all ordinary boundary
faces allowed by \eqref{eq:parameters-domain}--\eqref{eq:triangle-continuum}.
If
$\Gamma_{\mathrm{HC}}(\mathbf q)>s$, that strict inequality persists for all
sufficiently small $r$, and the objectives converge as claimed.
\end{proof}

\begin{remark}\label{rem:equal-row-limits}
The variational formula optimizes over the nondegenerate bulk scaling
$\eta,\rho<1/2$ and over ordinary boundary points of that domain.  This is
exactly the regime realized by the F{\o}lner boxes in
Lemma~\ref{lem:folner-box}.  A different phenomenon can occur if, for
example, $\eta\to1/2$ while the row difference $(n-i)-2j$ remains bounded.
The exact removal ratios \eqref{eq:ratio-first}--\eqref{eq:ratio-second} then
retain that bounded difference as an additional parameter, rather than
converging only to the bulk values $1/2$ and $1/2$.  Such boundary-layer
families are not identified solely by the normalized quadruple
$(x,\eta,\rho,z)$.

No such family is needed for Theorem~\ref{thm:asymptotic-proof}: every point
in the definition of $\kappa_{\mathrm{HC}}$ is realized by nondegenerate
bulk boxes, with ordinary boundary points supplied by
Lemma~\ref{lem:boundary-approximation}.  Omitting extra equal-row boundary
layers can only forgo possible further improvements; it cannot invalidate the
honeycomb bound.  Accordingly, ``complete'' refers to the finite two-row
representation family, not to an optimization over every possible
boundary-layer scaling.
\end{remark}

\subsection{The entropy potential and the asymptotic bound}

For later reference, we isolate the dimension calculation already used in
Lemma~\ref{lem:folner-box}.  Stirling's formula and
\eqref{eq:two-row-dim} give, as computed in Appendix~\ref{app:entropy},
\begin{align}
 \frac1n\log_2 D_{i,j,\ell}
 &\longrightarrow
 \Htwo(x)+(1-x)\Htwo(\eta)+x\Htwo(\rho),                \label{eq:ambient-exp}\\
 \frac1n\log_2 d_k
 &\longrightarrow \Htwo(z).                             \label{eq:stabilizer-exp}
\end{align}
The box has subexponentially many vertices and all normalized parameters vary
by \(o(1)\), so
\begin{equation}\label{eq:ratio-exp}
 \frac1n\log_2
 \frac{\sum_{\omega\in\Omega_n}D_\omega}{d_k}
 \longrightarrow \Phi_{\mathrm{HC}}(x,\eta,\rho,z).
\end{equation}

\begin{corollary}
\label{cor:Phi-nonnegative}
For every $(x,\eta,\rho,z)\in\cD_{\mathrm{HC}}$,
\[
 \Phi_{\mathrm{HC}}(x,\eta,\rho,z)\ge0.
\]
\end{corollary}

\begin{proof}
First take a strict-interior point and the compatible triples from
Lemma~\ref{lem:folner-box}.  Since $S^{(n-k_n,k_n)}$ occurs in the restriction
of every retained $V_\omega$, one has
$f^{(n-k_n,k_n)}\le D_\omega\le\sum_{\omega'\in\Omega_n}D_{\omega'}$.
Divide logarithms by $n$ and use \eqref{eq:ratio-exp}.  General boundary
points follow from Lemma~\ref{lem:boundary-approximation} and continuity.
\end{proof}

\begin{theorem}
\label{thm:asymptotic-proof}
For every \(0<\delta<1/2\),
\[
 R_2(\delta)\le\kappa_{\mathrm{HC}}(\delta).
\]
\end{theorem}

\begin{proof}
Fix a strict-interior parameter point with
\(\Gamma_{\mathrm{HC}}>1-2\delta\).  Apply
Lemma~\ref{lem:folner-box} with \(k_n/n\to z\), let
\(d_n=\lceil\delta n\rceil\), and set \(s_n=1-2d_n/n\).
The spectral gap
\(\lambda_{\max}(J_{\Omega_n})-s_n\) is bounded below by a positive
constant for all sufficiently large \(n\).  The dimension-sum specialization
\eqref{eq:finite-sum-hyp} of Theorem~\ref{thm:finite-hyp}, together with
\eqref{eq:box-dimension-limit}, therefore gives
\[
 R_2(\delta)\le\Phi_{\mathrm{HC}}(x,\eta,\rho,z),
\]
since the remaining finite prefactor is subexponential.  Every feasible
boundary point with strict spectral inequality is handled by
Lemma~\ref{lem:boundary-approximation}.  Taking the infimum proves the
theorem.
\end{proof}

\section{Strict improvement and the combined binary exponent}
\label{sec:strict}

This section proves the two strict comparisons announced in
Theorem~\ref{thm:main-asymptotic}, identifies the symmetric slice with the
fully optimized second MRRW problem, and combines the resulting branch with
the constant-weight bound.

\subsection{A one-sided face already gives strict improvement}

We first prove strict improvement on a particularly simple face of the
honeycomb domain.  Only one second row must be opened: we keep
\(\eta=0\) and take \(\rho>0\).  The new cross edge has strength
\(\Theta(\sqrt\rho)\), while its dimension cost is
\(x\Htwo(\rho)=o(\sqrt\rho)\).

Set \(\eta=\rho=0\) and rename \(x=a\), \(z=b\).  Then
\[
 A=1-a,\qquad B=a,\qquad E=1-2b,
\]
and
\begin{equation}\label{eq:P-boundary}
 P_0=\frac{(a-b)(1-a-b)}{a(1-a)}.
\end{equation}
Equations~\eqref{eq:Gamma-hyp} and \eqref{eq:Phi-hyp} reduce to
\begin{equation*}\label{eq:boundary-identification}
 \Gamma_{\mathrm{HC}}(a,0,0,b)=\Gamma_H(a,b),
 \qquad
 \Phi_{\mathrm{HC}}(a,0,0,b)=\Htwo(a)-\Htwo(b).
\end{equation*}
Thus \(\kappa_{\mathrm{HC}}\le\kappa_H\) immediately.  The issue is
strictness on the smaller face \(\eta=0\).

\begin{lemma}
\label{lem:minimizer}
Fix \(0<\delta<1/2\), put \(s=1-2\delta\), and minimize
\(\Htwo(a)-\Htwo(b)\) over the compact closed feasible region
\[
 0<a\le\frac12,\qquad 0\le b\le a,
 \qquad \Gamma_H(a,b)\ge s.
\]
The minimum equals \(\kappa_H(\delta)\).  Every minimizer \((a,b)\) satisfies
\begin{equation*}\label{eq:minimizer-properties}
 0<b<a\le\frac12,
 \qquad \Gamma_H(a,b)=s.
\end{equation*}
\end{lemma}

\begin{proof}
The strict and closed infima agree by approximation from the feasible
interior.  Put
\[
 a_0=\frac12-\sqrt{\delta(1-\delta)}>0.
\]
Since \(\Gamma_H(a,b)\le 2\sqrt{a(1-a)}\), feasibility implies
\(a\ge a_0\).  The feasible region is therefore a closed subset of
\([a_0,1/2]\times[0,1/2]\) and is compact.  On $b=0$, the same inequality forces $a\ge a_0$, so the least
objective there is $\Htwo(a_0)=M_1(\delta)$.  We now show directly that the
closed optimum is strictly smaller.  At $(a_0,0)$ one has
$\Gamma_H(a_0,0)=s$.  For a constant
\[
 c>\frac{2}{1-2a_0}
\]
and sufficiently small $b>0$, set $a=a_0+cb$.  Differentiating at
$(a_0,0)$ gives
\[
 \left.\frac{d}{db}\Gamma_H(a_0+cb,b)\right|_{b=0}
 =\frac{c(1-2a_0)-2}{\sqrt{a_0(1-a_0)}}>0.
\]
Hence these points are strictly feasible.  On the other hand,
\[
 \Htwo(a_0+cb)-\Htwo(b)
 =\Htwo(a_0)+O(b)-b\log_2(1/b)<\Htwo(a_0)
\]
for all sufficiently small $b$.  Thus no minimizer lies on $b=0$.  It cannot
have $b=a$, where $\Gamma_H=0<s$.  Finally, if the spectral inequality were
strict, increasing $b$ slightly would preserve feasibility and strictly
decrease the objective.
\end{proof}

\begin{theorem}
\label{thm:strict}
For every \(0<\delta<1/2\),
\begin{equation*}\label{eq:strict-face}
 \kappa_{\mathrm{HC}}(\delta)
 \le\kappa_{\mathrm{face}}(\delta)
 <\kappa_H(\delta).
\end{equation*}
\end{theorem}

\begin{proof}
Let \((a,b)\) be a minimizer from Lemma~\ref{lem:minimizer}, and let
\(P_0\) be \eqref{eq:P-boundary}.  Since
\begin{equation*}\label{eq:one-minus-P}
 1-P_0=\frac{b(1-b)}{a(1-a)}>0,
\end{equation*}
there is a genuinely missing cross channel at the one-row boundary.

Set
\begin{equation*}\label{eq:one-sided-perturb}
 x=a,\qquad \eta=0,\qquad \rho=\eps,
 \qquad z_\eps=b+c\sqrt\eps,
\end{equation*}
where \(c>0\) will be fixed.  At fixed \(z=b\),
\(P(a,0,\eps,b)=P_0+O(\eps)\), and therefore
\begin{align}
 \Gamma_{\mathrm{HC}}(a,0,\eps,b)
 &=2\sqrt{a(1-a)}
   \left[P\sqrt{1-\eps}+(1-P)\sqrt\eps\right]\nonumber\\
 &=s+2\sqrt{a(1-a)}(1-P_0)\sqrt\eps+O(\eps). \notag
\end{align}
On the whole-cube boundary,
\begin{equation}\label{eq:dGamma-db}
 \partial_b\Gamma_H(a,b)
 =-\frac{2(1-2b)}{\sqrt{a(1-a)}}.
\end{equation}
Consequently
\begin{align}
 \Gamma_{\mathrm{HC}}(a,0,\eps,z_\eps)
 =s+\left[
 2\sqrt{a(1-a)}(1-P_0)
 -\frac{2c(1-2b)}{\sqrt{a(1-a)}}
 \right]\sqrt\eps+O(\eps).                              \label{eq:one-sided-Gamma}
\end{align}
Choose
\begin{equation*}\label{eq:one-sided-c}
 0<c<c_0\defeq
 \frac{a(1-a)(1-P_0)}{1-2b}
 =\frac{b(1-b)}{1-2b}.
\end{equation*}
Then \eqref{eq:one-sided-Gamma} is strictly larger than \(s\) for all
sufficiently small \(\eps\).

The exponent at this strict witness is
\begin{align}
 \Phi_{\mathrm{HC}}(a,0,\eps,z_\eps)
 &=\Htwo(a)+a\Htwo(\eps)-\Htwo(b+c\sqrt\eps)\nonumber\\
 &=\Htwo(a)-\Htwo(b)
   -c\Htwo'(b)\sqrt\eps+o(\sqrt\eps),  \notag
\end{align}
where \(a\Htwo(\eps)=o(\sqrt\eps)\) and
\(\Htwo'(b)=\log_2((1-b)/b)>0\).  Thus the value is strictly below
\(\kappa_H(\delta)\).

It remains to place the improvement on the exact two-parameter boundary
\eqref{eq:kappa-face}.  For fixed \((a,\eps)\) with \(\eta=0\),
\(\Gamma_{\mathrm{HC}}\) is strictly decreasing in \(z\) as long as
\(0<P<1\), because \(S_+>S_-\) and \(P\) is strictly decreasing in \(z\).
Starting from the strict point above, increase \(z\).  When \(P=0\), the
spectral value is \(2\sqrt{a(1-a)\eps}<s\) for all sufficiently small
\(\eps\), so continuity gives a unique intervening point with
\(\Gamma_{\mathrm{HC}}=s\).  This boundary point has an even smaller
objective, since \(\Htwo\) is increasing on \([0,1/2]\).  It is
exactly parameterized by \eqref{eq:P-face}--\eqref{eq:z-face}.  Hence
\(\kappa_{\mathrm{face}}(\delta)<\kappa_H(\delta)\), and the inclusion of
the face in the closure of the complete domain gives
\(\kappa_{\mathrm{HC}}\le\kappa_{\mathrm{face}}\).
\end{proof}

\begin{remark}
The proof isolates the mechanism more sharply than a symmetric perturbation:
one box-transfer channel of weight \((1-P_0)\sqrt\rho\) already produces a
spectral gain of order \(\sqrt\rho\), whereas opening the corresponding
second row costs only \(\rho\log(1/\rho)\) in the exponent.  This is why the
strict improvement survives on a two-dimensional face of the full
four-parameter region.
\end{remark}

The angle form \eqref{eq:Gamma-angle} reveals two facts that are invisible in
the one-sided parametrization.  First, the symmetric line is exactly the full
second MRRW construction.  Second, once the whole-cube boundary is
perturbed, the entropy-optimal division of a small total angle between the
two ambient partitions is proportional to their sizes.

\subsection{Exact embedding of the fully optimized second MRRW bound}

\begin{proposition}
\label{prop:MRRW-embedding}
Fix \(0<\delta<1/2\), put \(s=1-2\delta\), and restrict the balanced family
to \(x=1/2\).  Write \(\tau=\sin t\).  For
\(0\le\tau\le s\), its spectral boundary is
\begin{equation}\label{eq:MRRW-slice-P-E}
 P=\frac{s-\tau}{1-\tau},
 \qquad E^2=(1+\tau)(s-\tau),
\end{equation}
and its objective is
\begin{equation}\label{eq:MRRW-slice-objective}
 \Phi_{\mathrm{bal}}\!\left(\frac12,t;\delta\right)
 =F_\delta(\tau).
\end{equation}
Consequently, the infimum of the symmetric honeycomb slice is exactly
\(M_2(\delta)\), and
\begin{equation*}\label{eq:kbal-le-M2}
 \kappa_{\mathrm{HC}}(\delta)
 \le\kappa_{\mathrm{bal}}(\delta)\le M_2(\delta).
\end{equation*}
\end{proposition}

\begin{proof}
At \(x=1/2\), the balanced choice is \(u=v=t/2\), so
\(\eta=\rho=\sin^2(t/2)\).  Equation~\eqref{eq:Gamma-angle} becomes
\[
 \Gamma_{\mathrm{HC}}=P+(1-P)\sin t.
\]
Solving \(\Gamma_{\mathrm{HC}}=s\) gives the first formula in
\eqref{eq:MRRW-slice-P-E}.  Moreover,
\(A=B=\tfrac12\cos t\), and hence
\[
 E^2=4A^2P=(1-\tau^2)\frac{s-\tau}{1-\tau}
     =(1+\tau)(s-\tau).
\]
Now
\[
 \Htwo(\sin^2(t/2))=g(\tau^2),
 \qquad
 1-E^2=\tau^2+2\delta\tau+2\delta,
\]
so \(\Htwo((1-E)/2)=g(1-E^2)\).  Substitution into
\eqref{eq:Phi-balanced} gives \eqref{eq:MRRW-slice-objective}.

For completeness, consider the entire symmetric slice $x=1/2$,
$\eta=\rho=\sin^2(t/2)$, without imposing the balanced boundary value of
$P$.  For arbitrary $P\in[0,1]$,
\begin{equation}\label{eq:symmetric-general-P}
 E^2=(1-\tau^2)P,
 \qquad
 \Phi_{\mathrm{HC}}
 =1+g(\tau^2)-g\bigl(1-(1-\tau^2)P\bigr),
\end{equation}
and the spectral constraint is
$P+(1-P)\tau\ge s$.  If $0\le\tau<s$, the objective in
\eqref{eq:symmetric-general-P} is increasing in $P$, while feasibility
requires $P\ge(s-\tau)/(1-\tau)$.  Its infimum is therefore the boundary
value $F_\delta(\tau)$.  If $\tau\ge s$, then $P=0$ is feasible after an
arbitrarily small strict perturbation and gives
$g(\tau^2)\ge g(s^2)=M_1(\delta)\ge M_2(\delta)$.
Thus no point outside the displayed boundary segment
improves $M_2$.  The endpoints follow by limits from the open balanced domain,
and boundary points are approached by strict spectral witnesses by decreasing
$z$.
\end{proof}

\subsection{The entropy-balanced perturbation}

\begin{theorem}
\label{thm:balanced-strict}
For every \(0<\delta<1/2\),
\begin{equation}\label{eq:balanced-strict}
 \kappa_{\mathrm{HC}}(\delta)
 \le\kappa_{\mathrm{bal}}(\delta)<\kappa_H(\delta).
\end{equation}
More precisely, let \((a,b)\) be any minimizing whole-cube point from
Lemma~\ref{lem:minimizer}, and put
\begin{equation*}\label{eq:q-balanced}
 q\defeq\frac{a(1-a)(1-P_0)}{1-2b}
   =\frac{b(1-b)}{1-2b}>0,
\end{equation*}
where \(P_0\) is \eqref{eq:P-boundary}.  On the exact balanced spectral
boundary with
\begin{equation}\label{eq:balanced-local-ray}
 x=a,
 \qquad u=at,
 \qquad v=(1-a)t,
\end{equation}
there is a smooth function \(z(t)\) such that
\begin{align}
 z(t)&=b+qt+O(t^2),                                       \label{eq:z-balanced-expansion}\\
 \Phi_{\mathrm{bal}}(a,t;\delta)
 &=\kappa_H(\delta)-q\Htwo'(b)t
   +2a(1-a)t^2\log_2(1/t)+O(t^2). \notag
\end{align}
\end{theorem}

\begin{proof}
Let \(G(t,z)\) denote \(\Gamma_{\mathrm{HC}}\) at the parameters in
\eqref{eq:balanced-local-ray}.  At \(t=0\), this is the whole-cube
boundary, so \(G(0,b)=s\).  The same-row and cross-row factors in
\eqref{eq:Gamma-angle} have the expansions
\[
 \sqrt{a(1-a)}\cos((1-2a)t)
   =\sqrt{a(1-a)}+O(t^2),
 \qquad
 \sqrt{a(1-a)}\sin t
   =\sqrt{a(1-a)}t+O(t^3).
\]
The quantities \(A\), \(B\), and hence \(P\) at fixed \(z=b\), change only
by \(O(t^2)\).  Therefore
\begin{equation*}\label{eq:balanced-first-spectral-gain}
 \partial_tG(0,b)=2\sqrt{a(1-a)}(1-P_0).
\end{equation*}
On the whole-cube boundary,
\[
 \partial_zG(0,b)
 =-\frac{2(1-2b)}{\sqrt{a(1-a)}}<0
\]
by \eqref{eq:dGamma-db}.  The implicit-function theorem supplies a unique
smooth spectral-boundary function \(z(t)\), and
\[
 z'(0)=-\frac{\partial_tG(0,b)}{\partial_zG(0,b)}
 =\frac{a(1-a)(1-P_0)}{1-2b}=q.
\]
This proves \eqref{eq:z-balanced-expansion}.

For \(c>0\),
\begin{equation}\label{eq:H-small-angle}
 \Htwo(\sin^2(ct))
 =2c^2t^2\log_2(1/t)+O(t^2).
\end{equation}
Thus the two newly opened rows contribute
\[
 2\bigl((1-a)a^2+a(1-a)^2\bigr)t^2\log_2(1/t)+O(t^2)
 =2a(1-a)t^2\log_2(1/t)+O(t^2).
\]
Taylor expansion of \(-\Htwo(z(t))\) gives the linear term
\(-q\Htwo'(b)t\).  Since \(0<b<1/2\),
\(\Htwo'(b)=\log_2((1-b)/b)>0\), and the negative linear term dominates the
entropy cost.  Hence the boundary objective is strictly below
\(\kappa_H(\delta)\) for all sufficiently small positive \(t\).  Decreasing
\(z(t)\) by an arbitrarily smaller amount makes the spectral inequality
strict without destroying the objective gap.  This proves
\eqref{eq:balanced-strict}.
\end{proof}

\begin{proposition}
\label{prop:balanced-optimality}
Fix the same whole-cube minimizer \((a,b)\).  More generally, let
\begin{equation*}\label{eq:general-angle-ray}
 u=\theta t,
 \qquad v=(1-\theta)t,
 \qquad 0\le\theta\le1,
\end{equation*}
and choose \(z_\theta(t)\) on the exact spectral boundary.  Then
\begin{align}
 z_\theta(t)&=b+qt+O(t^2),                                \label{eq:z-theta}\\
 \Phi_\theta(t)
 &=\kappa_H(\delta)-q\Htwo'(b)t
   +2C_a(\theta)t^2\log_2(1/t)+O(t^2),                   \label{eq:Phi-theta}
\end{align}
where
\begin{equation*}\label{eq:C-theta}
 C_a(\theta)=(1-a)\theta^2+a(1-\theta)^2.
\end{equation*}
The first-order spectral gain is independent of \(\theta\), while the
logarithmic entropy coefficient is uniquely minimized at \(\theta=a\), where
\(C_a(a)=a(1-a)\).  Thus the balanced ray
\eqref{eq:balanced-local-ray} is second-order optimal among all fixed-ratio
small two-sided perturbations with the same total angle.
\end{proposition}

\begin{proof}
Since \(u+v=t\), the cross-row factor is always \(\sin t\); the same-row
factor \(\cos(u-v)\), and the triangle variables \(A,B\), have no linear
term at \(t=0\).  The calculation of \(z_\theta'(0)\) is therefore identical
to that in Theorem~\ref{thm:balanced-strict}, proving
\eqref{eq:z-theta}.  Applying \eqref{eq:H-small-angle} to the two new rows
gives \eqref{eq:Phi-theta}.  Finally,
\[
 C_a(\theta)=a(1-a)+(\theta-a)^2,
\]
which proves the unique minimization claim.
\end{proof}
A second-variation calculation around the symmetric
\(M_2\) slice, useful for interpreting the numerical optimizer but unnecessary
for the rate theorem, is deferred to Appendix~\ref{app:mrrw-stability}.

\subsection{Combining the honeycomb and constant-weight branches}
\label{sec:combination}

The honeycomb construction strengthens the whole-cube component of the
previous moving-projection bound of \cite[Chapter~2]{OpenAI2026}; the
constant-weight component is left unchanged.  We recall that component here
both to define the combined exponent and to make the comparison in
Appendix~\ref{app:channel-comparison} self-contained.  Its parameters satisfy
\begin{align}
 \frac\delta2<\alpha<\frac12,
 &\qquad 0\le\beta<\frac\alpha2,
 \qquad 0\le\gamma<\frac{1-\alpha}{2},                    \label{eq:CW-ranges-1}\\
 \beta+\gamma<u
 &<\min\{\alpha,\alpha-\beta+\gamma,
                  1-\alpha+\beta-\gamma\}.               \label{eq:CW-ranges-2}
\end{align}
Put
\begin{equation}\label{eq:CW-affine}
 z_{\mathrm{CW}}=1-2u,
 \qquad m=1-2\alpha,
 \qquad \zeta=1-2\beta-2\gamma,
 \qquad \xi=1-2\alpha+2\beta-2\gamma,
\end{equation}
and define
\begin{equation*}
 \Lambda_{\alpha,\beta,\gamma}(u)
 ={}\frac{(\zeta\xi-mz_{\mathrm{CW}}^2)^2}
 {z_{\mathrm{CW}}^2(1-m^2)(1-z_{\mathrm{CW}}^2)}
 + \frac{(z_{\mathrm{CW}}^2-\xi^2)
              (\zeta^2-z_{\mathrm{CW}}^2)}
 {z_{\mathrm{CW}}^2(1-m^2)\sqrt{1-z_{\mathrm{CW}}^2}}.
                                                               \label{eq:CW-Lambda}
\end{equation*}
Let \(\cD_{\mathrm{CW}}(\delta)\) be the set of points satisfying
\eqref{eq:CW-ranges-1}--\eqref{eq:CW-ranges-2} and
\begin{equation}\label{eq:CW-constraint}
 \Lambda_{\alpha,\beta,\gamma}(u)
 >1-\frac{\delta}{2\alpha(1-\alpha)}.
\end{equation}
Then
\begin{equation*}
    \kappa_{\mathrm{CW}}(\delta)
 =\inf_{(\alpha,\beta,\gamma,u)\in\cD_{\mathrm{CW}}(\delta)}
 \biggl\{1-\Htwo(\alpha)+\Htwo(u)
 -\alpha\Htwo\!\left(\frac\beta\alpha\right)
 \nonumber -(1-\alpha)\Htwo\!\left(\frac\gamma{1-\alpha}\right)
 \biggr\}.
\end{equation*}
The previous combined moving-projection exponent is
\begin{equation*}\label{eq:previous-v-MRRW}
 R_2(\delta)\le
 \kappa_{\mathrm{bin}}(\delta)
 =\min\{\kappa_H(\delta),\kappa_{\mathrm{CW}}(\delta)\}
 <M_2(\delta).
\end{equation*}
The fixed-line specialization \(\beta=\gamma=0\), interpreted by limits at
the endpoints, is the classical MRRW objective~\eqref{eq:Fdelta}; positive
Johnson harmonic degrees improve its interior minimizers.  The honeycomb
extension is structurally different: its new degrees are the second rows of
the two ambient \(B_n\)-partitions, while the moving stabilizer type remains
\(S^{(n-k,k)}\).  Nevertheless, the symmetric honeycomb slice again recovers
the full \(M_2\) optimization exactly.

Combining the unchanged constant-weight branch with either the complete
honeycomb exponent or its explicit face and balanced subfamilies gives
\begin{equation*}\label{eq:new-v-previous}
 R_2\le\kappa_{\mathrm{best}}
 \le\kappa_{\mathrm{explicit}}
 \le\kappa_{\mathrm{bin}}<M_2\le M_1.
\end{equation*}
The comparison with \(\kappa_{\mathrm{bin}}\) is non-strict at distances
where \(\kappa_{\mathrm{CW}}\) is already the active branch.  At every
distance where \(\kappa_{\mathrm{bal}}<\kappa_{\mathrm{CW}}\), however,
\(\kappa_{\mathrm{explicit}}\) is strictly smaller than the previous
combined exponent; replacing \(\kappa_{\mathrm{bal}}\) by
\(\kappa_{\mathrm{HC}}\) can only strengthen the bound.

Appendix~\ref{app:channel-comparison} places both the preceding and the new
combined bounds against 2MQC:
\begin{align}
 \kappa_{\mathrm{bin}}(\delta)&\le\RtwoMQC(\delta), \notag \\
 \kappa_{\mathrm{pair}}(\delta)&<\RtwoMQC(\delta)
 \qquad(0<\delta<1/2). \notag
\end{align}
Thus the previous combined OpenAI exponent is already no weaker than 2MQC,
while the explicit honeycomb combination is pointwise strictly stronger.
Since \(\kappa_{\mathrm{best}}\le\kappa_{\mathrm{explicit}}
\le\kappa_{\mathrm{pair}}\), the complete honeycomb combination strictly
improves 2MQC at every nontrivial distance.

\section{Finite certificates for the honeycomb hierarchy}
\label{sec:hierarchy-overview}

\subsection{Te Three Directions of the Hierarchy}
The two-row construction above is not an isolated representation graph.  It
is the first explicit member of a hierarchy with two finite-certificate
directions and, as proved below, an unconditional asymptotic Horn--channel
companion.  The representation direction increases the number of rows
allowed in the stabilizer and ambient partitions.  This exposes the complete
Littlewood--Richardson multiplicity spaces and replaces the four scalar
transfer channels by matrix-valued box-transfer channels.  The anchor
direction localizes the resulting quadratic coding certificates at
successively larger collections of anchored codewords.  This is a moment
hierarchy on the forbidden-distance graph of the Hamming cube.  The channel
companion keeps the same Horn spectra and entropy potential, evaluates them
through the Alrabiah--Guruswami pretty-good criterion, and replaces the
unresolved higher-rank recoupling symbol by an explicit quantum-affinity
statistic.

The three directions have different purposes.
\begin{enumerate}[label=(\roman*),leftmargin=2.2em]
\item The \emph{representation depth} is the finite-certificate direction.
At depth one, all partitions have at most two rows and the transfer fiber is
scalar; this is exactly the honeycomb bound proved in the preceding sections.
At depth two, partitions have at most three rows, the
Littlewood--Richardson fiber is a one-dimensional hive interval, and there
are up to nine forward box-transfer channel types.  At general fixed depth,
the support is a Horn--hive polytope and the transfer fiber has only
polynomial dimension in the blocklength.  Extracting its bulk symbol is an
additional semiclassical problem.
\item The \emph{Horn--channel depth} is an unconditional asymptotic
direction indexed by the same number of rows.  Its level \(r\) uses two
positive \((r+1)\times(r+1)\) matrices.  It recovers the representation bulk
formula exactly at level one and yields a nested explicit matrix optimization
at every higher level, without claiming equality with the representation
symbol there.
\item The \emph{anchor depth} is the completeness direction.  It places every
representation-theoretic quadratic certificate into localizing moment
matrices indexed by small independent subsets of the forbidden-distance
graph.  The resulting bounds are monotone in the anchor depth and recover
\(\Atwo(n,d)\) exactly at sufficiently high level.
\end{enumerate}

The distinction is essential.  Enlarging a one-point representation graph,
even all the way to every stabilizer type, need not by itself distinguish all
finite codes.  Completeness is instead obtained by retaining higher joint
moments of code membership.  Conversely, the raw moment hierarchy is exact
but does not explain which low-degree constraints are analytically useful.
The representation axis supplies structured localizers with explicit
spectral and entropy meaning.

We use the following terminology throughout the rest of the paper.
\begin{definition}
\label{def:two-axis-hierarchy}
The \emph{row level} \(r\ge1\) allows stabilizer and ambient partitions with
at most \(r+1\) rows.  The \emph{anchor level} \(t\ge0\) is the order of the
localized moment relaxation.  Level \((1,0)\) is the finite two-row
moving-projection bound of Theorem~\ref{thm:finite-hyp}; its explicit bulk
exponent is
\begin{equation*}\label{eq:first-level-identification}
 \kappa_{\mathrm{HC}}^{[1]}(\delta)
 \defeq \kappa_{\mathrm{HC}}(\delta).
\end{equation*}
Thus the bound developed in Sections~\ref{sec:representation}--
\ref{sec:asymptotics} is, by definition, the \emph{first honeycomb bound}.
\end{definition}

The construction below is finite-dimensional and unconditional.  Its exact
convergence theorem does not assert that any fixed row or anchor level gives
the unknown asymptotic rate \(R_2(\delta)\).  It says that, for each fixed
\((n,d)\), the hierarchy reaches the exact integer optimum once the anchor
level is sufficiently high.  Thus the hierarchy is finitely convergent, while
leaving open the distinct question of whether a bounded level already yields
a sharp asymptotic exponent.


We extend the finite moving-projection argument to nontrivial multiplicity spaces by introducing frame-profile certificates.
Instead of seeking a matrix-valued Collatz--Wielandt theorem, we formulate a
finite \emph{frame certificate}.  It simultaneously chooses a moving frame
inside every multiplicity space and a contractive combination of the
box-transfer intertwiners.  The scalar profile theorem is the rank-one
special case.

\subsection{Multiplicity-space transfer operators}

Retain the general transitive setup of Section~\ref{sec:general}.  Thus a
finite group \(G\) acts transitively on \(X\), \(o\in X\),
\(H=\operatorname{Stab}_G(o)\), and \(x\mapsto\ell_x\in W\) is an
equivariant unit-vector embedding with
\[
 t(x,y)=\langle\ell_x,\ell_y\rangle\in\R.
\]
Fix an irreducible unitary \(H\)-representation \(E\) of dimension \(d_E\).
Let \(\Omega\) be a finite family of pairwise inequivalent irreducible
unitary \(G\)-representations, but do not impose multiplicity one.  Put
\begin{equation*}\label{eq:multiplicity-space}
 M_\omega=\Hom_H(E,\Res_H^G V_\omega),
 \qquad m_\omega=\dim M_\omega.
\end{equation*}
Choose the Hilbert-space structure on \(M_\omega\) for which the evaluation
map identifies the \(E\)-isotypic component unitarily with
\begin{equation*}\label{eq:isotypic-unitary}
 \Phi_{\omega,o}:M_\omega\otimes E\longrightarrow V_\omega.
\end{equation*}
Transporting by \(G\) gives moving isometries
\(\Phi_{\omega,x}:M_\omega\otimes E_x\to V_\omega\).

It is convenient to allow parallel tensor-product edges.  Let \(\cE\) be a
finite directed edge set.  An edge \(e\in\cE\) has source \(s(e)=\omega\),
target \(t(e)=\nu\), and an isometric \(G\)-intertwiner
\begin{equation*}\label{eq:parallel-intertwiner}
 C_e:V_\nu\longrightarrow W\otimes V_\omega.
\end{equation*}
For each fixed pair \((\omega,\nu)\), choose the images of the parallel
intertwiners orthogonally.  Contraction by \(\ell_x\) is \(H_x\)-equivariant,
so Schur's lemma on the irreducible factor \(E_x\) gives a unique operator
\begin{equation*}\label{eq:multiplicity-transfer}
 T_e:M_\omega\longrightarrow M_\nu
\end{equation*}
satisfying
\begin{equation}\label{eq:multiplicity-transfer-identity}
 C_e^*\bigl(\ell_x\otimes
 \Phi_{\omega,x}(a\otimes u)\bigr)
 =\Phi_{\nu,x}(T_ea\otimes u)
\end{equation}
for every \(a\in M_\omega\) and \(u\in E_x\).

The full tensor-product decomposition gives a row contraction on every
source multiplicity space.
\begin{lemma}
\label{lem:matrix-transfer-normalization}
If \(\cE\) contains every irreducible summand of
\(W\otimes V_\omega\) whose restriction contains \(E\), with an orthonormal
basis of parallel intertwiners, then
\begin{equation*}\label{eq:matrix-row-normalization}
 \sum_{e:\,s(e)=\omega}T_e^*T_e=\Id_{M_\omega}.
\end{equation*}
For an arbitrary retained edge set, the left-hand side is at most the
identity.
\end{lemma}

\begin{proof}
Fix \(a\in M_\omega\) and a unit vector \(u\in E\).  The vector
\(\ell_o\otimes\Phi_{\omega,o}(a\otimes u)\) has squared norm
\(\|a\|^2\).  Its orthogonal projections onto the retained irreducible
summands have squared norms \(\|T_ea\|^2\) by
\eqref{eq:multiplicity-transfer-identity}.  Pythagoras gives
\(\sum_e\|T_ea\|^2\le\|a\|^2\), with equality for the full decomposition.
Polarization proves the operator identity and inequality.
\end{proof}

\subsection{Frame certificates}

A scalar positive profile chooses one moving copy of \(E\) across the
ambient blocks.  With multiplicities, the correct replacement is a common
frame whose image in each block may have rank larger than one.

\begin{definition}
\label{def:frame-certificate}
A \emph{frame-profile certificate} on
\((E,\Omega,\cE)\) consists of the following data.
\begin{enumerate}[label=(\alph*),leftmargin=2.2em]
\item A nonzero finite-dimensional Hilbert space \(F\), with
\(m=\dim F\).
\item Linear maps \(A_\omega:F\to M_\omega\) satisfying the Parseval
condition
\begin{equation}\label{eq:frame-parseval}
 \sum_{\omega\in\Omega}A_\omega^*A_\omega=\Id_F.
\end{equation}
\item Scalars \(b_e\in\C\) satisfying the column-contraction conditions
\begin{equation}\label{eq:b-column-contraction}
 \sum_{e:\,t(e)=\nu}|b_e|^2\le1
 \qquad(\nu\in\Omega).
\end{equation}
\item A number \(\lambda>0\) such that the harmonic frame equations
\begin{equation}\label{eq:frame-harmonic}
 \sum_{e:\,t(e)=\nu}\overline{b_e}\,T_eA_{s(e)}
 =\sqrt\lambda\,A_\nu
 \qquad(\nu\in\Omega)
\end{equation}
hold.
\end{enumerate}
Set
\begin{equation*}\label{eq:qomega-Delta}
 q_\omega=\|A_\omega\|_{\mathrm F}^2,
 \qquad
 \Delta(A)=\sum_{\omega\in\Omega}\frac{q_\omega^2}{D_\omega},
 \qquad D_\omega=\dim V_\omega.
\end{equation*}
By \eqref{eq:frame-parseval}, \(\sum_\omega q_\omega=m\).
\end{definition}

The coefficients \(b_e\) define a global contraction, while
\eqref{eq:frame-harmonic} says that contraction by the moving coordinate
vector acts as \(\sqrt\lambda\) on the selected moving frame.  No
commutativity or detailed-balance hypothesis is needed.  The combined map
\(A:F\to\bigoplus_\omega M_\omega\),
\(Aa=(A_\omega a)_\omega\), is an isometry by
\eqref{eq:frame-parseval}.  Consequently
\begin{equation}\label{eq:frame-dimension-cap}
 m\le\sum_{\omega\in\Omega}m_\omega.
\end{equation}
Thus, on every fixed retained graph, the frame search is
finite-dimensional, which allows an arbitrarily large auxiliary space \(F\)
does not enlarge the certificate family.

\begin{theorem}
\label{thm:matrix-profile}
Let \(\mathfrak h=(F,A,b,\lambda)\) be a frame-profile certificate.  If
\(\lambda>\max\{s,0\}\), then every
\(\cC\subseteq X\) satisfying \(t(x,y)\le s\) for distinct
\(x,y\in\cC\) obeys
\begin{equation}\label{eq:matrix-profile-bound}
 |\cC|
 \le \mathsf B(\mathfrak h;s)
 \defeq
 \frac{(1-s)m}
 {d_E(\lambda-s)\Delta(A)}.
\end{equation}
In particular, arbitrary Littlewood--Richardson multiplicities are compatible
with the moving-projection method.
\end{theorem}

\begin{proof}
Put \(\cV_\Omega=\bigoplus_{\omega\in\Omega}V_\omega\).  For
\(a\in F\) and \(u\in E_x\), define
\begin{equation*}\label{eq:matrix-moving-frame}
 \Psi_x(a\otimes u)
 =\bigoplus_{\omega\in\Omega}
   \Phi_{\omega,x}(A_\omega a\otimes u).
\end{equation*}
The Parseval identity \eqref{eq:frame-parseval} makes \(\Psi_x\) an
isometry from \(F\otimes E_x\) to \(\cV_\Omega\).  Let
\(P_x=\Psi_x\Psi_x^*\); thus
\begin{equation*}\label{eq:matrix-projector-rank}
 \rank P_x=md_E.
\end{equation*}

Define \(B:\cV_\Omega\to W\otimes\cV_\Omega\) blockwise by
\begin{equation*}\label{eq:matrix-global-B}
 (Bh)_\omega
 =\sum_{e:\,s(e)=\omega}b_eC_eh_{t(e)}.
\end{equation*}
For fixed target \(\nu\), the relevant intertwiner images are orthogonal, so
\eqref{eq:b-column-contraction} implies \(B^*B\preceq\Id\).  Moreover,
\eqref{eq:multiplicity-transfer-identity} and
\eqref{eq:frame-harmonic} give
\begin{align}
 B^*\bigl(\ell_x\otimes\Psi_x(a\otimes u)\bigr)
 &=\bigoplus_{\nu\in\Omega}
 \Phi_{\nu,x}\left(
  \sum_{e:\,t(e)=\nu}\overline{b_e}T_eA_{s(e)}a
  \otimes u\right)\nonumber\\
 &=\sqrt\lambda\,\Psi_x(a\otimes u).                    \label{eq:matrix-key-contraction}
\end{align}
Since \(B\) is a contraction and \(\Psi_x\) is an isometry,
\eqref{eq:matrix-key-contraction} also implies \(\lambda\le1\).

Now repeat the positive-definite decomposition from
Theorem~\ref{thm:general-graph}.  Set
\[
 K_{\mathfrak h}(x,y)=\tr(P_xP_y),
 \quad L_x=(\ell_x\otimes\Id)P_x,
 \quad G_x=BP_x,
 \quad \Theta_x=L_x-\sqrt\lambda G_x.
\]
With
\(K_B(x,y)=\tr(P_xB^*BP_y)\), identity
\eqref{eq:matrix-key-contraction} yields
\begin{equation*}\label{eq:matrix-PD-decomp}
 (t(x,y)-s)K_{\mathfrak h}(x,y)
 = (\lambda-s)K_{\mathfrak h}(x,y)
   +\langle\Theta_x,\Theta_y\rangle_{\mathrm{HS}}
   +\lambda\bigl(K_{\mathfrak h}(x,y)-K_B(x,y)\bigr).
\end{equation*}
All three kernels on the right are positive definite.  Summing over
\(\cC^2\), using the off-diagonal inequality on the left and
\(K_{\mathfrak h}(x,x)=md_E\), gives
\begin{equation}\label{eq:matrix-upper-sum}
 (\lambda-s)\left\|\sum_{x\in\cC}P_x\right\|_{\mathrm{HS}}^2
 \le |\cC|(1-s)md_E.
\end{equation}

It remains to use every ambient block separately.  The \(V_\omega\)-block
of \(P_x\) has trace \(d_Eq_\omega\).  Therefore trace
Cauchy--Schwarz gives
\begin{align}
 \left\|\sum_{x\in\cC}P_x\right\|_{\mathrm{HS}}^2
 &\ge\sum_{\omega\in\Omega}
 \frac{\left(|\cC|d_Eq_\omega\right)^2}{D_\omega}
 =|\cC|^2d_E^2\Delta(A).                                \label{eq:matrix-block-trace}
\end{align}
Combining \eqref{eq:matrix-upper-sum} and
\eqref{eq:matrix-block-trace} proves \eqref{eq:matrix-profile-bound}.
\end{proof}

\begin{corollary}
\label{cor:scalar-profile-recovery}
Suppose every \(M_\omega\) is one-dimensional and write
\(T_{\omega,\nu}=c_{\omega,\nu}\ge0\).  Given a positive profile \(w\), set
\[
 \lambda=\min_\nu
 \frac{\sum_\xi c_{\xi,\nu}^2w_\xi}{w_\nu}
\]
and assume \(\lambda>0\).  Let \(Z=\sum_\omega w_\omega\), take
\(F=\C\), and put
\begin{align}
 A_\omega&=\sqrt{w_\omega/Z},\nonumber\\
 b_{\omega,\nu}
 &=c_{\omega,\nu}\sqrt{w_\omega}
   \frac{\sqrt{\lambda w_\nu}}
        {\sum_\xi c_{\xi,\nu}^2w_\xi}, \notag
\end{align}
Then Definition~\ref{def:frame-certificate} holds and
\eqref{eq:matrix-profile-bound} is exactly
\eqref{eq:general-profile-bound}.
\end{corollary}

\begin{proof}
The column norm in \eqref{eq:b-column-contraction} is
\(
 \lambda w_\nu/
 \sum_\xi c_{\xi,\nu}^2w_\xi\le1.
\)
The harmonic equation is the same cancellation as
\eqref{eq:key-contraction}.  Finally,
\(
 q_\omega=w_\omega/Z
\)
and hence
\(
 \Delta(A)=Z^{-2}\sum_\omega w_\omega^2/D_\omega.
\)
Substitution gives the effective dimension in \eqref{eq:Deff}.
\end{proof}

\subsection{Quadratic certificate polynomials}

The proof gives more than the cardinality inequality.  It produces two
quadratic polynomials that are nonnegative on every code indicator and can
therefore be localized at higher anchor levels.

Let \(z=(z_x)_{x\in X}\) and put
\begin{equation*}\label{eq:sum-variable}
 S(z)=\sum_{x\in X}z_x,
 \qquad
 \cK_{\mathfrak h}(z)
 =\sum_{x,y\in X}K_{\mathfrak h}(x,y)z_xz_y.
\end{equation*}
Define
\begin{align}
 Q_{\mathfrak h}(z)
 &=(1-s)md_E S(z)-(\lambda-s)\cK_{\mathfrak h}(z),
                                                               \label{eq:Qh}\\
 T_{\mathfrak h}(z)
 &=\cK_{\mathfrak h}(z)-d_E^2\Delta(A)S(z)^2.            \label{eq:Th}
\end{align}

\begin{proposition}
\label{prop:honeycomb-polynomials}
If \(z=\one_{\cC}\) is the indicator of a code with minimum distance at
least \(d\), then
\begin{equation*}\label{eq:QT-nonnegative}
 Q_{\mathfrak h}(z)\ge0,
 \qquad T_{\mathfrak h}(z)\ge0.
\end{equation*}
Moreover,
\begin{equation}\label{eq:cap-polynomial-combination}
 Q_{\mathfrak h}(z)+(\lambda-s)T_{\mathfrak h}(z)
 =(\lambda-s)d_E^2\Delta(A)
 \bigl(\mathsf B(\mathfrak h;s)S(z)-S(z)^2\bigr).
\end{equation}
Thus the ordinary cardinality cap is a positive linear combination of two
more structured quadratic certificates.
\end{proposition}

\begin{proof}
Inequality \(Q_{\mathfrak h}(\one_\cC)\ge0\) is precisely
\eqref{eq:matrix-upper-sum}.  The block trace argument
\eqref{eq:matrix-block-trace} proves
\(T_{\mathfrak h}(\one_\cC)\ge0\).  The final identity follows by direct
substitution of \eqref{eq:matrix-profile-bound}.
\end{proof}

The distinction between \(Q_{\mathfrak h}\) and
\(T_{\mathfrak h}\) matters at positive anchor depth.  Their scalar sum
only remembers the global cap, whereas their separate localizing matrices
retain the positive-definite kernel and the blockwise trace geometry.

\section{Higher-row representation bounds}
\label{sec:row-hierarchy}

We now specialize the matrix theorem to \(B_n=C_2^n\rtimes S_n\).  The
higher levels have a particularly natural combinatorial description: the
Littlewood--Richardson multiplicity spaces are counted by the integral hives,
or equivalently the representation-theoretic honeycombs, of Knutson and Tao
\cite{KnutsonTao1999}.  Thus the term ``honeycomb'' becomes literal beyond
the first scalar level.

For \(r\ge1\), let
\begin{equation*}\label{eq:row-partitions}
 \Par_r(n)=\{\nu\vdash n:\ell(\nu)\le r+1\}.
\end{equation*}
For \(\nu\in\Par_r(n)\) and \(0\le i\le n\), define the compatibility
fiber
\begin{equation}\label{eq:row-compatibility-fiber}
 \Omega_r(\nu,i)
 =\{(\lambda,\mu):\lambda\vdash n-i,\ \mu\vdash i,
                  c_{\lambda\mu}^{\nu}>0\}.
\end{equation}
The box-transfer graph associated with \(\nu\) is defined on the union
\begin{equation}\label{eq:row-compatibility-union}
 \Omega_r(\nu)=\bigcup_{i=0}^n\Omega_r(\nu,i),
\end{equation}
since every transfer changes the color size \(i\) by one.  Positivity of the
Littlewood--Richardson coefficient implies
\(\ell(\lambda),\ell(\mu)\le r+1\).  The stabilizer multiplicity space at
\(\omega=(\lambda,\mu)\) is
\begin{equation}\label{eq:LR-multiplicity-space}
 M_\omega^\nu
 =\Hom_{S_n}\bigl(S^\nu,
   \Res_{S_n}^{B_n}V_{\lambda,\mu}\bigr),
 \qquad
 \dim M_\omega^\nu=c_{\lambda\mu}^{\nu}.
\end{equation}

The signed natural representation transfers one box from one component to
the other.  A forward edge chooses a removable row \(a\) of \(\lambda\) and
an addable row \(b\) of \(\mu\):
\begin{equation}\label{eq:general-box-transfer}
 (\lambda,\mu)
 \longrightarrow
 (\lambda-\square_a,\mu+\square_b).
\end{equation}
There are at most \((r+1)^2\) forward channels and the same reverse
channels.  Contraction with \(\ell_o\) produces a matrix
\begin{equation}\label{eq:LR-transfer-matrix}
 T_{a,b}^{\lambda,\mu;\nu}:
 M_{\lambda,\mu}^\nu
 \longrightarrow
 M_{\lambda-\square_a,\mu+\square_b}^\nu.
\end{equation}
These matrices are canonically determined up to unitary changes of basis in
the multiplicity spaces.  They can be computed without a closed recoupling
formula by using the central projectors of \(B_n\), the \(S_n\)-isotypic
projector for \(S^\nu\), and the contraction map.  Equivalently, one may use
Young--Yamanouchi bases and orthonormalize the Littlewood--Richardson tableau
basis.  At row level one, every multiplicity is one and
\eqref{eq:LR-transfer-matrix} is the scalar Racah coefficient calculated in
Section~\ref{sec:weights}.

Let \(s_{n,d}=1-2d/n\).  We define nested libraries
\(\mathfrak C_r(n)\) of frame certificates.
\begin{itemize}[leftmargin=2em]
\item \(\mathfrak C_1(n)\) consists exactly of the scalar certificates from
Theorem~\ref{thm:finite-hyp}, over every \(0\le k\le n/2\), every connected
two-row truncation \(\Omega\), and every positive profile \(w\).
\item For \(r\ge2\), \(\mathfrak C_r(n)\) contains
\(\mathfrak C_{r-1}(n)\) and every frame-profile certificate supported on a
finite box-transfer subgraph associated with
\(\nu\in\Par_r(n)\).
\end{itemize}
Define the raw row bound
\begin{equation*}\label{eq:raw-row-bound}
 \HC^{\mathrm{row}}_r(n,d)
 =\inf_{\substack{\mathfrak h\in\mathfrak C_r(n)\\
                   \lambda(\mathfrak h)>s_{n,d}}}
   \mathsf B(\mathfrak h;s_{n,d}).
\end{equation*}
An empty infimum is interpreted as \(2^n\), although the inclusion of the
first-level library makes the bound nontrivial in the regime treated in this
paper.

\begin{theorem}
\label{thm:row-monotonicity}
For every \(n,d\),
\begin{equation*}\label{eq:row-monotone-chain}
 \Atwo(n,d)
 \le \HC^{\mathrm{row}}_{r+1}(n,d)
 \le \HC^{\mathrm{row}}_r(n,d).
\end{equation*}
Moreover, \(\HC^{\mathrm{row}}_1(n,d)\) is exactly the full finite profile
optimization furnished by Theorem~\ref{thm:finite-hyp}.  Its explicit
F{\o}lner-box, dimension-sum exponent is the first honeycomb exponent
\(\kappa_{\mathrm{HC}}^{[1]}=\kappa_{\mathrm{HC}}\).
\end{theorem}

\begin{proof}
Every member of every certificate library gives a valid upper bound by
Theorem~\ref{thm:matrix-profile}.  The libraries are nested, so taking the
infimum can only decrease the bound.  The level-one identification is built
into the definition of \(\mathfrak C_1(n)\), and the asymptotic statement is
Theorem~\ref{thm:asymptotic-proof}.
\end{proof}

The theorem does \emph{not} claim that
\(\HC^{\mathrm{row}}_{n-1}(n,d)=\Atwo(n,d)\).  It is conceivable that the
full one-point representation family already has much greater power than is
currently visible, but there is no completeness argument of that form.  The
anchor axis in Section~\ref{sec:moment-hierarchy} supplies a rigorous route
to exactness.


At row level two, write
\begin{align*}
 \lambda&=(p-a-b,a,b),&
 \mu&=(q-c-e,c,e),&
 \nu&=(n-u-v,u,v),                                      \label{eq:three-row-shapes}
\end{align*}
with \(p+q=n\).  The coefficient \(c_{\lambda\mu}^{\nu}\) is the number of
integer points in a one-dimensional hive interval.  Thus multiplicity can
exceed one, but it grows at most linearly under proportional scaling.  The
up to nine forward channel types are indexed by
\begin{equation*}\label{eq:nine-channels}
 (a_0,b_0)\in\{1,2,3\}\times\{1,2,3\},
\end{equation*}
subject to the usual removable/addable-box conditions.  Each channel acts by
a matrix between two adjacent hive intervals.

\subsection{Horn domains and entropy potentials at fixed row level}
\label{subsec:horn-entropy}

The dimension part of every fixed row level has a clean asymptotic form.
Let
\begin{equation*}\label{eq:normalized-higher-shapes}
 x=\frac{|\mu|}{n},
 \qquad
 \alpha=\frac{\lambda}{(1-x)n},
 \qquad
 \beta=\frac{\mu}{xn},
 \qquad
 \gamma=\frac{\nu}{n},
\end{equation*}
where \(\alpha,\beta,\gamma\) are probability vectors with at most
\(r+1\) nonzero, weakly decreasing coordinates.  The compatibility condition
\(c_{\lambda\mu}^{\nu}>0\) converges to the Horn cone
\cite{Horn1962,Klyachko1998}. Write the resulting compact normalized
domain of quadruples \( (x,\alpha,\beta,\gamma) \) as
\(\cD_r^{\mathrm{Horn}}\).  The strict bulk has \(0<x<1\); the cases
\(x=0,1\) are interpreted by ordinary boundary limits.

For a probability vector \(p\), let
\begin{equation*}\label{eq:vector-entropy}
 \mathsf H(p)=-\sum_j p_j\log_2p_j.
\end{equation*}
The hook formula, uniformly away from changes in row number, gives
\begin{align}
 \frac1n\log_2\dim V_{\lambda,\mu}
 &\longrightarrow
 \Htwo(x)+(1-x)\mathsf H(\alpha)+x\mathsf H(\beta),
                                                               \label{eq:higher-ambient-entropy}\\
 \frac1n\log_2 f^\nu
 &\longrightarrow \mathsf H(\gamma).                 \label{eq:higher-stabilizer-entropy}
\end{align}
Consequently the row-\(r\) entropy potential is
\begin{equation}\label{eq:higher-Phi}
 \Phi_r(x,\alpha,\beta,\gamma)
 =\Htwo(x)+(1-x)\mathsf H(\alpha)+x\mathsf H(\beta)
  -\mathsf H(\gamma).
\end{equation}
For \(r=1\), write
\(
 \alpha=(1-\eta,\eta),
 \beta=(1-\rho,\rho),
 \gamma=(1-z,z)
\); then \eqref{eq:higher-Phi} is exactly
\(\Phi_{\mathrm{HC}}\).

At fixed \(r\), the hive polytope has fixed dimension, so its number of
integer points grows only polynomially with \(n\).  Littlewood--Richardson
multiplicity therefore contributes no additional exponential entropy term.
It can nevertheless change the local spectral value, and that is where the
matrix-valued theorem is essential.

\section{An unconditional Horn--channel hierarchy}
\label{sec:horn-channel}

The coding-theoretic input in this section is the pretty-good criterion of
Alrabiah and Guruswami \cite[Theorem~1]{AlrabiahGuruswami2026}: a
binary-input output-symmetric classical--quantum channel whose PGM bit error
is below \(\delta\) gives a rate upper bound equal to its uniform Holevo
information.  Their MQC and 2MQC bounds arise by optimizing carefully chosen
qubit-channel families, we use their criterion as a black box and introduce
a matrix family whose feasible spectra form the Horn domain. 

The fixed-row entropy calculation of Section~\ref{sec:row-hierarchy} does
not require the asymptotic Racah matrix.  The same Horn data admit the explicit classical--quantum realization
below.  It produces an unconditional higher-row asymptotic hierarchy that
agrees exactly with the first honeycomb exponent, while remaining logically
separate from the conditional representation-theoretic bulk symbol of
Section~\ref{sec:row-hierarchy}.

For a positive semidefinite matrix \(M\) of trace at most one, define its
subnormalized entropy by
\begin{equation*}\label{eq:subnormalized-entropy}
 \mathsf S(M)=-\tr(M\log_2M),
\end{equation*}
with the usual continuous convention at zero.  For \(m\ge1\), let
\begin{equation*}\label{eq:horn-channel-domain}
 \mathcal K_m=
 \{(K,L):K,L\in\operatorname{Herm}_m(\C),\ K,L\succeq0,
             \ \tr(K+L)=1\}.
\end{equation*}
For \((K,L)\in\mathcal K_m\), put
\begin{align}
 \Gamma_{\mathrm{ch}}(K,L)
 &=2\tr(K^{1/2}L^{1/2}), \notag         \\
 \Phi_{\mathrm{ch}}(K,L)
 &=\mathsf S(K)+\mathsf S(L)-\mathsf S(K+L).\notag        
\end{align}
Although \(K^{1/2}L^{1/2}\) need not be Hermitian, its trace is real and
nonnegative, since
\(\tr(K^{1/2}L^{1/2})
 =\tr(L^{1/4}K^{1/2}L^{1/4})\).  Hilbert--Schmidt
Cauchy--Schwarz gives \(0\le\Gamma_{\mathrm{ch}}\le1\).

\begin{proposition}
\label{prop:horn-channel-realization}
Let \((K,L)\in\mathcal K_m\), and define the \(2m\times2m\) states
\begin{equation*}\label{eq:horn-channel-states}
 Y=\begin{pmatrix}K^{1/2}\\ L^{1/2}\end{pmatrix},
 \qquad
 \sigma_0=YY^*,
 \qquad
 J=\begin{pmatrix}I_m&0\\0&-I_m\end{pmatrix},
 \qquad
 \sigma_1=J\sigma_0J.
\end{equation*}
Then \((\sigma_0,\sigma_1)\) is a binary-input output-symmetric
classical--quantum channel.  Its uniform-prior Holevo information and PGM
bit error are
\begin{align}
 \chi(\sigma_0,\sigma_1)
 &=\Phi_{\mathrm{ch}}(K,L),                              \label{eq:channel-Holevo}\\
 p_{\mathrm e}(\sigma_0,\sigma_1)
 &=\frac{1-\Gamma_{\mathrm{ch}}(K,L)}2.                  \label{eq:channel-PGM}
\end{align}
\end{proposition}

\begin{proof}
The states are positive and have trace one because
\(Y^*Y=K+L\).  The involution \(J\) exchanges them, so the channel is
output-symmetric.  Their average is
\begin{equation*}\label{eq:horn-channel-average}
 \overline\sigma=\frac{\sigma_0+\sigma_1}{2}=K\oplus L.
\end{equation*}
The nonzero eigenvalues of \(YY^*\) are those of \(Y^*Y=K+L\), and the same
holds for \(JYY^*J\).  Hence the uniform Holevo information is
\[
 \mathsf S(K\oplus L)-\mathsf S(K+L)
 =\mathsf S(K)+\mathsf S(L)-\mathsf S(K+L).
\]
For positive definite \(K,L\), the binary PGM trace formula of
\cite[Proposition~12]{AlrabiahGuruswami2026} gives
\begin{align*}
 p_{\mathrm e}
 &=\frac12\tr\!\left(
   \sigma_0\overline\sigma^{-1/2}
   \sigma_1\overline\sigma^{-1/2}\right)\\
 &=\frac12\left(1-2\tr(K^{1/2}L^{1/2})\right).
\end{align*}
Indeed, after multiplying by
\(\overline\sigma^{-1/2}=K^{-1/2}\oplus L^{-1/2}\), the two diagonal
trace contributions are \(\tr K+\tr L=1\), while the two off-diagonal
contributions are both \(-\tr(K^{1/2}L^{1/2})\).  Singular pairs follow by
adding \(\varepsilon I_m\), renormalizing, and taking
\(\varepsilon\downarrow0\).
\end{proof}

\begin{definition}
\label{def:horn-channel-exponent}
For \(r\ge0\) and \(0<\delta<1/2\), define
\begin{equation}\label{eq:kappa-channel-r}
 \kappa_{\mathrm{ch}}^{[r]}(\delta)
 =\inf_{\substack{(K,L)\in\mathcal K_{r+1}\\
        \Gamma_{\mathrm{ch}}(K,L)>1-2\delta}}
       \Phi_{\mathrm{ch}}(K,L).
\end{equation}
We call \(r\) the \emph{Horn--channel row level}.  Thus level \(r\) uses
matrices of size \(r+1\).
\end{definition}

\begin{theorem}
\label{thm:horn-channel-hierarchy}
For every \(r\ge0\) and \(0<\delta<1/2\),
\begin{equation}\label{eq:horn-channel-rate-bound}
 R_2(\delta)\le \kappa_{\mathrm{ch}}^{[r]}(\delta).
\end{equation}
Moreover,
\begin{equation}\label{eq:horn-channel-nesting}
 \kappa_{\mathrm{ch}}^{[r+1]}(\delta)
 \le \kappa_{\mathrm{ch}}^{[r]}(\delta).
\end{equation}
The strict-constraint infimum in \eqref{eq:kappa-channel-r} is equal to the
minimum over the compact closed region
\(\Gamma_{\mathrm{ch}}\ge1-2\delta\).
\end{theorem}

\begin{proof}
If \(\Gamma_{\mathrm{ch}}(K,L)>1-2\delta\), then
\eqref{eq:channel-PGM} gives \(p_{\mathrm e}<\delta\).  The pretty-good
criterion \cite[Theorem~1]{AlrabiahGuruswami2026}, together with
\eqref{eq:channel-Holevo}, gives
\(R_2(\delta)\le\Phi_{\mathrm{ch}}(K,L)\).  Taking the infimum proves
\eqref{eq:horn-channel-rate-bound}.  The embedding
\((K,L)\mapsto(K\oplus0,L\oplus0)\) preserves both functions, proving
\eqref{eq:horn-channel-nesting}.

For the final assertion, compactness and continuity give a minimizer on the
closed region.  If its channel statistic is already strict, there is nothing
to prove.  Otherwise write \(A=K^{1/2}\), \(B=L^{1/2}\),
\(\gamma=\Gamma_{\mathrm{ch}}(K,L)<1\), and for \(t>0\) put
\begin{equation}\label{eq:coherence-boost}
 K_t=\frac{(A+tB)^2}{N_t},
 \qquad
 L_t=\frac{(B+tA)^2}{N_t},
 \qquad
 N_t=1+t^2+2t\gamma.
\end{equation}
These matrices remain in \(\mathcal K_m\), converge to \((K,L)\), and a
direct trace calculation gives
\begin{equation}\label{eq:coherence-boost-gamma}
 \Gamma_{\mathrm{ch}}(K_t,L_t)
 =\frac{\gamma(1+t^2)+2t}{1+t^2+2t\gamma}
 =\gamma+\frac{2t(1-\gamma^2)}{1+t^2+2t\gamma}>\gamma.
\end{equation}
Continuity of \(\Phi_{\mathrm{ch}}\) identifies the strict infimum with the
closed minimum.  Finally, \(\gamma=1\) can occur only at the equality case
\(K=L\) of Hilbert--Schmidt Cauchy--Schwarz, where
\(\Phi_{\mathrm{ch}}=1\); this is not minimizing because the scalar feasible
family already gives \(M_1(\delta)<1\).
\end{proof}

\begin{proposition}
\label{prop:channel-equals-honeycomb}
For every \(0<\delta<1/2\),
\begin{equation*}\label{eq:channel-honeycomb-equality}
 \kappa_{\mathrm{ch}}^{[0]}(\delta)=M_1(\delta),
 \qquad
 \kappa_{\mathrm{ch}}^{[1]}(\delta)
 =\kappa_{\mathrm{HC}}(\delta).
\end{equation*}
\end{proposition}

\begin{proof}
For scalar \(K=[1-x]\), \(L=[x]\), one has
\(\Gamma_{\mathrm{ch}}=2\sqrt{x(1-x)}\) and
\(\Phi_{\mathrm{ch}}=\Htwo(x)\).  Minimizing on the spectral boundary gives
\eqref{eq:M1}.

Now let \(K,L\) be \(2\times2\).  After exchanging them if necessary, write
their ordered eigenvalues and those of their sum as
\begin{align*}
 \operatorname{spec}(K)
 &=\bigl((1-x)(1-\eta),(1-x)\eta\bigr),\\
 \operatorname{spec}(L)
 &=\bigl(x(1-\rho),x\rho\bigr),\\
 \operatorname{spec}(K+L)&=(1-z,z),
\end{align*}
with the ranges in \eqref{eq:parameters-domain}.  Let \(u_1,v_1\) be top
eigenvectors of \(K,L\), respectively, and set
\(P=|\langle u_1,v_1\rangle|^2\).  The squared overlap matrix of two
orthonormal bases in dimension two is
\[
 \begin{pmatrix}P&1-P\\1-P&P\end{pmatrix}.
\]
Therefore
\begin{align*}
 \tr(K^{1/2}L^{1/2})
 =\sqrt{x(1-x)}\Bigl[&P\bigl(
 \sqrt{(1-\eta)(1-\rho)}+\sqrt{\eta\rho}\bigr)\\
 &+(1-P)\bigl(
 \sqrt{\rho(1-\eta)}+\sqrt{\eta(1-\rho)}\bigr)\Bigr].
\end{align*}
This is exactly one half of \eqref{eq:Gamma-hyp}.  To identify the domain,
write the traceless parts of \(K,L\) as Bloch vectors of lengths
\(A=(1-x)(1-2\eta)\) and \(B=x(1-2\rho)\).  Their top-eigenvector overlap
satisfies \(2P-1=\cos\vartheta\), where \(\vartheta\) is the angle between
the Bloch vectors.  Since the traceless part of \(K+L\) has length
\(E=1-2z\),
\[
 E^2=A^2+B^2+2AB\cos\vartheta
     =(A-B)^2+4ABP.
\]
Thus the Horn triangle is exactly \eqref{eq:triangle-continuum}, and the
resulting \(P\) is \eqref{eq:Pdef}.  Finally,
\[
 \Phi_{\mathrm{ch}}(K,L)
 =\Htwo(x)+(1-x)\Htwo(\eta)+x\Htwo(\rho)-\Htwo(z)
 =\Phi_{\mathrm{HC}}.
\]
Conversely, every point of \(\cD_{\mathrm{HC}}\) determines two Bloch
vectors of lengths \(A,B\) and angle \(\cos\vartheta=2P-1\), hence a pair
\(K,L\) with exactly these spectra.  The two feasible optimizations are
therefore identical.
\end{proof}

\begin{corollary}
\label{cor:rank-three-channel-bound}
For every \(0<\delta<1/2\),
\begin{equation*}\label{eq:rank-three-chain}
 R_2(\delta)
 \le\kappa_{\mathrm{ch}}^{[2]}(\delta)
 \le\kappa_{\mathrm{HC}}(\delta).
\end{equation*}
Consequently
\begin{equation*}\label{eq:rank-three-combined}
 R_2(\delta)\le
 \min\{\kappa_{\mathrm{CW}}(\delta),
        \kappa_{\mathrm{ch}}^{[2]}(\delta)\}
 \le\kappa_{\mathrm{best}}(\delta).
\end{equation*}
\end{corollary}

The matrix formula is exactly the higher-row entropy potential in spectral
coordinates.  Let \(x=\tr L\), and let
\begin{equation*}\label{eq:channel-spectral-coordinates}
 \operatorname{spec}(K)=(1-x)\alpha,
 \qquad
 \operatorname{spec}(L)=x\beta,
 \qquad
 \operatorname{spec}(K+L)=\gamma.
\end{equation*}
Then the spectral triple lies in
\(\cD_r^{\mathrm{Horn}}\), and
\begin{equation*}\label{eq:channel-Phi-equals-higher-Phi}
 \Phi_{\mathrm{ch}}(K,L)
 =\Htwo(x)+(1-x)\mathsf H(\alpha)+x\mathsf H(\beta)
  -\mathsf H(\gamma)
 =\Phi_r(x,\alpha,\beta,\gamma).
\end{equation*}
Conversely, the solution of Horn's problem
\cite{Horn1962,Klyachko1998,KnutsonTao1999} realizes every point of the Horn
domain by such a sum.  For \(m=2\), the spectra determine the relative basis
up to the single overlap \(P\).  For \(m\ge3\), the same spectral boundary
data can carry a positive-dimensional isospectral Horn fiber, and
\(\Gamma_{\mathrm{ch}}\) varies over that fiber.  This is the new
higher-level degree of freedom.

For computation, every pair can be written, on the support of
\(G=K+L\), as
\begin{equation}\label{eq:channel-contraction-parametrization}
 G=V\diag(\gamma_1,\ldots,\gamma_m)V^*,
 \qquad
 K=G^{1/2}D G^{1/2},
 \qquad
 L=G^{1/2}(I-D)G^{1/2},
\end{equation}
where \(V\in U(m)\) and \(D=\diag(a_1,\ldots,a_m)\) with
\(0\le a_j\le1\).  Indeed,
\(G^{-1/2}KG^{-1/2}\) is a positive contraction and can be diagonalized.
The real restriction \(V\in O(m)\) is already a valid subfamily and is the
one used in Appendix~\ref{app:numerics}.

\begin{proposition}
\label{prop:strict-rank-opening}
Fix \(r\ge0\) and \(0<\delta<1/2\).  Suppose the compact closed problem for
\(\kappa_{\mathrm{ch}}^{[r]}(\delta)\) has a minimizer \((K,L)\) for which
both matrices are positive definite.  Then
\begin{equation*}\label{eq:strict-rank-opening}
 \kappa_{\mathrm{ch}}^{[r+1]}(\delta)
 <\kappa_{\mathrm{ch}}^{[r]}(\delta).
\end{equation*}
In particular, the scalar optimizer is positive, so
\(\kappa_{\mathrm{HC}}(\delta)<M_1(\delta)\) also follows from this
rank-opening mechanism.
\end{proposition}

\begin{proof}
Let \(s=1-2\delta\).  If the minimizing pair has
\(\Gamma_{\mathrm{ch}}>s\), set \(K_t=K\) and \(L_t=L\) for every
\(t>0\); the available feasibility slack is then fixed.  If instead
\(\Gamma_{\mathrm{ch}}=s\), use the coherence boost
\eqref{eq:coherence-boost}; by
\eqref{eq:coherence-boost-gamma} it creates slack of order \(t\).
In both cases positive definiteness makes the entropy smooth and
\begin{equation}\label{eq:boost-objective-expansion}
 \Phi_{\mathrm{ch}}(K_t,L_t)
 =\Phi_{\mathrm{ch}}(K,L)+O(t),
\end{equation}
where the left-hand side is actually constant in the first case.

Embed \(K_t,L_t\) into \(\C^{m+1}=\C^m\oplus\C e\).  Choose an eigenvector
\(v\) of \(L_t\) with eigenvalue \(\ell>0\), and let \(U_\theta\) rotate
\(v\) toward \(e\) through angle \(\theta\), acting identically on the
orthogonal complement.  Put
\[
 \widetilde K_t=K_t\oplus0,
 \qquad
 \widetilde L_{t,\theta}
 =U_\theta(L_t\oplus0)U_\theta^*.
\]
The individual spectra, and hence their two entropy terms, are unchanged.
Since \(K_t\) is positive definite,
\begin{equation}\label{eq:rank-opening-gamma-loss}
 \Gamma_{\mathrm{ch}}(\widetilde K_t,
                       \widetilde L_{t,\theta})
 =\Gamma_{\mathrm{ch}}(K_t,L_t)
  -2\sqrt\ell\,\langle v,K_t^{1/2}v\rangle\sin^2\theta.
\end{equation}
On the other hand, a Schur-complement expansion of their sum gives a new
small eigenvalue
\begin{equation*}\label{eq:rank-opening-small-eigenvalue}
 \lambda_{m+1}(\widetilde K_t+
               \widetilde L_{t,\theta})
 =c_t\theta^2+O(\theta^4),
 \qquad
 c_t=\ell\left(1-
   \ell\langle v,(K_t+L_t)^{-1}v\rangle\right)>0.
\end{equation*}
The positivity of \(c_t\) follows from
\(K_t+L_t\succ \ell vv^*\).  The other eigenvalues move by
\(O(\theta^2)\).  Therefore
\begin{equation}\label{eq:rank-opening-entropy-gain}
 \Phi_{\mathrm{ch}}(\widetilde K_t,
                     \widetilde L_{t,\theta})
 =\Phi_{\mathrm{ch}}(K_t,L_t)
  -c_t\theta^2\log_2(1/\theta^2)+O(\theta^2).
\end{equation}
Choose \(\theta^2=a t\) with a fixed sufficiently small \(a>0\).  The
linear slack in \eqref{eq:coherence-boost-gamma} dominates the loss in
\eqref{eq:rank-opening-gamma-loss}, while the negative
\(t\log(1/t)\) term in \eqref{eq:rank-opening-entropy-gain} dominates the
\(O(t)\) change in \eqref{eq:boost-objective-expansion}.  The resulting
rank-\((m+1)\) pair is feasible and has strictly smaller objective.
\end{proof}


\section{The anchored moment hierarchy and finite completeness}
\label{sec:moment-hierarchy}

We now add the axis that guarantees tightness.  The construction is the
stable-set moment hierarchy augmented by every quadratic honeycomb
certificate from the chosen row library.  It is useful to formulate it
entirely in the quotient by the code constraints.

Let \(\mathsf G_{n,d}\) be the graph with vertex set
\(X=\{\pm1\}^n\), where distinct \(x,y\) are adjacent when
\(d_H(x,y)<d\).  Then
\begin{equation*}\label{eq:A-as-independence}
 \Atwo(n,d)=\alpha(\mathsf G_{n,d}).
\end{equation*}
Let \(z_x\) be a Boolean membership variable for each \(x\in X\).  The
stable-set ideal is
\begin{equation*}\label{eq:stable-set-ideal}
 I_{n,d}
 =\left\langle z_x^2-z_x:\ x\in X\right\rangle
  +\left\langle z_xz_y:\ \{x,y\}\in E(\mathsf G_{n,d})\right\rangle.
\end{equation*}
Squarefree monomials surviving in
\(\R[z]/I_{n,d}\) are indexed by codes.  In fact these monomials form a
basis: their evaluations on code indicators give the inclusion matrix
\(\one\{C\subseteq D\}\), which is unitriangular after ordering codes by
size and inclusion.  Hence the quotient is the full algebra of real-valued
functions on the finite set of code indicators; in particular, pointwise
identities on codes are identities in the quotient.  Let
\begin{equation*}\label{eq:independent-subsets}
 \cI_{\le q}
 =\{C\subseteq X:\ C\text{ is a code and }|C|\le q\},
 \qquad z^C=\prod_{x\in C}z_x.
\end{equation*}

For an order-\(t\) truncated moment sequence
\(y=(y_C)_{C\in\cI_{\le2t}}\), define the linear functional
\begin{equation*}\label{eq:moment-functional}
 L_y(z^C)=y_C,
\end{equation*}
extended after Boolean reduction, with every monomial containing a forbidden
pair mapped to zero.  Its moment matrix is
\begin{equation*}\label{eq:moment-matrix}
 M_t(y)_{I,J}=y_{I\cup J},
 \qquad I,J\in\cI_{\le t},
\end{equation*}
again using zero when the union is not a code.

If \(p(z)=\sum_{|A|\le2}p_Az^A\) is a quadratic polynomial, its order
\(t-1\) localizing matrix is
\begin{equation*}\label{eq:localizing-matrix}
 M_{t-1}(p\,y)_{I,J}
 =L_y\bigl(z^{I\cup J}p(z)\bigr)
 =\sum_{|A|\le2}p_Ay_{I\cup J\cup A},
 \qquad I,J\in\cI_{\le t-1}.
\end{equation*}

\subsection{The anchored relaxation}

Fix a row level \(r\), and let \(\mathfrak C_r(n)\) be the nested
certificate library from Section~\ref{sec:row-hierarchy}.  For every
\(\mathfrak h\in\mathfrak C_r(n)\) satisfying
\(\lambda(\mathfrak h)>s_{n,d}\), form the two valid polynomials
\(Q_{\mathfrak h}\) and \(T_{\mathfrak h}\) from
\eqref{eq:Qh}--\eqref{eq:Th}.

\begin{definition}
\label{def:anchored-moment-bound}
For \(t\ge1\), let \(\Mom_{r,t}(n,d)\) be the optimum of
\begin{equation*}\label{eq:anchored-moment-objective}
 \max\ \sum_{x\in X}y_{\{x\}}
\end{equation*}
subject to
\begin{align}
 y_\varnothing&=1, \notag \\
 M_t(y)&\succeq0,   \notag \\
 M_{t-1}(Q_{\mathfrak h}y)&\succeq0
 \quad(\mathfrak h\in\mathfrak C_r(n)), \notag\\
 M_{t-1}(T_{\mathfrak h}y)&\succeq0
 \quad(\mathfrak h\in\mathfrak C_r(n)). \notag
\end{align}
The optimization is a symmetry-invariant, possibly semi-infinite SDP.  Any
finite sublibrary gives a finite SDP and remains a valid relaxation.
\end{definition}

If \(I\in\cI_{\le t}\), the principal
minor of \(M_t(y)\) indexed by \(\varnothing,I\) is
\[
 \begin{pmatrix}1&y_I\\y_I&y_I\end{pmatrix}\succeq0,
\]
so \(0\le y_I\le1\).  If a code \(C\) has \(|C|\le2t\), write
\(C=I\cup J\) with \(|I|,|J|\le t\); a second principal-minor inequality
gives \(|y_C|^2\le y_Iy_J\le1\).  Noncode moments are zero by convention.
Thus all moment variables lie in a compact box, and the feasible set is an
intersection of closed positive-semidefinite constraints.  The maximum is
therefore attained, even for the full semi-infinite localizer library.

The unaugmented constraint \(M_t(y)\succeq0\) is already the ordinary
stable-set moment bound.  The ordinary moment constraint supplies finite convergence; the additional localizers may strengthen the relaxation at lower anchor orders.

For notational continuity with the first honeycomb bound, define the seeded
two-axis bound
\begin{equation}\label{eq:seeded-two-axis}
 \HC_{r,0}(n,d)=\HC_r^{\mathrm{row}}(n,d),
 \qquad
 \HC_{r,t}(n,d)
 =\min\{\HC_r^{\mathrm{row}}(n,d),\Mom_{r,t}(n,d)\}
 \quad(t\ge1).
\end{equation}

\subsection{Soundness, monotonicity, and finite convergence}

\begin{theorem}
\label{thm:complete-hierarchy}
Let \(\alpha=\Atwo(n,d)\).  For all \(r\ge1\) and \(t\ge0\),
\begin{align}
 \alpha&\le\HC_{r,t}(n,d),                               \label{eq:hierarchy-sound}\\
 \HC_{r+1,t}(n,d)&\le\HC_{r,t}(n,d),                     \label{eq:hierarchy-row-monotone}\\
 \HC_{r,t+1}(n,d)&\le\HC_{r,t}(n,d).                    \label{eq:hierarchy-anchor-monotone}
\end{align}
Most importantly,
\begin{equation*}\label{eq:hierarchy-exact}
 \HC_{r,t}(n,d)=\Atwo(n,d)
 \qquad\text{for every }t\ge\Atwo(n,d).
\end{equation*}
Thus the anchor axis recovers the tight finite code bound at sufficiently
high level, independently of the row depth.
\end{theorem}

\begin{proof}
\emph{Soundness.}
Let \(C\subseteq X\) be any code and take the Dirac moment sequence
\begin{equation}\label{eq:Dirac-moments}
 y_I=\one\{I\subseteq C\}.
\end{equation}
Its moment matrix is rank one and positive semidefinite.  For a polynomial
\(p\) nonnegative on code indicators, the localizing matrix equals
\(p(\one_C)vv^*\) for the corresponding monomial vector \(v\), hence is
positive semidefinite.  Proposition~\ref{prop:honeycomb-polynomials} applies
to every \(Q_{\mathfrak h}\) and \(T_{\mathfrak h}\).  Therefore a maximum
code gives a feasible moment point of objective \(\alpha\).  The raw row
bound is sound by Theorem~\ref{thm:matrix-profile}, proving
\eqref{eq:hierarchy-sound}.

\emph{Monotonicity.}
Increasing \(r\) enlarges the certificate library.  It therefore decreases
the raw row bound and adds localizer constraints, so
\(\Mom_{r+1,t}(n,d)\le\Mom_{r,t}(n,d)\).  Increasing \(t\) adds moment
information; every order-\(t+1\) feasible sequence truncates to an
order-\(t\) feasible sequence, including all localizers.  Taking the minimum
in \eqref{eq:seeded-two-axis} proves
\eqref{eq:hierarchy-row-monotone}--\eqref{eq:hierarchy-anchor-monotone}.

\emph{Exactness.}
It remains to prove that \(M_t(y)\succeq0\) alone is exact once
\(t\ge\alpha\).  In the quotient algebra \(\R[z]/I_{n,d}\), define for each
code \(C\)
\begin{equation*}\label{eq:code-idempotent}
 e_C(z)
 =\sum_{\substack{J\supseteq C\\J\text{ a code}}}
   (-1)^{|J\setminus C|}z^J.
\end{equation*}
Every code \(J\) has size at most \(\alpha\), so \(e_C\) has degree at most
\(\alpha\le t\).  On a Boolean code indicator \(\one_D\), inclusion--
exclusion gives
\begin{equation*}\label{eq:idempotent-evaluation}
 e_C(\one_D)=\one\{C=D\}.
\end{equation*}
Consequently, in the quotient algebra,
\begin{equation*}\label{eq:idempotent-relations}
 e_C^2=e_C,
 \qquad e_Ce_D=0\ (C\ne D),
 \qquad \sum_Ce_C=1.
\end{equation*}
Because \(\deg e_C\le t\), moment positivity gives
\begin{equation*}\label{eq:idempotent-weights}
 \pi_C\defeq L_y(e_C)=L_y(e_C^2)\ge0,
 \qquad \sum_C\pi_C=1.
\end{equation*}
Every polynomial function on code indicators decomposes in the quotient as
\(
 f=\sum_C f(\one_C)e_C.
\)
In particular,
\begin{equation*}\label{eq:objective-mixture}
 \sum_x y_{\{x\}}
 =L_y(S)
 =\sum_C\pi_C|C|
 \le\alpha.
\end{equation*}
A maximum code attains equality by \eqref{eq:Dirac-moments}.  Therefore
\(\Mom_{r,t}(n,d)=\alpha\), and the minimum with the sound raw row bound is
also \(\alpha\).
\end{proof}


\subsection{The anchor interpretation}

The localizer has a direct coding meaning.  For an actual code indicator and
a monomial \(z^I\),
\begin{equation*}\label{eq:anchor-localizer-meaning}
 z^I Q_{\mathfrak h}(z)
 =\begin{cases}
   Q_{\mathfrak h}(\one_C),&I\subseteq C,\\
   0,&I\not\subseteq C.
  \end{cases}
\end{equation*}
Thus each matrix entry in \(M_{t-1}(Q_{\mathfrak h}y)\) conditions the
honeycomb inequality on the event that a union of at most \(2t-2\) anchor
words belongs to the code.  The same applies to the block-trace polynomial
\(T_{\mathfrak h}\).  In this sense, positive anchor depth repeatedly applies
the first-level spectral geometry inside every conditional code fiber.  This
is the precise meaning of an \emph{anchored honeycomb} bound.

The use of both localizers is stronger than merely inserting the scalar cap
\(\mathsf B(\mathfrak h;s)S-S^2\).  By
\eqref{eq:cap-polynomial-combination}, the cap localizer follows from the two
structured localizers, while the converse need not hold.

\section*{Conclusion}
\label{sec:conclusion}

We introduced the first honeycomb bound by retaining every two-row
hyperoctahedral irreducible compatible with a moving two-row stabilizer type.
The finite proof combines a profile-optimized moving projection, exact
box-transfer coefficients, and a blockwise effective dimension.  Its
F{\o}lner limit is the explicit four-parameter exponent
\(\kappa_{\mathrm{HC}}^{[1]}=\kappa_{\mathrm{HC}}\).  The one-sided and
entropy-balanced slices strictly improve the whole-cube/MQC curve, the
symmetric slice recovers the fully optimized second MRRW problem, and the
combined constant-weight branch strictly improves 2MQC throughout the open
distance interval.

The new Horn--channel construction makes higher-row asymptotic optimization
available unconditionally.  A pair of positive matrices \((K,L)\) produces
an explicit output-symmetric channel whose PGM statistic is
\(2\tr(K^{1/2}L^{1/2})\) and whose uniform-prior Holevo information is
\(\mathsf S(K)+\mathsf S(L)-\mathsf S(K+L)\).  The spectra of
\((K,L,K+L)\) range over exactly the Horn--hive domain, and the uniform-prior Holevo information is
exactly the higher-row entropy potential.  The resulting exponents are
nested with the matrix size: the scalar level is \(M_1\), the
\(2\times2\) level is exactly the first honeycomb exponent, and the
\(3\times3\) level is a new explicit rate bound no weaker than
\(\kappa_{\mathrm{HC}}\).  The rank-opening proposition explains why a
regular optimizer should improve strictly after one more dimension, and the
rank-three numerics display such decreases.  Certifying the global numerical
curve by interval branch-and-bound is a concrete next step.

The paper also places these asymptotic bounds inside a complete finite
hierarchy.  The representation axis allows successively more rows and
replaces scalar Littlewood--Richardson transitions by matrix-valued hive
transfers.  The frame-profile theorem proves a moving-projection bound at
arbitrary multiplicity and recovers the current scalar profile theorem as
its rank-one special case.  The anchor axis localizes the resulting positive
quadratic kernels in the stable-set moment hierarchy.  These finite bounds
are monotone in both axes and recover \(\Atwo(n,d)\) exactly once the anchor
order is at least \(\Atwo(n,d)\).  The same finite construction extends to
every alphabet size \(q\) through \(C_q\wr S_n\).

These results are just an initial proof-of-concept that the honeycomb framework perspective yields stronger new code bounds already at lower-levels. We leave a much more comprehensive analysis to future versions of this paper. 

\section*{Acknowledgment}
The authors mainly used GPT-5.6 Pro for exploratory analysis, idea testing, and editorial polishing. The authors reviewed and finalized all mathematical claims and proofs, and take full responsibility for the correctness, exposition, and attribution in the final manuscript. The authors used previous versions of GPT 5.x, but only 5.6 materially helped for code bounds.

\appendix

\section{Two-row multiplicities, Racah coefficients, and entropy asymptotics}
\label{app:two-row-computations}

\subsection{Two-row Littlewood--Richardson multiplicities}
\label{app:lr}

Let
\[
 \lambda=(p+r,p),\qquad \mu=(q+s,q).
\]
As polynomial \(GL_2\)-representations,
\[
 \mathbf S_\lambda(\C^2)
 \cong\det^p\otimes\Sym^r(\C^2),
 \qquad
 \mathbf S_\mu(\C^2)
 \cong\det^q\otimes\Sym^s(\C^2).
\]
The Clebsch--Gordan decomposition is
\[
 \Sym^r(\C^2)\otimes\Sym^s(\C^2)
 \cong
 \bigoplus_{t=0}^{\min\{r,s\}}
 \det^t\otimes\Sym^{r+s-2t}(\C^2).
\]
The term indexed by \(t\) has partition
\[
 (p+q+r+s-t,\ p+q+t).
\]
Hence every two-row Littlewood--Richardson coefficient is zero or one.  In
the notation of Proposition~\ref{prop:complete-region},
\(r=A_0\), \(s=B_0\), and the desired second row \(k\) requires
\(t=k-j-\ell\).  This gives
\[
 0\le k-j-\ell\le\min\{A_0,B_0\}.
\]
Multiplying by two and rearranging gives
\(|A_0-B_0|\le E_0\le A_0+B_0\).

\subsection{Half-spin Racah coefficients and their asymptotics}
\label{app:racah}

The Wigner recoupling identity is
\begin{align*}
 \bigl|((a\,\tfrac12)c,\,d)e\bigr\rangle
 =\sum_f(-1)^{a+d+e+1/2}
 \sqrt{(2c+1)(2f+1)}
 \begin{Bmatrix}a&\tfrac12&c\\d&e&f\end{Bmatrix}
 \bigl|(a,(\tfrac12\,d)f)e\bigr\rangle.
\end{align*}
For fixed external spins \(a,d,e\), the two admissible values of \(c\) and
\(f\) form a \(2\times2\) real orthogonal recoupling matrix.  The half-spin
specialization of the Racah factorial formula has a single surviving summand.
For the entry \(c=a+\tfrac12\), \(f=d+\tfrac12\), it gives
\begin{equation*}\label{eq:half-spin-6j-square}
 \begin{Bmatrix}
  a&\tfrac12&a+\tfrac12\\
  d&e&d+\tfrac12
 \end{Bmatrix}^{\!2}
 =\frac{(e+\tfrac12)^2-(a-d)^2}
 {(2a+1)(2a+2)(2d+1)(2d+2)}.
\end{equation*}
This identity follows by substituting \(b=\tfrac12\) into the standard Racah
sum \cite[Section~9.6]{Varshalovich1988}; the four coupling triangles ensure
that every factorial argument is a nonnegative integer.  Multiplication by
\((2c+1)(2f+1)=(2a+2)(2d+2)\) yields
\begin{equation}\label{eq:Q-half-spin-appendix}
 Q(a,d,e)=
 \frac{(e+\tfrac12)^2-(a-d)^2}{(2a+1)(2d+1)}.
\end{equation}
A real \(2\times2\) orthogonal matrix has squared entries of the form
\(\bigl(\begin{smallmatrix}Q&1-Q\\1-Q&Q\end{smallmatrix}\bigr)\), up to
interchanging the two labels.  Thus \eqref{eq:Q-half-spin-appendix} determines
all four squared recoupling coefficients.  Notice also that
\[
 1-Q(a,d,e)
 =\frac{(a+d+1)^2-(e+\tfrac12)^2}{(2a+1)(2d+1)},
\]
so the coupling triangle inequalities make both quantities nonnegative.
Degenerate triangle faces follow by continuity.

For the four channels leaving one fixed source triple, the external spin is
\(a=a_\sigma\); hence the two source rows need not share one finite value of
\(Q\).  Nevertheless, along the scaling in \eqref{eq:scaling}, uniformly in
\(\sigma\),
\[
 \frac{2a_\sigma}{n}\to A,\qquad
 \frac{2d}{n}\to B,\qquad
 \frac{2e}{n}\to E.
\]
Therefore
\[
 Q(a_\sigma,d,e)
 \longrightarrow\frac{E^2-(A-B)^2}{4AB}=P.
\]
It follows that the two same-row channels converge to \(P\), while the two
cross-row channels converge to \(1-P\), proving \eqref{eq:R-limit}.  The
same conclusion on a triangle face follows by one-sided continuity.

\subsection{Entropy asymptotics}
\label{app:entropy}

From the hook-length formula,
\[
 f^{(m-r,r)}
 =\frac{m-2r+1}{m-r+1}\binom mr
 =\binom mr-\binom m{r-1}.
\]
If \(r/m\to\theta\le1/2\), then
\[
 \frac1m\log_2 f^{(m-r,r)}\to\Htwo(\theta).
\]
Indeed, the hook prefactor lies between $1/(m+1)$ and $1$, so it is
subexponential even at the equal-row endpoint.  Combining this with
\(\binom ni=2^{(\Htwo(x)+o(1))n}\) proves
\eqref{eq:ambient-exp}.  The same formula with \(m=n\), \(r=k\) proves
\eqref{eq:stabilizer-exp}.  Polynomial factors from the hook-length
prefactors and the number of vertices in a F{\o}lner box disappear after division
by \(n\).

\section{Local stability of the MRRW slice}
\label{app:mrrw-stability}

This appendix records a local calculation that is not needed for the
honeycomb rate bound or for any analytic comparison theorem.  It diagnoses
when an asymmetric entropy-balanced perturbation decreases the objective
relative to the symmetric slice, whose infimum is exactly \(M_2\).

\begin{proposition}
\label{prop:MRRW-second-variation}
Fix \(0<\tau<s=1-2\delta\), put
\begin{equation*}\label{eq:second-var-basics}
 t=\arcsin\tau,
 \qquad u_0=\frac{1-\cos t}{2},
 \qquad E_0=\sqrt{(1+\tau)(s-\tau)},
 \qquad z_0=\frac{1-E_0}{2},
\end{equation*}
and define
\begin{align}
 C_\delta(t)=-2\bigl[&s t^2+s\sin t+s-t^2\sin t+t^2\cos(2t)-t^2
 \nonumber\\[-1mm]
 &{}-2t\sin(2t)-2\sin t-2\bigr], \notag \\
 \mathcal Q_\delta(\tau)
 ={}&-\frac{2(1+t^2)}{\ln2}
 +\Htwo'(u_0)\bigl(t^2\cos t-2t\sin t\bigr)
 +\Htwo'(z_0)\frac{C_\delta(t)}{4E_0}.\notag
\end{align}
On the balanced spectral boundary, with \(x=1/2+y\) and fixed \(t\),
\begin{equation}\label{eq:second-var-expansion}
 \Phi_{\mathrm{bal}}\!\left(\frac12+y,t;\delta\right)
 =F_\delta(\tau)+\mathcal Q_\delta(\tau)y^2+O(y^4).
\end{equation}
In particular, if an interior minimizer \(\tau_*\) of \(F_\delta\) satisfies
\(\mathcal Q_\delta(\tau_*)<0\), then
\begin{equation*}\label{eq:strict-below-M2-criterion}
 \kappa_{\mathrm{bal}}(\delta)<M_2(\delta).
\end{equation*}
\end{proposition}

\begin{proof}
The symmetry \(x\leftrightarrow1-x\) makes the expansion even in \(y\).
A direct Taylor expansion of \eqref{eq:E-balanced} gives
\[
 E_{\mathrm{bal}}(1/2+y,t;\delta)^2
 =E_0^2+C_\delta(t)y^2+O(y^4).
\]
Consequently
\(z_{\mathrm{bal}}=z_0-C_\delta(t)y^2/(4E_0)+O(y^4)\).
Also
\[
 \Htwo(1/2+y)=1-\frac{2}{\ln2}y^2+O(y^4).
\]
Writing \(f(y)=\Htwo(\sin^2((1/2+y)t))\), the weighted sum of the two
second-row entropies is
\[
 (1/2-y)f(y)+(1/2+y)f(-y)
 =f(0)+\left[
   -\frac{2t^2}{\ln2}
   +\Htwo'(u_0)(t^2\cos t-2t\sin t)
  \right]y^2+O(y^4).
\]
Combining these three expansions proves
\eqref{eq:second-var-expansion}.  A negative quadratic coefficient at a
global interior minimizer gives a nearby asymmetric point with smaller
objective.
\end{proof}

\section{Numerical experiments}
\label{app:numerics}

The computations in this appendix are diagnostic only: no numerical value is
used in a proof.  The reference implementation used CPython~3.13.5,
NumPy~2.3.5, and SciPy~1.17.0.  We evaluated the one-dimensional MRRW and
whole-cube formulas by bounded scalar minimization.  For each first-level
constrained multivariate problem, we generated 1024 scrambled Sobol starting
points with seed $20260811$, retained the best feasible candidates, and
polished the best 64 with SLSQP at objective and constraint tolerances
$10^{-13}$.  The balanced search used angle coordinates; the face,
full-honeycomb, and constant-weight searches used direct coordinates together
with boundary-resolving charts.  The rank-three Horn--channel search used the
real contraction parametrization
\eqref{eq:channel-contraction-parametrization}, differential-evolution
multistarts followed by constrained SLSQP polishing, and continuation in
\(\delta\).  Every reported point was re-evaluated from the resulting
matrices rather than from the penalized objective.

The values were stable under these checks, but they are not interval
certificates and do not prove global optimality of the competing
multidimensional infima.  No numerical value in the table is used in any
theorem.

\subsection{Benchmark table}

Table~\ref{tab:benchmark} compares the classical MRRW exponents, the
previous combined exponent
\(\kappa_{\mathrm{bin}}=\min\{\kappa_H,\kappa_{\mathrm{CW}}\}\), the
entropy-balanced branch, and the updated explicit combination
\(\min\{\kappa_{\mathrm{bin}},\kappa_{\mathrm{bal}}\}\).  The final column is
the observed decrease relative to the previous combined exponent.

\begin{table}[ht]
\centering
\scriptsize
\setlength{\tabcolsep}{3.2pt}
\begin{tabular}{@{}ccccccc@{}}
\toprule
\(\delta\)&\(M_1\)&\(M_2\)&previous \(\kappa_{\mathrm{bin}}\)
&\(\kappa_{\mathrm{bal}}\)&updated&decrease\\
\midrule
0.050 & 0.858235875 & 0.825136808 & 0.825114800 & 0.825136808 & 0.825114800 & 0\\
0.075 & 0.789354961 & 0.756662514 & 0.756586983 & 0.756662514 & 0.756586983 & 0\\
0.100 & 0.721928095 & 0.692740743 & 0.692557458 & 0.692740743 & 0.692557458 & 0\\
0.125 & 0.656057563 & 0.631945316 & 0.631576981 & 0.631945316 & 0.631576981 & 0\\
0.150 & 0.591857407 & 0.573450044 & 0.572792363 & 0.573334649 & 0.572792363 & 0\\
0.175 & 0.529455417 & 0.516715194 & 0.515631887 & 0.515637323 & 0.515631887 & 0\\
0.200 & 0.468995594 & 0.461359604 & 0.459676098 & 0.458906663 & 0.458906663 & \(7.69\!\times\!10^{-4}\)\\
0.225 & 0.410641241 & 0.407098656 & 0.404594036 & 0.403463376 & 0.403463376 & \(1.13\!\times\!10^{-3}\)\\
0.250 & 0.354578903 & 0.353710543 & 0.350379084 & 0.349630107 & 0.349630107 & \(7.49\!\times\!10^{-4}\)\\
0.275 & 0.301023462 & 0.301023462 & 0.298166311 & 0.297740747 & 0.297740747 & \(4.26\!\times\!10^{-4}\)\\
0.300 & 0.250224912 & 0.250224912 & 0.248376065 & 0.248150022 & 0.248150022 & \(2.26\!\times\!10^{-4}\)\\
0.350 & 0.158132937 & 0.158132937 & 0.157505576 & 0.157457709 & 0.157457709 & \(4.79\!\times\!10^{-5}\)\\
0.400 & 0.081468915 & 0.081468915 & 0.081336612 & 0.081331365 & 0.081331365 & \(5.25\!\times\!10^{-6}\)\\
0.450 & 0.025266128 & 0.025266128 & 0.025257384 & 0.025257272 & 0.025257272 & \(1.12\!\times\!10^{-7}\)\\
\bottomrule
\end{tabular}
\caption{Floating-point comparison of the classical MRRW exponents, the
previous combined binary exponent, and the update obtained by adding the
entropy-balanced branch.  Rigorously, the updated exponent is no larger than
the previous one.  Positive entries in the final column compare independently
optimized infima and are numerical rather than certified global inequalities.}
\label{tab:benchmark}
\end{table}

Several features are visible.  At small distances the balanced optimization
returns the symmetric \(M_2\) value to the displayed precision, as predicted
by the exact embedding.
The second-variation coefficient in
Proposition~\ref{prop:MRRW-second-variation} changes sign numerically at
\begin{equation*}\label{eq:numeric-Q-threshold}
 \delta\approx0.13959776.
\end{equation*}
Beyond that point, symmetry breaking produces a value below the MRRW slice.
The constant-weight branch remains smaller over an additional interval; the
balanced branch appears to become active in the combined minimum near
\(\delta\approx0.1752\).  Its largest displayed gain is about
\(1.13\times10^{-3}\) bits per coordinate at \(\delta=0.225\), corresponding
to an asymptotic factor \(2^{0.00113n+o(n)}\) in the code-size upper bound.
The gain then decreases smoothly toward \(\delta=1/2\).

\subsection{Face, balanced, and full honeycomb optimization}

Table~\ref{tab:internal-hyp} compares the one-sided face, the exact balanced
slice, and the full spectral-boundary optimization.  In the displayed range,
the balanced slice captures nearly all of the numerically observed
full-honeycomb improvement.

\begin{table}[ht]
\centering
\small
\begin{tabular}{@{}ccccc@{}}
\toprule
\(\delta\)&face&balanced&full&balanced minus full\\
\midrule
0.200 & 0.459649170 & 0.458906663 & 0.458906624 & \(3.83\times10^{-8}\)\\
0.225 & 0.403836456 & 0.403463376 & 0.403463367 & \(8.98\times10^{-9}\)\\
0.250 & 0.349804834 & 0.349630107 & 0.349630106 & \(1.78\times10^{-9}\)\\
0.275 & 0.297816183 & 0.297740747 & 0.297740746 & \(2.98\times10^{-10}\)\\
0.300 & 0.248179593 & 0.248150022 & 0.248150022 & \(4.16\times10^{-11}\)\\
\bottomrule
\end{tabular}
\caption{Internal floating-point comparison of the new honeycomb
subfamilies.}
\label{tab:internal-hyp}
\end{table}

Rounded balanced minimizers \((x,t,\eta,\rho,z)\), with
\(\eta=\sin^2(xt)\) and \(\rho=\sin^2((1-x)t)\), are
\begin{align*}
 \delta=0.200:&\quad
 (0.14430788,\ 0.11769157,\ 2.88422\!\times10^{-4},\
  0.01010781,\ 0.02181697),\\
 \delta=0.225:&\quad
 (0.10755199,\ 0.09544193,\ 1.05366\!\times10^{-4},\
  0.00723759,\ 0.01251767),\\
 \delta=0.250:&\quad
 (0.08143856,\ 0.07607063,\ 3.83786\!\times10^{-5},\
  0.00487465,\ 0.00722648),\\
 \delta=0.275:&\quad
 (0.06173582,\ 0.05950942,\ 1.34972\!\times10^{-5},\
  0.00311437,\ 0.00413044),\\
 \delta=0.300:&\quad
 (0.04633895,\ 0.04552662,\ 4.45064\!\times10^{-6},\
  0.00188385,\ 0.00230140).
\end{align*}
The small difference between the balanced and full columns should not be
interpreted as an exact equality.  It only suggests that the entropy-balanced
condition tracks the full optimizer closely in this range.

\subsection{Rank-three Horn--channel optimization}
\label{app:rank-three-numerics}

For the real rank-three search, write
\begin{equation*}\label{eq:rank-three-numeric-param}
 G=R\diag(\gamma_1,\gamma_2,\gamma_3)R^{\mathsf T},
 \qquad
 D=\diag(a_1,a_2,a_3),
 \qquad
 K=G^{1/2}DG^{1/2},
 \qquad
 L=G-K,
\end{equation*}
where \(\gamma\) is a probability vector, \(0\le a_i\le1\), and
\(R=R_{12}(\theta_1)R_{13}(\theta_2)R_{23}(\theta_3)\in SO(3)\).  This is
an eight-variable subfamily of the full complex rank-three problem.  It is
already sufficient to improve the displayed first-level curve at most of the
sampled distances.

\begin{table}[ht]
\centering
\small
\setlength{\tabcolsep}{4pt}
\begin{tabular}{@{}ccccc@{}}
\toprule
\(\delta\)&previous displayed combination&rank-three channel&new displayed combination&decrease\\
\midrule
0.100 & 0.692557458 & 0.692550953 & 0.692550953 & \(6.51\times10^{-6}\)\\
0.125 & 0.631576981 & 0.631575726 & 0.631575726 & \(1.26\times10^{-6}\)\\
0.150 & 0.572792363 & 0.572823014 & 0.572792363 & 0\\
0.175 & 0.515631887 & 0.515324055 & 0.515324055 & \(3.08\times10^{-4}\)\\
0.200 & 0.458906663 & 0.458749320 & 0.458749320 & \(1.57\times10^{-4}\)\\
0.225 & 0.403463376 & 0.403391927 & 0.403391927 & \(7.14\times10^{-5}\)\\
0.250 & 0.349630107 & 0.349599729 & 0.349599729 & \(3.04\times10^{-5}\)\\
0.275 & 0.297740747 & 0.297728695 & 0.297728695 & \(1.21\times10^{-5}\)\\
0.300 & 0.248150022 & 0.248145631 & 0.248145631 & \(4.39\times10^{-6}\)\\
0.350 & 0.157457709 & 0.157457310 & 0.157457310 & \(3.99\times10^{-7}\)\\
\bottomrule
\end{tabular}
\caption{Floating-point rank-three Horn--channel values and their combination
with the previously displayed curve.  The exact theorem is the domination
\(\kappa_{\mathrm{ch}}^{[2]}\le\kappa_{\mathrm{HC}}\); the strict decimal
decreases compare independently optimized floating-point values and are not
global interval certificates.}
\label{tab:rank-three-channel}
\end{table}
\FloatBarrier

\subsection[Log-scale gains relative to M2]{Log-scale gains relative to \(M_2\)}

For any exponent-valued upper bound \(B(\delta)\), define its gain over the
second MRRW exponent by
\begin{equation*}\label{eq:gain-over-M2}
 \Delta_B(\delta)\defeq M_2(\delta)-B(\delta).
\end{equation*}
A larger positive value of \(\Delta_B\) means a stronger asymptotic upper
bound.  Figure~\ref{fig:gain-log-M2} plots these gains on a logarithmic scale
on the sampled grid \(\delta\in[0.10,0.35]\) used throughout this
appendix, comparing
\(\kappa_{\mathrm{bal}}\), \(\kappa_{\mathrm{HC}}\), MQC,
2MQC, \(\kappa_{\mathrm{ch}}^{[2]}\), OpenAI
\(\kappa_{\mathrm{CW}}\), \(\kappa_H\), and the previous combined branch
\(\min\{\kappa_{\mathrm{CW}},\kappa_H\}\).
The curves for MQC and \(\kappa_H\) coincide analytically by
Proposition~\ref{prop:app-MQC-equivalence}, while the numerically optimized
first honeycomb curve \(\kappa_{\mathrm{HC}}\) is visually almost
indistinguishable from the explicit balanced slice on this grid.  Because the
vertical axis is logarithmic, points whose computed gains are below
\(10^{-10}\) are omitted from the display.

The one-dimensional quantities \(M_2\), MQC/\(\kappa_H\), and the 2MQC curve
were evaluated directly from the compact formulas in \eqref{eq:M2},
\eqref{eq:app-RMQC}, and \eqref{eq:app-R2MQC}.  The multidimensional curves
\(\kappa_{\mathrm{bal}}\), \(\kappa_{\mathrm{HC}}\), and OpenAI
\(\kappa_{\mathrm{CW}}\) were computed with the same floating-point
optimization pipeline described at the beginning of this appendix, and the
\(\kappa_{\mathrm{ch}}^{[2]}\) values are those of
Table~\ref{tab:rank-three-channel}.  The plotted values of the previous
combined branch are taken from Table~\ref{tab:benchmark}, so the figure and
the benchmark table use the same rounded data.  The figure makes the hierarchy of gains
transparent: the previous OpenAI combination
\(\min\{\kappa_{\mathrm{CW}},\kappa_H\}\) dominates the individual OpenAI
branches, the first honeycomb curves improve it in the middle-distance range,
and the sampled rank-three Horn--channel curve yields a further gain at most of
the displayed distances.

\begin{figure}[ht]
\centering
\definecolor{mplblue}{RGB}{31,119,180}
\definecolor{mplorange}{RGB}{255,127,14}
\definecolor{mplgreen}{RGB}{44,160,44}
\definecolor{mplred}{RGB}{214,39,40}
\definecolor{mplpurple}{RGB}{148,103,189}
\definecolor{mplbrown}{RGB}{140,86,75}
\definecolor{mplpink}{RGB}{227,119,194}
\definecolor{mplgray}{RGB}{127,127,127}
\begin{tikzpicture}
\begin{semilogyaxis}[
 width=0.96\textwidth,
 height=0.57\textwidth,
 xmin=0.0875,
 xmax=0.3625,
 ymin=8e-6,
 ymax=5e-3,
 xlabel={Relative distance \(\delta\)},
 ylabel={Gain over \(M_2\): \(M_2(\delta)-B(\delta)\)},
 title={Log-scale gains of several bounds relative to \(M_2\)},
 grid=both,
 major grid style={gray!35},
 minor grid style={gray!15},
 tick align=outside,
 xtick={0.10,0.15,0.20,0.25,0.30,0.35},
 mark size=1.55pt,
 every axis plot/.append style={line width=0.9pt},
 legend columns=2,
 legend style={
   at={(0.012,0.988)},
   anchor=north west,
   draw=gray!55,
   fill=white,
   fill opacity=0.92,
   text opacity=1,
   font=\scriptsize,
   column sep=0.8em,
   row sep=0.1em
 },
]
\addplot+[color=mplblue,mark=*] coordinates {
 (0.150,1.15395e-4) (0.175,1.077871e-3) (0.200,2.452941e-3)
 (0.225,3.635280e-3) (0.250,4.080436e-3) (0.275,3.282715e-3)
 (0.300,2.074890e-3) (0.350,6.75228e-4)
};
\addlegendentry{\(\kappa_{\mathrm{bal}}\)}
\addplot+[color=mplorange,mark=square*] coordinates {
 (0.150,1.15596e-4) (0.175,1.077995e-3) (0.200,2.452980e-3)
 (0.225,3.635289e-3) (0.250,4.080437e-3) (0.275,3.282716e-3)
 (0.300,2.074890e-3) (0.350,6.75227e-4)
};
\addlegendentry{\(\kappa_{\mathrm{HC}}\)}
\addplot+[color=mplgreen,mark=triangle*] coordinates {
 (0.200,4.591965e-4) (0.225,2.386224e-3) (0.250,3.331459e-3)
 (0.275,2.857151e-3) (0.300,1.848847e-3) (0.350,6.273607e-4)
};
\addlegendentry{MQC}
\addplot+[color=mplred,mark=diamond*] coordinates {
 (0.100,1.336424e-5) (0.125,3.785161e-5) (0.150,9.333916e-5)
 (0.175,2.129799e-4) (0.200,4.760511e-4) (0.225,2.386224e-3)
 (0.250,3.331459e-3) (0.275,2.857151e-3) (0.300,1.848847e-3)
 (0.350,6.273607e-4)
};
\addlegendentry{2MQC}
\addplot+[color=mplpurple,mark=pentagon*] coordinates {
 (0.100,1.89790e-4) (0.125,3.69590e-4) (0.150,6.27030e-4)
 (0.175,1.391139e-3) (0.200,2.610284e-3) (0.225,3.706729e-3)
 (0.250,4.110814e-3) (0.275,3.294767e-3) (0.300,2.079281e-3)
 (0.350,6.75627e-4)
};
\addlegendentry{\(\kappa_{\mathrm{ch}}^{[2]}\)}
\addplot+[color=mplbrown,mark=otimes*] coordinates {
 (0.100,1.83285e-4) (0.125,3.68335e-4) (0.150,6.57681e-4)
 (0.175,1.083307e-3) (0.200,1.683506e-3) (0.225,2.504620e-3)
 (0.250,3.331459e-3) (0.275,2.857151e-3) (0.300,1.848847e-3)
 (0.350,6.27361e-4)
};
\addlegendentry{OpenAI \(\kappa_{\mathrm{CW}}\)}
\addplot+[color=mplpink,mark=triangle*,mark options={rotate=180}] coordinates {
 (0.200,4.591965e-4) (0.225,2.386224e-3) (0.250,3.331459e-3)
 (0.275,2.857151e-3) (0.300,1.848847e-3) (0.350,6.273607e-4)
};
\addlegendentry{\(\kappa_H\)}
\addplot+[color=mplgray,mark=*,line width=1.1pt] coordinates {
 (0.100,1.83285e-4) (0.125,3.68335e-4) (0.150,6.57681e-4)
 (0.175,1.083307e-3) (0.200,1.683506e-3) (0.225,2.504620e-3)
 (0.250,3.331459e-3) (0.275,2.857151e-3) (0.300,1.848847e-3)
 (0.350,6.27361e-4)
};
\addlegendentry{\(\min\{\kappa_{\mathrm{CW}},\kappa_H\}\)}
\end{semilogyaxis}
\end{tikzpicture}
\caption{Logarithmic-scale plot of the gains
\(\Delta_B(\delta)=M_2(\delta)-B(\delta)\) for several bounds relative to
\(M_2\).  The sampled curves compare the balanced and full first-level
honeycomb exponents, MQC, 2MQC, the rank-three Horn--channel exponent
\(\kappa_{\mathrm{ch}}^{[2]}\), OpenAI \(\kappa_{\mathrm{CW}}\),
\(\kappa_H\), and the previous combined branch
\(\min\{\kappa_{\mathrm{CW}},\kappa_H\}\).  Larger values correspond to
stronger improvements over the second MRRW exponent.}
\label{fig:gain-log-M2}
\end{figure}
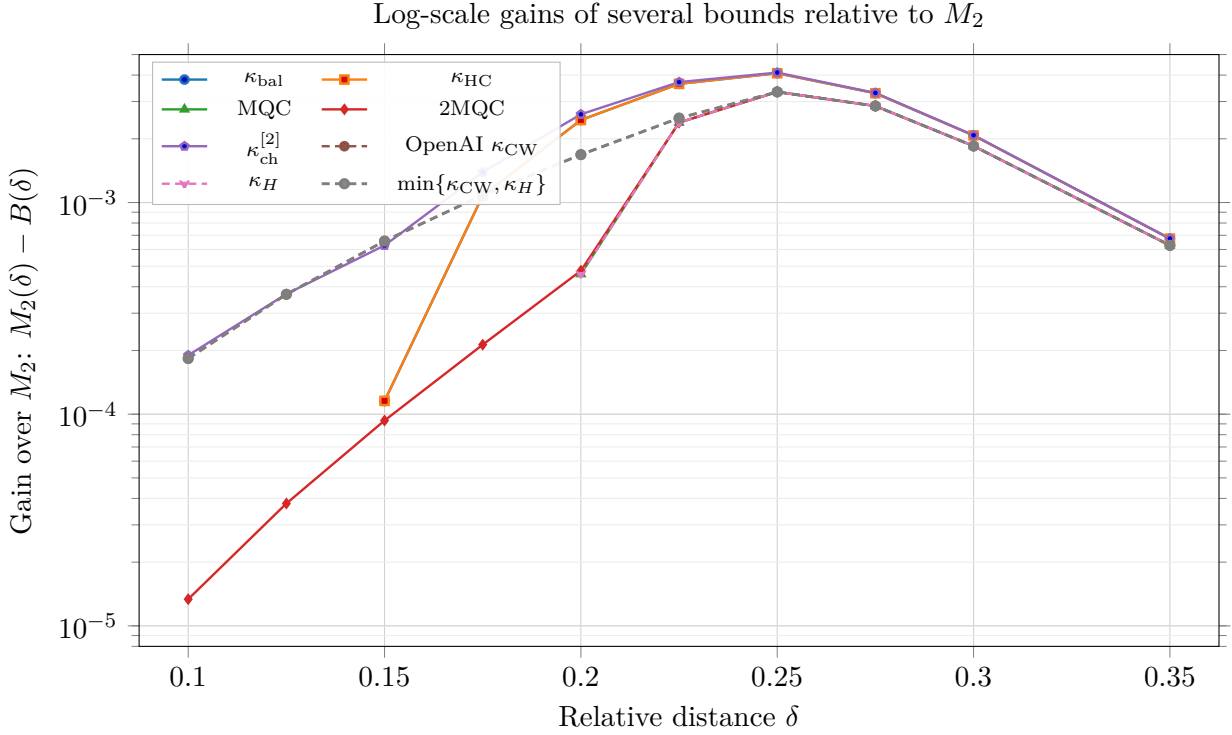
\FloatBarrier

As a concrete feasibility check near \(\delta=0.2\), the rounded matrices
\begin{align}\label{eq:rank-three-witness-K}
 K={}&\begin{pmatrix}
 0.0047195300&-0.0160777502&-0.0150542702\\
 -0.0160777502&0.0680587007&0.0722524707\\
 -0.0150542702&0.0722524707&0.0813122908
 \end{pmatrix},\notag\\[1mm]
 L={}&\begin{pmatrix}
 0.0003146000&0.0024544000&0.0084802901\\
 0.0024544000&0.0806689808&0.2477953725\\
 0.0084802901&0.2477953725&0.7649258976
 \end{pmatrix}\notag
\end{align}
have total trace one after the displayed normalization and are positive
definite to floating precision.  Direct diagonalization gives
\begin{equation*}\label{eq:rank-three-witness-values}
 \frac{1-\Gamma_{\mathrm{ch}}(K,L)}2
 \approx0.199999382,
 \qquad
 \Phi_{\mathrm{ch}}(K,L)\approx0.458750724.
\end{equation*}
The small margin in the error constraint was introduced by the coherence
boost \eqref{eq:coherence-boost}; it makes the displayed rounded witness
numerically feasible rather than merely boundary feasible.  A formal decimal
certificate would require interval bounds for the eigenvalues and matrix
square roots, which are not claimed here.

\subsection{Minimal evaluation formulas}

The full branch uses an endpoint-complete spectral reduction.  Put
\begin{align}
 S_+&=\sqrt{x(1-x)}
 \left(\sqrt{(1-\eta)(1-\rho)}+\sqrt{\eta\rho}\right),\nonumber\\
 S_-&=\sqrt{x(1-x)}
 \left(\sqrt{\rho(1-\eta)}+\sqrt{\eta(1-\rho)}\right).\notag
\end{align}
Because $\eta,\rho<1/2$, one has $S_+>S_-$.  For fixed
$(x,\eta,\rho)$, the objective is increasing in $P$, so the smallest feasible
value is
\begin{equation*}\label{eq:P-full-piecewise}
 P_*=
 \begin{cases}
 0,&(1-2\delta)/2\le S_-,\\[1mm]
 \dfrac{(1-2\delta)/2-S_-}{S_+-S_-},
   &S_-<(1-2\delta)/2\le S_+,\\[3mm]
 \text{infeasible},&(1-2\delta)/2>S_+.
 \end{cases}
\end{equation*}
At a feasible point,
\begin{equation*}\label{eq:full-boundary-reduction}
 E=\sqrt{(A-B)^2+4ABP_*},
 \qquad z=\frac{1-E}{2}.
\end{equation*}
This piecewise form includes the $P=0$ endpoint; the middle line is the
formula used at all full-branch minimizers displayed in
Table~\ref{tab:internal-hyp}.
A certified strict comparison with \(\kappa_{\mathrm{CW}}\) over a
continuum of distances would require interval evaluation and a global
branch-and-bound certificate for both variational problems.  No such
certificate is assumed here.

\section{Comparison with the MQC and 2MQC bounds}
\label{app:channel-comparison}

This appendix contains all comparisons with the classical--quantum channel
bounds of Alrabiah and Guruswami~\cite{AlrabiahGuruswami2026}.  First we
identify the whole-cube exponent exactly with the unmasked mixed-qubit-channel
(MQC) curve.  For 2MQC, we compactify the normalized variational problem and
split according to the optimized mask parameter $A$.  When $A<1$, the
minimizer maps isospectrally to the constant-weight boundary and admits a
strict improving deformation.  When $A=1$, the branch is precisely MQC and
is strictly improved by the entropy-balanced honeycomb perturbation.  This
two-case argument is analytic and exhaustive.

\subsection{Exact identification of the whole-cube and MQC curves}

Put
\begin{equation*}\label{eq:app-channel-J}
 \Jmap(q)=\frac{1-\sqrt{1-q}}2,
 \qquad g(q)=\Htwo(\Jmap(q)),
 \qquad 0\le q\le1.
\end{equation*}
The normalized MQC exponent of
\cite[Corollary~23]{AlrabiahGuruswami2026} is
\begin{equation}\label{eq:app-RMQC}
 \RMQC(\delta)
 =\min_{1-2\delta\le u\le1}
 \left\{g(u^2)-g\bigl(u(u-(1-2\delta))\bigr)\right\}.
\end{equation}

\begin{proposition}
\label{prop:app-MQC-equivalence}
For every $0<\delta<1/2$,
\begin{equation}\label{eq:app-MQC-equivalence}
 \kappa_H(\delta)=\RMQC(\delta).
\end{equation}
More explicitly, if $\sigma=1-2\delta$ and $u\in[\sigma,1]$, the closed
whole-cube spectral boundary is
\begin{equation}\label{eq:app-MQC-boundary-map}
 a=\Jmap(u^2),
 \qquad
 b=\Jmap\bigl(u(u-\sigma)\bigr).
\end{equation}
At this point $\Gamma_H(a,b)=\sigma$ and the objective equals the integrand
in~\eqref{eq:app-RMQC}.
\end{proposition}

\begin{proof}
For $0\le b<a\le1/2$, set
\[
 u=2\sqrt{a(1-a)},
 \qquad
 v=2\sqrt{b(1-b)}.
\]
Since $4\Jmap(q)(1-\Jmap(q))=q$, this is a bijective change of variables and
\begin{equation*}\label{eq:app-entropy-radius}
 \Htwo(a)-\Htwo(b)=g(u^2)-g(v^2).
\end{equation*}
Moreover,
\[
 (a-b)(1-a-b)=a(1-a)-b(1-b),
\]
so
\begin{equation*}\label{eq:app-Gamma-radius}
 \Gamma_H(a,b)=\frac{u^2-v^2}{u}.
\end{equation*}
Writing $\sigma=1-2\delta$, the strict whole-cube constraint is therefore
\[
 \sigma<u\le1,
 \qquad
 0\le v^2<u(u-\sigma).
\]
For fixed $u$, the objective decreases as $v^2$ increases, and hence its
infimum is the boundary value
$g(u^2)-g(u(u-\sigma))$.  This value is approached from the strict region.
The resulting continuous one-variable objective has the same infimum for
$\sigma<u\le1$ as on the compact interval $[\sigma,1]$, proving
\eqref{eq:app-MQC-equivalence}.  Equation~\eqref{eq:app-MQC-boundary-map}
follows from the inverse change of variables.
\end{proof}

\subsection{The normalized 2MQC formula and compactification}

For $0\le v\le u\le1$, define
\begin{align}
 D(u,v)&=u^2(1-v^2)+v^2(u+v^2), \notag                        \\
 \lambda_\delta(u,v)
 &=\frac{(u^2-v^2)^2+2\delta u^2(1+u)}{D(u,v)},          \label{eq:app-2mqc-lambda}\\
 A_\delta(u,v)&=\min\{1,\lambda_\delta(u,v)\},           \label{eq:app-2mqc-A}\\
 \Psi_\delta(u,v)
 &=1+g(u^2)-g(A_\delta(u,v))-g(v^2).\notag
\end{align}
At $(u,v)=(0,0)$ we use the continuous extension
$\lambda_\delta(0,0)=2\delta$.  In the normalized variables of
\cite[Equations~(43)--(46), Corollary~24]{AlrabiahGuruswami2026}, the
2MQC exponent is
\begin{equation}\label{eq:app-R2MQC}
 \RtwoMQC(\delta)
 =\inf_{\substack{0\le v\le u\le1\\
 v^2(u^2-v^2)<2\delta u^2(1-v^2)}}
 \Psi_\delta(u,v).
\end{equation}
For $u>0$, it is convenient to retain the unoptimized error functional
\begin{equation}\label{eq:app-2mqc-error}
 \mathcal E(u,v,A)
 =\frac{A D(u,v)-(u^2-v^2)^2}{2u^2(1+u)}.
\end{equation}
At fixed $(u,v)$ this is affine and strictly increasing in $A$, and
$\lambda_\delta(u,v)$ is the unique solution of
$\mathcal E(u,v,A)=\delta$.  Eliminating $A$ from the channel
parametrization leaves precisely the strict pair constraint in
\eqref{eq:app-R2MQC}.

\begin{lemma}
\label{lem:app-2mqc-compact}
Let
\begin{equation}\label{eq:app-2mqc-closed}
 \mathcal F_\delta
 =\left\{(u,v):0\le v\le u\le1,
 v^2(u^2-v^2)\le2\delta u^2(1-v^2)\right\}.
\end{equation}
Then $\Psi_\delta$ is continuous on $\mathcal F_\delta$ and
\begin{equation}\label{eq:app-2mqc-min}
 \RtwoMQC(\delta)=\min_{(u,v)\in\mathcal F_\delta}\Psi_\delta(u,v).
\end{equation}
\end{lemma}

\begin{proof}
The denominator in~\eqref{eq:app-2mqc-lambda} vanishes only at the origin.
For $0\le v\le u$ and $(u,v)\to(0,0)$,
\[
 \frac{D(u,v)}{u^2}
 =1-v^2+\frac{v^2}{u^2}(u+v^2)\longrightarrow1,
\]
whereas the numerator divided by $u^2$ tends to $2\delta$.  This proves the
claimed continuous extension.  The closed feasible set is compact.

It remains only to compare strict and closed feasibility.  If $(u,v)$ belongs
to~\eqref{eq:app-2mqc-closed}, $u>0$, and $0<t<1$, then
\begin{align*}
 (tv)^2\bigl((tu)^2-(tv)^2\bigr)
 &=t^4v^2(u^2-v^2)\\
 &\le2\delta t^4u^2(1-v^2)
 <2\delta t^2u^2(1-t^2v^2).
\end{align*}
Thus $(tu,tv)$ is strictly feasible.  The origin is approached by $(t,0)$.
Continuity identifies the strict infimum with the compact minimum.
\end{proof}

Alrabiah and Guruswami prove
\begin{equation}\label{eq:app-2mqc-below-M2}
 \RtwoMQC(\delta)<M_2(\delta)
 \qquad (0<\delta<1/2)
\end{equation}
\cite[Theorem~3]{AlrabiahGuruswami2026}.  We will also use
\begin{equation}\label{eq:app-M2-below-EB}
 M_2(\delta)<R_{\mathrm{EB}}(\delta)
 \defeq 1-g(2\delta).
\end{equation}
Indeed, $F_\delta(0)=R_{\mathrm{EB}}(\delta)$.  Moreover,
$2\tau g'(\tau^2)\to0$ as $\tau\downarrow0$, and hence
$F_\delta'(0^+)=-2\delta g'(2\delta)<0$.  These strict inequalities
rule out the two degenerate lines $v=0$ and $v=u$ for any minimizing point
with $A<1$.

\subsection{Structure of an interior 2MQC minimizer}

The following elementary entropy order will isolate the range in which the
constant-weight bridge is valid.

\begin{lemma}
\label{lem:app-entropy-product}
Let $a,b,c,d\in[0,1]$.  If
\begin{equation*}\label{eq:app-product-order-assumptions}
 \max\{a,b\}\le\max\{c,d\},
 \qquad
 ab\le cd,
\end{equation*}
then
\begin{equation*}\label{eq:app-product-order-conclusion}
 g(a)+g(b)\le g(c)+g(d).
\end{equation*}
\end{lemma}

\begin{proof}
For $x\in(0,1)$,
\[
 g'(x)=\frac{1}{4\sqrt{1-x}}
 \log_2\!\left(\frac{1+\sqrt{1-x}}{1-\sqrt{1-x}}\right).
\]
The function $\phi(x)=xg'(x)$ is strictly increasing.  To see this, put
$r=\sqrt{1-x}$.  Up to the positive factor $1/(2\ln2)$, differentiation
with respect to $r$ gives
\[
 \frac1r-\left(1+\frac1{r^2}\right)\operatorname{artanh}r<0,
\]
because $\operatorname{artanh}r>r$; since $x=1-r^2$ decreases with $r$,
$\phi$ increases with $x$.

By symmetry, assume $a\ge b$ and $c\ge d$.  The first hypothesis then
gives $a\le c$.  For fixed product $p$, define
$Q_p(t)=g(t)+g(p/t)$ on $\sqrt p\le t\le1$.  Then
\[
 Q_p'(t)=\frac{\phi(t)-\phi(p/t)}t\ge0.
\]
Increase the larger coordinate from $a$ to $c$ while preserving the product
$ab$; this does not decrease the entropy sum.  The other coordinate becomes
$ab/c\le d$ by the product assumption.  Increasing it to $d$ again cannot
decrease the sum because $g$ is increasing.  Endpoint cases follow by
continuity.
\end{proof}

\begin{proposition}
\label{prop:app-2mqc-structure}
Let $(u,v)$ minimize~\eqref{eq:app-2mqc-min}, and put
$A=A_\delta(u,v)$.  If $A<1$, then
\begin{equation*}\label{eq:app-2mqc-structure}
 0<v<u<1,
 \qquad
 \mathcal E(u,v,A)=\delta,
 \qquad
 A>u^2,
 \qquad
 A>2\delta.
\end{equation*}
\end{proposition}

\begin{proof}
Since the truncation in~\eqref{eq:app-2mqc-A} is inactive,
$A=\lambda_\delta(u,v)$ and $\mathcal E(u,v,A)=\delta$.
If $v=0$, then $A=u^2+2\delta(1+u)<1$.  Since the left-hand side is
strictly increasing in $u$ and equals $1$ at $u=1-2\delta$, this forces
$u<1-2\delta$.  Hence
$\Psi_\delta(u,0)=F_\delta(u)\ge M_2(\delta)$, contradicting
\eqref{eq:app-2mqc-below-M2}.  If $v=u$, then $A=2\delta$ and
$\Psi_\delta(u,u)=R_{\mathrm{EB}}(\delta)$, contradicting
\eqref{eq:app-2mqc-below-M2}--\eqref{eq:app-M2-below-EB}.  Hence
$0<v<u$.

Suppose $u=1$.  Since $v<u$ has already been proved, one has $v<1$.
Closed feasibility gives $v^2\le2\delta$, while
\[
 \lambda_\delta(1,v)-1
 =\frac{4\delta-2v^2}{1+v^4}\ge0,
\]
contrary to $A<1$.  Thus $u<1$.

Write $\Delta=u^2-v^2$.  Combining $A=\lambda_\delta(u,v)$ with closed
feasibility gives the exact factorization
\begin{equation*}\label{eq:app-product-factor}
 2\delta u^2-Av^2
 =\frac{\Delta\,[2\delta u^2(1-v^2)-v^2\Delta]}{D(u,v)}\ge0.
\end{equation*}
If $A\le u^2$, then
\[
 \max\{A,v^2\}\le\max\{u^2,2\delta\},
 \qquad
 Av^2\le u^2(2\delta).
\]
Lemma~\ref{lem:app-entropy-product}, applied to the pairs
$(A,v^2)$ and $(u^2,2\delta)$, yields
\[
 g(A)+g(v^2)\le g(u^2)+g(2\delta).
\]
Consequently $\Psi_\delta(u,v)\ge R_{\mathrm{EB}}(\delta)$, again a
contradiction.  Therefore $A>u^2$.

Finally, direct simplification of~\eqref{eq:app-2mqc-error} gives
\begin{equation*}\label{eq:app-A-half-factor}
 \frac A2-\delta
 =\frac{\Delta\,[A(u+v^2)+\Delta]}{2u^2(1+u)}>0,
\end{equation*}
which proves $A>2\delta$.
\end{proof}

\subsection{The isospectral bridge and its strict deformation}

The inequalities $A>u^2$ and $A>2\delta$ are precisely what place the
following channel-to-constant-weight map in the strict interior of the
admissible parameter ranges.  Let $(u,v,A)$ be an interior minimizer as in
Proposition~\ref{prop:app-2mqc-structure}, and define
\begin{equation}\label{eq:app-bridge-map}
 \alpha=\Jmap(A),
 \qquad
 w=\Jmap(u^2),
 \qquad
 q=\Jmap(v^2),
 \qquad
 \beta=\alpha q,
 \qquad
 \gamma=(1-\alpha)q.
\end{equation}
Write
\begin{align}
 \Phi_{\mathrm{CW}}(\alpha,\beta,\gamma,w)
 ={}&1-\Htwo(\alpha)+\Htwo(w)
 -\alpha\Htwo\!\left(\frac\beta\alpha\right)\notag\\
 &-(1-\alpha)\Htwo\!\left(\frac\gamma{1-\alpha}\right),
                                                               \label{eq:app-Phi-CW}\\
 \mathcal E_{\mathrm{CW}}(\alpha,\beta,\gamma,w)
 ={}&2\alpha(1-\alpha)
 \bigl(1-\Lambda_{\alpha,\beta,\gamma}(w)\bigr).\notag
\end{align}
Thus the strict constant-weight spectral condition
\eqref{eq:CW-constraint} is exactly
$\mathcal E_{\mathrm{CW}}<\delta$.

\begin{lemma}
\label{lem:app-exact-bridge}
The parameters~\eqref{eq:app-bridge-map} satisfy the constant-weight range
conditions~\eqref{eq:CW-ranges-1}--\eqref{eq:CW-ranges-2} strictly, and
\begin{align}
 \Phi_{\mathrm{CW}}(\alpha,\beta,\gamma,w)
 &=1+g(u^2)-g(A)-g(v^2),                                  \label{eq:app-objective-bridge}\\
 \Lambda_{\alpha,\beta,\gamma}(w)
 &=1-\frac{2\mathcal E(u,v,A)}A.                          \label{eq:app-spectral-bridge}
\end{align}
In particular, the bridge lies on the exact constant-weight spectral
boundary:
\begin{equation}\label{eq:app-CW-boundary-equality}
 \mathcal E_{\mathrm{CW}}(\alpha,\beta,\gamma,w)=\delta.
\end{equation}
\end{lemma}

\begin{proof}
Proposition~\ref{prop:app-2mqc-structure} and the monotonicity of $\Jmap$
give
\[
 0<q<w<\alpha<\frac12.
\]
Hence $0<\beta<\alpha/2$, $0<\gamma<(1-\alpha)/2$, and
$\beta+\gamma=q<w$.  Moreover,
\[
 \alpha-\beta+\gamma=\alpha+q(1-2\alpha)>\alpha>w,
\]
while
\[
 1-\alpha+\beta-\gamma
 =1-\alpha-q(1-2\alpha)>\frac12>w.
\]
Finally, $A>2\delta$ implies
$\alpha>\Jmap(2\delta)>\delta/2$, where the last inequality follows from
$\sqrt{1-2\delta}<1-\delta$.  This proves all range conditions.

Because $\beta/\alpha=\gamma/(1-\alpha)=q$, the objective is
\[
 \Phi_{\mathrm{CW}}=1-\Htwo(\alpha)+\Htwo(w)-\Htwo(q)
 =1-g(A)+g(u^2)-g(v^2),
\]
proving~\eqref{eq:app-objective-bridge}.  In the affine coordinates of
\eqref{eq:CW-affine}, the bridge satisfies
\[
 m=\sqrt{1-A},
 \qquad
 z_{\mathrm{CW}}=\sqrt{1-u^2},
 \qquad
 \zeta=\sqrt{1-v^2},
 \qquad
 \xi=m\zeta.
\]
Substitution into~\eqref{eq:CW-Lambda} gives, with
$\Delta=u^2-v^2$,
\[
 \Lambda_{\alpha,\beta,\gamma}(w)
 =\frac{\Delta[A(u+v^2)+\Delta]}{A u^2(1+u)}
 =1-\frac{2\mathcal E(u,v,A)}A,
\]
which is~\eqref{eq:app-spectral-bridge}.  Since
$A=4\alpha(1-\alpha)$ and $\mathcal E(u,v,A)=\delta$,
\eqref{eq:app-CW-boundary-equality} follows.
\end{proof}

\begin{lemma}
\label{lem:app-strict-CW-deformation}
If a minimizing 2MQC point has $A<1$, then
\begin{equation}\label{eq:app-CW-beats-interior}
 \kappa_{\mathrm{CW}}(\delta)<\RtwoMQC(\delta).
\end{equation}
\end{lemma}

\begin{proof}
Fix the bridge point.  For $c>0$ and small $\eps>0$, put
\begin{align*}
 \alpha_\eps&=\alpha+c\eps,&
 p_\eps&=q+(1-\alpha)\eps,&
 r_\eps&=q-\alpha\eps,\\
 \beta_\eps&=\alpha_\eps p_\eps,&
 \gamma_\eps&=(1-\alpha_\eps)r_\eps,&
 w_\eps&=w.
\end{align*}
All range inequalities remain strict for sufficiently small $\eps$.
First set $c=0$.  Then the affine coordinates satisfy
$\dot\zeta=0$ and $\dot\xi=4\alpha(1-\alpha)=A$.  Differentiating
\eqref{eq:CW-Lambda} at the bridge gives
\begin{equation*}\label{eq:app-Lambda-derivative}
 \left.\frac{d}{d\eps}
 \Lambda_{\alpha_\eps,\beta_\eps,\gamma_\eps}(w)
 \right|_{\eps=0,c=0}
 =\frac{2\sqrt{1-A}\sqrt{1-v^2}(u^2-v^2)}{u^2(1+u)}>0.
\end{equation*}
Therefore
\[
 \left.\frac{d}{d\eps}\mathcal E_{\mathrm{CW}}
 (\alpha_\eps,\beta_\eps,\gamma_\eps,w)\right|_{\eps=0,c=0}<0.
\]
This derivative depends continuously on $c$, so it remains negative for all
sufficiently small fixed $c>0$.

At the same time, direct differentiation of~\eqref{eq:app-Phi-CW} gives
\begin{equation*}\label{eq:app-Phi-derivative}
 \left.\frac{d}{d\eps}
 \Phi_{\mathrm{CW}}(\alpha_\eps,\beta_\eps,\gamma_\eps,w)
 \right|_{\eps=0}
 =-c\Htwo'(\alpha)<0.
\end{equation*}
Indeed, the two conditional-entropy derivatives cancel because their
starting values are both $q$ and their weighted first-order changes sum to
zero.  Choose a small fixed $c>0$ for which the effective-error derivative is
negative, and then choose $\eps>0$ sufficiently small.  The perturbed point
has $\mathcal E_{\mathrm{CW}}<\delta$ and objective strictly below the bridge
value, which equals $\RtwoMQC(\delta)$ by
\eqref{eq:app-objective-bridge}.  Taking the constant-weight infimum proves
\eqref{eq:app-CW-beats-interior}.
\end{proof}

\subsection{The unmasked endpoint and the final comparison}

The preceding deformation uses $A<1$.  It remains to analyze the truncation
endpoint $A=1$, where the masked channel reduces to the unmasked MQC branch.

\begin{lemma}
\label{lem:app-A1-reduction}
Every closed 2MQC-feasible point with $A_\delta(u,v)=1$ satisfies
\begin{equation*}\label{eq:app-A1-lower}
 \Psi_\delta(u,v)\ge\RMQC(\delta).
\end{equation*}
The infimum over the $A=1$ branch is exactly $\RMQC(\delta)$.
\end{lemma}

\begin{proof}
Here $u>0$, and direct simplification gives
\begin{equation*}\label{eq:app-E-A1}
 \mathcal E(u,v,1)=\frac{u-u^2+v^2}{2u}.
\end{equation*}
Since $A_\delta(u,v)=1$ means $\lambda_\delta(u,v)\ge1$ and
$A\mapsto\mathcal E(u,v,A)$ is increasing, one has
$\mathcal E(u,v,1)\le\delta$.  With $\sigma=1-2\delta$, this is
$v^2\le u(u-\sigma)$ and in particular $u\ge\sigma$.  Since $g$ is
increasing and $g(1)=1$,
\[
 \Psi_\delta(u,v)
 =g(u^2)-g(v^2)
 \ge g(u^2)-g(u(u-\sigma))
 \ge\RMQC(\delta).
\]

Conversely, for any $u\in[\sigma,1]$, set $v^2=u(u-\sigma)$.  Then
$\mathcal E(u,v,1)=\delta$, so $\lambda_\delta(u,v)=1$, and the closed pair
constraint follows from
\begin{align*}
 &2\delta u^2(1-v^2)-v^2(u^2-v^2)\\
 &\qquad={}
 u^2(1-u)\bigl[\sigma^2+(1-\sigma)(1+u)\bigr]\ge0.
\end{align*}
The displayed slack is strictly positive for $u<1$.  Hence every point of
this curve except its endpoint $u=1$ is strictly pair-feasible and has
$A_\delta=1$; the endpoint is approached along the same curve as
$u\uparrow1$.  The objective is the integrand in~\eqref{eq:app-RMQC}, so the
strict-feasibility infimum on the $A=1$ branch is exactly $\RMQC(\delta)$.
\end{proof}

\begin{theorem}
\label{thm:channel-comparison}
For every $0<\delta<1/2$,
\begin{align}
 \kappa_{\mathrm{bin}}(\delta)
 &\le\RtwoMQC(\delta),                                      \label{eq:app-openai-v-2mqc}\\
 \kappa_{\mathrm{pair}}(\delta)
 &=\min\{\kappa_{\mathrm{CW}}(\delta),\kappa_{\mathrm{bal}}(\delta)\}
 <\RtwoMQC(\delta).                                         \label{eq:app-main-2mqc-comparison}
\end{align}
Consequently,
\begin{equation}\label{eq:app-main-2mqc-chain}
 R_2(\delta)
 \le\kappa_{\mathrm{best}}(\delta)
 \le\kappa_{\mathrm{explicit}}(\delta)
 \le\min\{\kappa_{\mathrm{bin}}(\delta),
            \kappa_{\mathrm{pair}}(\delta)\}
 <\RtwoMQC(\delta)<M_2(\delta).
\end{equation}
\end{theorem}

\begin{proof}
Choose a minimizer of the compact problem~\eqref{eq:app-2mqc-min} and put
$A=A_\delta(u,v)$.

\smallskip
\noindent\emph{Case 1: $A<1$.}
Lemma~\ref{lem:app-strict-CW-deformation} gives
$\kappa_{\mathrm{CW}}<\RtwoMQC$.  Since both
$\kappa_{\mathrm{bin}}$ and $\kappa_{\mathrm{pair}}$ are at most
$\kappa_{\mathrm{CW}}$, both desired comparisons are strict in this case.

\smallskip
\noindent\emph{Case 2: $A=1$.}
Lemma~\ref{lem:app-A1-reduction} gives $\RMQC\le\RtwoMQC$ at the minimizing
value.  Proposition~\ref{prop:app-MQC-equivalence} and
Theorem~\ref{thm:balanced-strict} imply
\[
 \kappa_{\mathrm{bin}}\le\kappa_H=\RMQC\le\RtwoMQC,
 \qquad
 \kappa_{\mathrm{pair}}\le\kappa_{\mathrm{bal}}
 <\kappa_H=\RMQC\le\RtwoMQC.
\]
Thus \eqref{eq:app-openai-v-2mqc} holds in both cases and
\eqref{eq:app-main-2mqc-comparison} is always strict.  The remaining
inequalities in~\eqref{eq:app-main-2mqc-chain} follow from the definitions and
from~\eqref{eq:app-2mqc-below-M2}.
\end{proof}

\begin{remark}
The theorem is pointwise on $0<\delta<1/2$ and does not assert a positive gap
uniform in $\delta$.  The two branches of the proof use different mechanisms:
a genuinely masked minimizer is deformed inside the constant-weight family,
whereas the unmasked endpoint is improved by honeycomb symmetry breaking.
Compactification makes this dichotomy exhaustive and also handles minimizers
that lie only on the closure of the original strict feasible set.
\end{remark}


\begin{thebibliography}
\small

\bibitem[ABG+26]{AlonEtAlUnitDistance2026}
N.~Alon, T.~F. Bloom, W.~T. Gowers, D.~Litt, W.~Sawin, A.~Shankar,
J.~Tsimerman, V.~Wang, and M.~Matchett Wood,
\emph{Remarks on the disproof of the unit distance conjecture},
preprint, arXiv:2605.20695, 2026.

\bibitem[AG26]{AlrabiahGuruswami2026}
O.~Alrabiah and V.~Guruswami,
\emph{Binary code rate bounds via classical--quantum channels},
preprint, arXiv:2608.09347, 2026.

\bibitem[BV08]{BachocVallentin2008}
C.~Bachoc and F.~Vallentin,
New upper bounds for kissing numbers from semidefinite programming,
\emph{Journal of the American Mathematical Society} \textbf{21}
(2008), no.~3, 909--924.

\bibitem[Bae26]{Bae2026}
J.~H. Bae,
\emph{Quantum query complexity of the hyperoctahedral group},
arXiv:2604.13554, 2026.

\bibitem[BN06]{BargNogin2006}
A.~M. Barg and D.~Yu. Nogin,
Spectral approach to linear programming bounds on codes,
\emph{Problems of Information Transmission} \textbf{42} (2006), no.~2,
77--89.

\bibitem[BN08]{BargNogin2008}
A.~Barg and D.~Nogin,
A functional view of upper bounds on codes,
in \emph{Coding and Cryptology}, Series on Coding Theory and Cryptology 4,
World Scientific, 2008, 15--24.

\bibitem[CD25]{ChaillouxDebrisAlazard2025}
A.~Chailloux and T.~Debris-Alazard,
New solutions to Delsarte's dual linear programs,
\emph{IEEE Transactions on Information Theory} \textbf{71} (2025), no.~1,
297--316.

\bibitem[Col42]{Collatz1942}
L.~Collatz,
Einschlie{\ss}ungssatz f\"ur die charakteristischen Zahlen von Matrizen,
\emph{Mathematische Zeitschrift} \textbf{48} (1942), 221--226,
\url{https://doi.org/10.1007/BF01180013}.

\bibitem[CJJ22]{CoreglianoJeronimoJones2022}
L.~N. Coregliano, F.~G. Jeronimo, and C.~Jones,
A complete linear programming hierarchy for linear codes,
in \emph{13th Innovations in Theoretical Computer Science Conference
(ITCS 2022)}, LIPIcs 215, 51:1--51:22, 2022;
arXiv:2112.09221.

\bibitem[CJJ23]{CoreglianoJeronimoJones2023}
L.~N. Coregliano, F.~G. Jeronimo, and C.~Jones,
Exact completeness of LP hierarchies for linear codes,
in \emph{14th Innovations in Theoretical Computer Science Conference
(ITCS 2023)}, LIPIcs 251, 40:1--40:18, 2023;
arXiv:2211.01248.

\bibitem[CJJ+26]{CoreglianoJeronimoJonesLinialLoyfer2026}
L.~N. Coregliano, F.~G. Jeronimo, C.~Jones, N.~Linial, and E.~Loyfer,
Higher-order Delsarte dual LPs: lifting, constructions and completeness,
in \emph{17th Innovations in Theoretical Computer Science Conference
(ITCS 2026)}, LIPIcs 362, 44:1--44:22, 2026; arXiv:2501.04854.

\bibitem[dLV15]{deLaatVallentin2015}
D.~de Laat and F.~Vallentin,
A semidefinite programming hierarchy for packing problems in discrete
geometry,
\emph{Mathematical Programming} \textbf{151} (2015), no.~2, 529--553.

\bibitem[Del73]{Delsarte1973}
P.~Delsarte,
\emph{An Algebraic Approach to the Association Schemes of Coding Theory},
Philips Research Reports Supplements, no.~10, 1973.

\bibitem[DI89]{DoeraeneIommi1989}
J.-P.~Doeraene and G.~Iommi Amun\'ategui,
Branching rules for the hyperoctahedral group,
\emph{Journal of Mathematical Physics} \textbf{30} (1989), no.~11,
2469--2475.

\bibitem[FT05]{FriedmanTillich2005}
J.~Friedman and J.-P.~Tillich,
Generalized Alon--Boppana theorems and error-correcting codes,
\emph{SIAM Journal on Discrete Mathematics} \textbf{19} (2005), no.~3,
700--718.

\bibitem[F{\o}l55]{Folner1955}
E.~F{\o}lner,
On groups with full Banach mean value,
\emph{Mathematica Scandinavica} \textbf{3} (1955), 243--254,
\url{https://doi.org/10.7146/math.scand.a-10442}.

\bibitem[GMS12]{GijswijtMittelmannSchrijver2012}
D.~C. Gijswijt, H.~D. Mittelmann, and A.~Schrijver,
Semidefinite code bounds based on quadruple distances,
\emph{IEEE Transactions on Information Theory} \textbf{58} (2012), no.~5,
2697--2705.

\bibitem[Gil52]{Gilbert1952}
E.~N. Gilbert,
A comparison of signalling alphabets,
\emph{Bell System Technical Journal} \textbf{31} (1952), 504--522.

\bibitem[Gre22]{Green2022}
R.~Green,
Specht module branching rules for wreath products of symmetric groups,
\emph{Algebraic Combinatorics} \textbf{5} (2022), no.~4, 609--628.

\bibitem[HW94]{HausladenWootters1994}
P.~Hausladen and W.~K. Wootters,
A ``pretty good'' measurement for distinguishing quantum states,
\emph{Journal of Modern Optics} \textbf{41} (1994), no.~12, 2385--2390,
\url{https://doi.org/10.1080/09500349414552221}.

\bibitem[Hol73]{Holevo1973}
A.~S. Holevo,
Bounds for the quantity of information transmitted by a quantum communication
channel,
\emph{Problemy Peredachi Informatsii} \textbf{9} (1973), no.~3, 3--11;
English translation in \emph{Problems of Information Transmission}
\textbf{9} (1973), no.~3, 177--183.

\bibitem[Hor62]{Horn1962}
A.~Horn,
Eigenvalues of sums of Hermitian matrices,
\emph{Pacific Journal of Mathematics} \textbf{12} (1962), no.~1, 225--241,
\url{https://doi.org/10.2140/pjm.1962.12.225}.

\bibitem[JK81]{JamesKerber1981}
G.~James and A.~Kerber,
\emph{The Representation Theory of the Symmetric Group},
Encyclopedia of Mathematics and its Applications 16,
Addison--Wesley, 1981.

\bibitem[Kly98]{Klyachko1998}
A.~A. Klyachko,
Stable bundles, representation theory and Hermitian operators,
\emph{Selecta Mathematica (N.S.)} \textbf{4} (1998), no.~3, 419--445,
\url{https://doi.org/10.1007/s000290050037}.

\bibitem[KT99]{KnutsonTao1999}
A.~Knutson and T.~Tao,
The honeycomb model of $GL_n(\C)$ tensor products I: proof of the saturation
conjecture,
\emph{Journal of the American Mathematical Society} \textbf{12} (1999),
no.~4, 1055--1090.

\bibitem[Lau03]{Laurent2003}
M.~Laurent,
A comparison of the Sherali--Adams, Lov\'asz--Schrijver and Lasserre
relaxations for $0$--$1$ programming,
\emph{Mathematics of Operations Research} \textbf{28} (2003), no.~3,
470--496.

\bibitem[Lau07]{Laurent2007}
M.~Laurent,
Strengthened semidefinite programming bounds for codes,
\emph{Mathematical Programming} \textbf{109} (2007), no.~2--3, 239--261.

\bibitem[LL26]{LinialLoyfer2023Elementary}
N.~Linial and E.~Loyfer,
An elementary proof of the first LP bound on the rate of binary codes,
preprint, arXiv:2303.16619, 2023.

\bibitem[LL23]{LoyferLinial2023}
E.~Loyfer and N.~Linial,
New LP-based upper bounds in the rate-vs.-distance problem for binary linear
codes,
\emph{IEEE Transactions on Information Theory} \textbf{69} (2023), no.~5,
2886--2899; arXiv:2206.09211.

\bibitem[Mac63]{MacWilliams1963}
F.~J. MacWilliams,
A theorem on the distribution of weights in a systematic code,
\emph{Bell System Technical Journal} \textbf{42} (1963), no.~1, 79--94.

\bibitem[MRRW77]{MRRW1977}
R.~J. McEliece, E.~R. Rodemich, H.~Rumsey, Jr., and L.~R. Welch,
New upper bounds on the rate of a code via the Delsarte--MacWilliams
inequalities,
\emph{IEEE Transactions on Information Theory} \textbf{23} (1977), no.~2,
157--166.

\bibitem[NS05]{NavonSamorodnitsky2005}
M.~Navon and A.~Samorodnitsky,
On Delsarte's linear programming bounds for binary codes,
in \emph{46th Annual IEEE Symposium on Foundations of Computer Science
(FOCS 2005)}, 327--336, 2005.

\bibitem[NS09]{NavonSamorodnitsky2009}
M.~Navon and A.~Samorodnitsky,
Linear programming bounds for codes via a covering argument,
\emph{Discrete \& Computational Geometry} \textbf{41} (2009), no.~2,
199--207.

\bibitem[OAI26]{OpenAI2026}
OpenAI,
\emph{Ten Advances in Mathematics and Theoretical Computer Science},
Chapter~2: ``Improved Bounds for Binary and Spherical Codes,'' updated
August~6, 2026,
\url{https://cdn.openai.com/pdf/ten-proofs-oai.pdf}.

\bibitem[Sam01]{Samorodnitsky2001}
A.~Samorodnitsky,
On the optimum of Delsarte's linear program,
\emph{Journal of Combinatorial Theory, Series A} \textbf{96} (2001), no.~2,
261--287.

\bibitem[Sam23]{Samorodnitsky2023}
A.~Samorodnitsky,
One more proof of the first linear programming bound for binary codes and two
conjectures,
\emph{Israel Journal of Mathematics} \textbf{256} (2023), no.~2, 639--673.

\bibitem[Sam25]{Samorodnitsky2025}
A.~Samorodnitsky,
On the difficulty to beat the first linear programming bound for binary
codes,
\emph{IEEE Transactions on Information Theory} \textbf{71} (2025), no.~4,
2383--2388; arXiv:2308.16038.

\bibitem[Sch79]{Schrijver1979}
A.~Schrijver,
A comparison of the Delsarte and Lov\'asz bounds,
\emph{IEEE Transactions on Information Theory} \textbf{25} (1979), no.~4,
425--429.

\bibitem[Sch05]{Schrijver2005}
A.~Schrijver,
New code upper bounds from the Terwilliger algebra and semidefinite
programming,
\emph{IEEE Transactions on Information Theory} \textbf{51} (2005), no.~8,
2859--2866.

\bibitem[Ste17]{Stein2017}
I.~Stein,
The Littlewood--Richardson rule for wreath products with symmetric groups
and the quiver of the category $F\wr\mathbf{FI}_n$,
\emph{Communications in Algebra} \textbf{45} (2017), no.~5, 2105--2126.

\bibitem[VMK88]{Varshalovich1988}
D.~A. Varshalovich, A.~N. Moskalev, and V.~K. Khersonskii,
\emph{Quantum Theory of Angular Momentum},
World Scientific, 1988.

\bibitem[Wie50]{Wielandt1950}
H.~Wielandt,
Unzerlegbare, nicht negative Matrizen,
\emph{Mathematische Zeitschrift} \textbf{52} (1950), 642--648,
\url{https://doi.org/10.1007/BF02230720}.

\bibitem[Var57]{Varshamov1957}
R.~R. Varshamov,
Estimate of the number of signals in error correcting codes,
\emph{Doklady Akademii Nauk SSSR} \textbf{117} (1957), 739--741.

\end{thebibliography}
\end{document}